\documentclass[11pt]{article}
\usepackage[utf8]{inputenc}

\usepackage{microtype}
\usepackage{amsmath}
\usepackage{amssymb}
\usepackage{amsthm}
\usepackage{bm}
\usepackage{dsfont}
\usepackage{wrapfig}
\usepackage{authblk}
\usepackage{fullpage}
\usepackage{comment}
\usepackage{nicefrac}
\usepackage{tikz}
\usepackage[ruled,vlined,linesnumbered,noend]{algorithm2e}
\usetikzlibrary{positioning}
\usepackage{mathtools}
\usepackage{bbm}
\usepackage{complexity}
\usepackage{float}
\usepackage{array,booktabs}
\usepackage[
    backend=biber,
    style=alphabetic,
    maxcitenames=10,
    maxbibnames=99,
    natbib=true,
    url=false,
    doi=false, 
    isbn=false]{biblatex}
\usepackage[linktocpage=true]{hyperref}
\usepackage{mleftright}
\RestyleAlgo{boxruled}
\usepackage{thm-restate}
\usepackage[capitalize,noabbrev]{cleveref}
\usepackage[shortlabels]{enumitem}

\newcommand{\declarecolor}[2]{\definecolor{#1}{RGB}{#2}\expandafter\newcommand\csname #1\endcsname[1]{\textcolor{#1}{##1}}}
\declarecolor{White}{255, 255, 255}
\declarecolor{Black}{0, 0, 0}
\declarecolor{Maroon}{128, 0, 0}
\declarecolor{Coral}{255, 127, 80}
\declarecolor{Red}{236, 67, 84}
\declarecolor{Orange}{230, 142, 64}
\declarecolor{LightBlue}{40, 150, 220}
\declarecolor{LimeGreen}{50, 205, 50}
\declarecolor{DarkGreen}{0, 100, 0}
\declarecolor{Navy}{0, 0, 128}

\hypersetup{
     colorlinks = true,
     linkcolor = brown,
     anchorcolor = blue,
     citecolor = teal,
     filecolor = blue,
     urlcolor = black
     }

\theoremstyle{plain}
\newtheorem{theorem}{Theorem}[section]
\newtheorem{lemma}[theorem]{Lemma}
\newtheorem{corollary}[theorem]{Corollary}
\newtheorem{proposition}[theorem]{Proposition}
\newtheorem{fact}[theorem]{Fact}

\newtheorem{observation}[theorem]{Observation}
\newtheorem{claim}[theorem]{Claim}

\newtheorem{property}[theorem]{Property}
\newtheorem{remark}[theorem]{Remark}

\theoremstyle{definition}
\newtheorem{definition}[theorem]{Definition}

\DeclareMathOperator{\ppad}{\ComplexityFont{PPAD}}

\newcommand{\nakedcite}[1]{\citeauthor{#1}, \citeyear{#1}}

\newcommand*{\N}{{\mathbb{N}}}

\let\R\relax
\newcommand*{\R}{{\mathbb{R}}}

\let\C\relax
\newcommand*{\C}{{\mathbb{C}}}

\DeclareMathOperator{\gap}{gap}

\let\cL\relax

\newcommand*{\cX}{{\mathcal{X}}}
\newcommand*{\cA}{{\mathcal{A}}}
\newcommand*{\cY}{{\mathcal{Y}}}
\newcommand*{\cL}{{\mathcal{L}}}
\newcommand*{\cN}{{\mathcal{N}}}
\newcommand*{\cR}{{\mathcal{R}}}

\newcommand{\defeq}{\coloneqq}

\DeclareMathOperator*{\argmax}{arg\,max}
\DeclareMathOperator*{\argmin}{arg\,min}

\let\triangleq\defeq

\renewcommand{\vec}[1]{\bm{#1}}
\newcommand{\mat}[1]{\mathbf{#1}}

\let\E\relax

\DeclareMathOperator{\ball}{\mathcal{B}}
\DeclareMathOperator{\reg}{Reg}
\DeclareMathOperator{\phimax}{\Phi_{max}}

\DeclareMathOperator{\sw}{\textsc{SW}}
\newcommand{\rvu}{\texttt{RVU}\xspace}
\DeclareMathOperator{\opt}{\textsc{OPT}}

\DeclareMathOperator{\E}{\mathbb{E}}

\bibliography{main}

\title{On Last-Iterate Convergence Beyond Zero-Sum Games}

\author[1]{Ioannis Anagnostides}
\author[2]{Ioannis Panageas}
\author[3]{Gabriele Farina}
\author[4]{Tuomas Sandholm}

\affil[1,3,4]{Carnegie Mellon University}
\affil[2]{University of California Irvine}
\affil[4]{Strategy Robot, Inc.}
\affil[4]{Optimized Markets, Inc.}
\affil[4]{Strategic Machine, Inc.}
\affil[ ]{\texttt {\{ianagnos,gfarina,sandholm\}@cs.cmu.edu}, and \texttt{ipanagea@ics.uci.edu}}

\date{}

\begin{document}

\maketitle

\pagenumbering{gobble}
\begin{abstract}
        Most existing results about \emph{last-iterate convergence} of learning dynamics are limited to two-player zero-sum games, and only apply under rigid assumptions about what dynamics the players follow.
    In this paper we provide new results and techniques that apply to broader families of games and learning dynamics.
    %
    First, we use a regret-based analysis to show that in a class of games that includes constant-sum polymatrix and strategically zero-sum games, dynamics such as \emph{optimistic mirror descent (OMD)} have \emph{bounded second-order path lengths}, a property which holds even when players employ different algorithms and prediction mechanisms. This enables us to obtain $O(1/\sqrt{T})$ rates and optimal $O(1)$ regret bounds.
    Our analysis also reveals a surprising property: 
    OMD either reaches arbitrarily close to a Nash equilibrium, or it outperforms the \emph{robust price of anarchy} in efficiency.
    Moreover, for potential games we establish convergence to an $\epsilon$-equilibrium after $O(1/\epsilon^2)$ iterations for mirror descent under a broad class of regularizers, as well as optimal $O(1)$ regret bounds for OMD variants. Our framework also extends to near-potential games, and unifies known analyses for distributed learning in Fisher's market model. Finally, we analyze the convergence, efficiency, and robustness of \emph{optimistic gradient descent (OGD)} in general-sum continuous games.
\end{abstract}
\clearpage

  \tableofcontents
\clearpage

\pagenumbering{arabic}

\section{Introduction}

No-regret learning and game theory share an intricately connected history tracing back to Blackwell's seminal \emph{approachability theorem} \citep{Blackwell56:An,Abernethy11:Blackwell}, leading to fundamental connections between no-regret learning and game-theoretic equilibrium concepts. For example, it is folklore that if both players in a zero-sum game employ a no-regret algorithm, the \emph{time average} of their strategies will eventually converge to a \emph{Nash equilibrium (NE)}. However, typical guarantees within the no-regret framework provide no insights about the \emph{final} state of the system. This begs the question: \emph{Will the agents eventually play according to an equilibrium strategy?} In general, the answer to this question is \emph{no}: Broad families of regret minimization algorithms such as \emph{mirror descent (MD)} are known to exhibit recurrent or even chaotic behavior~\citep{Yuzuru02:Chaos, Sandholm10:Population, Mertikopoulos18:Cycles}.

In an attempt to stabilize the behavior of traditional no-regret learning algorithms and ameliorate the notoriously tedious training process of \emph{generative adversarial networks (GANs)}~\citep{Goodfellow14:Generative}, \citet{Daskalakis18:Training} discovered that \emph{optimistic gradient descent (OGD)}~\citep{Popov80:A} guarantees \emph{last-iterate} convergence in unconstrained bilinear ``games''.\footnote{We call them ``games'' with an abuse of terminology. They are not games in the game-theoretic sense.} Thereafter, their result has been substantially extended along several lines (\emph{e.g.}, \citep{Daskalakis19:Last,Mertikopoulos19:Optimistic, Wei21:Linear,Golowich20:Tight,Azizian21:The}). Last-iterate convergence is also central in economics~\citep{Milgrom90:Rationalizability}, and goes back to the fundamental question of what it really means to ``learn in a game''. Indeed, it is unclear how a time average guarantee is meaningful from an economic standpoint. Nevertheless, the known last-iterate guarantees only apply for restricted classes of games such as two-player zero-sum games. As a result, it is natural to ask whether last-iterate convergence could be a \emph{universal} phenomenon in games.


Unfortunately, this cannot be the case: if every player employs a no-regret algorithm and the individual strategies converge, this would necessarily mean that the limit point is a Nash equilibrium.\footnote{To see this, observe that the average product distribution of play---which is a \emph{coarse correlated equilibrium}---will converge to a product distribution, and hence, a Nash equilibrium.} But this is precluded by fundamental impossibility results \citep{Hart03:Uncoupled,Rubinstein16:Settling,Chen09:Settling,Babichenko14:Query,Daskalakis09:The}, at least in a polynomial number of iterations. This observation offers an additional crucial motivation for convergence since Nash equilibria states can be exponentially more efficient compared to \emph{correlated equilibria}~\citep{Blum08:Regret}. In light of this, we ask the following question: \emph{For which classes of games beyond two-player zero-sum games can we guarantee last-iterate convergence?} Characterizing dynamics beyond zero-sum games is a recognized important open question in the literature~\citep{Candogan13:Dynamics}. In fact, even in GANs many practitioners employ a \emph{general-sum} objective due to its practical superiority. However, many existing techniques are tailored to the min-max structure of the problem.

Another drawback of existing last-iterate guarantees is that they make restrictive assumptions about the dynamics the players follow; namely, they employ exactly the \emph{same learning algorithm}. However, we argue that in many settings such a premise is unrealistic under independent and decentralized players. Indeed, to quote from the work of \citet{Candogan13:Dynamics}: \emph{``The limit behavior of dynamic processes where players adhere to different update rules is [...] an open question, even for potential games''}. In this paper, we make progress towards addressing these fundamental questions.

\subsection{Overview of our Contributions} \label{subsection:contributions} The existing techniques employed to show last-iterate convergence are inherently different from the ones used to analyze regret. Indeed, it is tacitly accepted that these two require a different treatment; this viewpoint is reflected by \citet{Mertikopoulos18:Cycles}: \emph{``a regret-based analysis cannot distinguish between a self-stabilizing system, and one with recurrent cycles''.} 

Our first contribution is to challenge this conventional narrative. We show that for a broad class of games the \emph{regret bounded by variation in utilities} (\rvu) property established by \citet{Syrgkanis15:Fast} implies that the second-order path length is bounded (\Cref{theorem:bounded_traj}) for a broad family of dynamics including \emph{optimistic mirror descent (OMD)}~\citep{Rakhlin13:Online}. We use this bound to obtain optimal $O(1)$ individual regret bounds (\Cref{corollary:optreg}), and, more importantly, to show that $O(1/\epsilon^2)$ iterations suffice to reach an $\epsilon$-approximate Nash equilibrium (\Cref{theorem:rate-OMD}). This characterization holds for a class of games intricately connected to the \emph{minimax theorem}, and includes \emph{constant-sum polymatrix}~\citep{Kearns01:Graphical} and \emph{strategically zero-sum} games~\citep{Moulin78:Strategically}. 

Furthermore, our results hold even if players employ \emph{different OMD variants} (with smooth regularizers) and even with \emph{advanced predictions}, as long as an appropriate \rvu bound holds. Additionally, our techniques apply under arbitrary (convex and compact) constrained sets, thereby allowing for a direct extension, for example, to \emph{extensive-form games}. We also illustrate how our framework can be applied to smooth min-max optimization. 
Such results appear hard to obtain with prior techniques. Also, in cases where prior techniques applied, our proofs are considerably simpler. Overall, our framework inherits all of the robustness and the simplicity deriving from a regret-based analysis.

Our approach also reveals an intriguing additional result: OMD variants 
either converge arbitrarily close to a Nash equilibrium, or they outperform the \emph{robust price of anarchy} in terms of efficiency~\citep{Roughgarden15:Intrinsic} (see \Cref{theorem:smooth} for a formal statement). This is based on the observation that lack of last-iterate convergence can be leveraged to show that the sum of the players' regrets at a sufficiently large time $T$ will be at most $- C T$, for some parameter $C > 0$. This substantially refines the result of \citet{Syrgkanis15:Fast}---which showed that the sum of the players' regrets is always $O(1)$---using the rather underappreciated fact that \emph{regrets can be negative}. In fact, the further the dynamics are from a Nash equilibrium, the larger is the improvement compared to the robust price of anarchy.

Next, we study the convergence of \emph{mirror descent (MD)} learning dynamics in (weighted) potential games. We give a new potential argument applicable even if players employ different regularizers (\Cref{theorem:monotone}), thereby addressing an open question of \citet{Candogan13:Dynamics} regarding heterogeneous dynamics in potential games. Such results for no-regret algorithms were known when all players employed variants of \emph{multiplicative weights update}~\citep{Palaiopanos17:Multiplicative,HeliouCM17}, though in \citep{HeliouCM17} the analysis requires vanishing learning rates. 

Our potential argument implies a boundedness property for the trajectories, (similarly to \Cref{theorem:bounded_traj} we previously discussed); we use this property to show that $O(1/\epsilon^2)$ iterations suffice to reach an $\epsilon$-Nash equilibrium. We also show a similar boundedness property for \emph{optimistic} variants, which is then used---along with the \rvu property---to show optimal $O(1)$ regret for each individual player (\Cref{theorem:optregpot}). To our knowledge, this is the first result showing optimal $O(1)$ individual regret in potential games, and it is based on the observation that last-iterate convergence can be leveraged to show improved regret bounds. In a sense, this is the ``converse'' of the technique employed in \Cref{section:optimistic}. Our potential argument also extends to \emph{near-potential games}~\citep{Candogan13:Dynamics}, implying convergence to approximate Nash equilibria (\Cref{theorem:nearpot}). Importantly, our framework is general enough to unify results from distributed learning in Fisher's market model as well~\citep{Birnbaum11:Distributed}.

Finally, motivated by applications such as GANs, we study the convergence properties of \emph{optimistic gradient descent (OGD)} in \emph{continuous games}. We characterize the class of two-player games for which OGD converges, extending the known prior results (\Cref{theorem:two_player-convergence}). 
Also, we show that OGD can be arbitrarily inefficient beyond zero-sum games (\Cref{proposition:inefficiency}). More precisely, under OGD, players can fail to coordinate even if the objective presents a ``clear'' coordination aspect. Finally, in \Cref{theorem:one-many-convergence} we use our techniques to characterize convergence in multiplayer settings as well. 

\subsection{Further Related Work}



The related literature on the subject is too vast to cover in its entirety; below we emphasize certain key contributions.

For \emph{constrained} zero-sum games \citet{Daskalakis19:Last} established asymptotic last-iterate convergence for \emph{optimistic multiplicative weights update}, although they had to assume the existence of a \emph{unique} Nash equilibrium. Our techniques do not require a uniqueness assumption, which is crucial for \emph{equilibrium refinements}~\citep{VanDamme87:Stability}. \citet{Wei21:Linear} improved several aspects of the result of \citet{Daskalakis19:Last}, showing a surprising linear rate of convergence, albeit with a dependence on condition number-like quantities. Some of the latter results were extended for the substantially more complex class of extensive-form games by \citet{Lee2021:Last}. While the aforementioned results apply under a \emph{time-invariant} learning rate, there has also been a substantial amount of work considering a \emph{vanishing} scheduling (\emph{e.g.}, \citep{Mertikopoulos19:Optimistic, Mertikopoulos19:Learning,Zhou17:Mirror,Zhou18:Learning,Hsieh21:Adaptive}). Our results apply under a constant learning rate, which has been extensively argued to be desirable in the literature~\citep{Bailey19:Fast, Golowich20:Tight}. Last-iterate convergence has also been recently documented in certain settings from \emph{reinforcement learning}~\citep{Daskalakis20:Independent,Wei21:Last, Leonardos21}.

Beyond zero-sum games, our knowledge remains much more limited. \citet{Cheung20:Chaos,Cheung21:Chaos} made progress by establishing chaotic behavior for instances of OMD in \emph{bimatrix games}. Last-iterate convergence under \emph{follow the regularizer leader (FTRL)} is known to be rather elusive~\citep{Vlatakis-Gkaragkounis20:No}, occurring only under \emph{strict Nash equilibria}~\citep{Giannou21:Survival}. On the positive side, in \citep{Hsieh21:Adaptive}, last-iterate convergence for a variant of OMD with adaptive learning rates was established for \emph{variationally stable games}, a class of games that includes zero-sum convex-concave.

In the unconstrained setting, the behavior of OGD is by now well-understood in bilinear zero-sum games~\citep{Liang19:Interaction,Mokhtari20:A,Zhang20:Convergence}. Various results have also been shown for \emph{convex-concave} landscapes (\emph{e.g.}, \citep{Mertikopoulos19:Optimistic}), and \emph{ cocoercive games} (\emph{e.g.}, \citep{LinZMJ20}), but covering this literature goes beyond our scope. For \emph{monotone variational inequalities (VIs)} problems~\citep{Harker:Finite}, \citet{Golowich20:Tight} showed last-iterate rates of $O(1/\sqrt{T})$, which is also tight for their considered setting. Note that this rate is slower than that of the time average for the \emph{extra-gradient (EG)} method~\citep{Golowich20:Last}. Lastly, last-iterates rates for variants of OMD have been obtained by \citet{Azizian21:The} in a certain stochastic VI setup.


\section{Preliminaries}
\label{section:prel}

\paragraph{Conventions} In the main body we use the $O(\cdot)$ notation to suppress game-dependent parameters \emph{polynomial} in the game; precise statements are given in the Appendix. We use subscripts to indicate players, and superscripts for time indices. The $j$-th coordinate of $\Vec{x} \in \R^d$ is denoted by $\Vec{x}(j)$.

\paragraph{Normal-Form Games} In this paper we primarily focus on \emph{normal-form games (NFGs)}. We let $[n] \defeq \{1, \dots, n\}$ be the set of players. Each player $i \in [n]$ has a set of available \emph{actions} $\cA_i$. The \emph{joint action profile} is denoted by $\vec{a} \defeq (a_1, \dots, a_n) \in \prod_{i \in [n]} \cA_i$. For a given $\vec{a}$ each player $i$ receives some (normalized) \emph{utility} $u_i(\vec{a})$, where $u_i : \prod_{i \in [n]} \cA_i \to [-1,1]$. Players are allowed to randomize, and we denote by $\vec{x}_i(a_i)$ the probability that player $i$ will chose action $a_i \in \cA_i$, so that $\vec{x}_i \in \Delta(\cA_i) \triangleq \left\{ \Vec{x} \in \R^{|\cA_i|}_{\geq 0} : \sum_{a_i \in \cA_i} \Vec{x}(a_i) = 1 \right\}$. The \emph{expected utility} of player $i$ can be expressed as $\E_{\vec{a} \sim \vec{x}}[u_i(\vec{a})]$, where $\vec{x} \defeq (\vec{x}_1, \dots, \vec{x}_n)$ is the \emph{joint (mixed) strategy profile}. The \emph{social welfare} for a given $\vec{a}$ is defined as the sum of the players' utilities: $\sw(\vec{a}) \defeq \sum_{i = 1}^n u_i(\vec{a})$. As usual, we overload notation to denote $\sw(\vec{x}) \defeq \E_{\vec{a} \sim \vec{x}} \sw(\vec{a})$. We also let $\opt \triangleq \max_{\vec{a} \in \cA} \sw(\vec{a})$ be the maximum social welfare.

\begin{definition}[Approximate Nash equilibrium]
    \label{def:Nash}
    A joint strategy profile $(\vec{x}^*_1, \dots, \vec{x}^*_n) \in \prod_{i \in [n]} \Delta(\cA_i)$ is an \emph{$\epsilon$-approximate (mixed) Nash equilibrium} if for any player $i \in [n]$ and any unilateral deviation $\vec{x}_i \in \Delta(\cA_i)$,
    \begin{equation*}
        u_i(\vec{x}_i, \vec{x}^*_{-i}) \leq u_i(\vec{x}^*) + \epsilon.
    \end{equation*}
\end{definition}

\paragraph{Regret} Let $\cX \subseteq \R^d$ be a convex and compact set. A \emph{regret minimization algorithm} produces at every time $t \in \N$ a strategy $\vec{x}^{(t)} \in \cX$, and receives as feedback from the environment a (linear) utility function $u^{(t)}(\vec{x}) : \vec{x} \mapsto \langle \vec{x}, \vec{u}^{(t)} \rangle$. The \emph{cumulative regret} (or simply regret) $\reg^T$ up to time $T$ measures the performance of the regret minimization algorithm compared to the optimal \emph{fixed strategy in hindsight}:
\begin{equation*}
    \reg^T \triangleq \max_{\vec{x}^* \in \cX} \left\{ \sum_{t=1}^T \langle \vec{x}^*, \vec{u}^{(t)} \rangle \right\} - \sum_{t=1}^T \langle \Vec{x}^{(t)}, \vec{u}^{(t)} \rangle.
\end{equation*}

\paragraph{Optimistic Dynamics} Following a recent trend in online learning \citep{Rakhlin13:Online}, we study \emph{optimistic} (or \emph{predictive}) algorithms. Let $\mathcal{R}$ be a $1$-strongly convex \emph{regularizer} (DGF) with respect to some norm $\| \cdot \|$, continuously differentiable on the relative interior of the feasible set \citep{Rockafellar70:Convex}. We denote by $D_{\mathcal{R}}$ the \emph{Bregman divergence} generated by $\mathcal{R}$; that is, $D_{\mathcal{R}}(\vec{x}, \vec{y}) \defeq \mathcal{R}(\vec{x}) - \mathcal{R}(\vec{y}) - \langle \nabla \mathcal{R}(\vec{y}), \vec{x} - \vec{y} \rangle$. Canonical DGFs include the \emph{squared Euclidean} $\mathcal{R}(\vec{x}) \defeq \frac{1}{2} \|\vec{x}\|_2^2$ which generates the \emph{squared Euclidean distance}, as well as the \emph{negative entropy} $\mathcal{R}(\vec{x}) = \sum_{r=1}^d \vec{x}(r) (\log \vec{x}(r) - 1)$ which generates the \emph{Kulllback-Leibler divergence}. \emph{Optimistic mirror descent (OMD)} has an internal prediction $\vec{m}^{(t)}$ at every time $t \in \N$, so that the update rule takes the following form for $t \geq 0$:
\begin{equation}
    \tag{OMD}
    \label{eq:OMD}
    \begin{split}
        \vec{x}^{(t+1)} \defeq \argmax_{\vec{x} \in \cX} \left\{ \langle \vec{x}, \vec{m}^{(t+1)} \rangle - \frac{1}{\eta} D_{\cR} (\vec{x}, \widehat{\vec{x}}^{(t)}) \right\}; \\
        \widehat{\vec{x}}^{(t+1)} \defeq \argmax_{\widehat{\vec{x}} \in \cX} \left\{ \langle \widehat{\vec{x}}, \vec{u}^{(t+1)} \rangle - \frac{1}{\eta} D_{\cR} ( \widehat{\vec{x}}, \widehat{\vec{x}}^{(t)}) \right\},
    \end{split}
\end{equation}
where $\widehat{\Vec{x}}^{(0)} \defeq \argmin_{\widehat{\Vec{x}} \in \cX} \cR(\widehat{\Vec{x}}) \eqqcolon \vec{x}^{(0)}$. We also consider the \emph{optimistic follow the regularizer leader (OFTRL)} algorithm of \citet{Syrgkanis15:Fast}, defined as follows.
\begin{equation}
    \label{eq:OFTRL}
    \tag{OFTRL}
    \Vec{x}^{(t+1)} \defeq \argmax_{\Vec{x} \in \cX} \left\{ \left\langle \Vec{x}, \Vec{m}^{(t+1)} + \sum_{\tau = 1}^{t} \Vec{u}_i^{(\tau)} \right\rangle - \frac{1}{\eta} \cR(\Vec{x}) \right\}.
\end{equation}
Unless specified otherwise, it will be assumed that $\vec{m}^{(t)} \defeq \vec{u}^{(t-1)}$ for $t \in \N$. (To make $\vec{u}^{(0)}$ well-defined in games, each player receives at the beginning the utility corresponding to the players' strategies at $t = 0$; this is only made for convenience in the analysis.)  \citet{Syrgkanis15:Fast} established that OMD and OFTRL have the so-called \rvu property, which we recall below.

\begin{definition}[\rvu Property]
    \label{definition:rvu}
    A regret minimizing algorithm satisfies the $(\alpha, \beta, \gamma)$-\rvu property under a pair of dual norms\footnote{The dual of norm $\|\cdot\|$ is defined as $\|\vec{v}\|_* \triangleq \sup_{\|\vec{u}\| \leq 1} \langle \vec{u}, \vec{v} \rangle$.} $(\|\cdot\|, \|\cdot\|_*)$ if for any sequence of utility vectors $\vec{u}^{(1)}, \dots, \vec{u}^{(T)}$ its regret $\reg^T$ can be bounded as
    \begin{equation*}
        \label{eq:rvu}
        \reg^T \!\leq\! \alpha + \beta \sum_{t=1}^T \|\vec{u}^{(t)} - \vec{u}^{(t-1)}\|_*^2 - \gamma \sum_{t=1}^T \| \vec{x}^{(t)} - \vec{x}^{(t-1)}\|^2\!,
    \end{equation*}
    where $\alpha, \beta, \gamma > 0$ are time-independent parameters, and $\vec{x}^{(1)}, \dots, \vec{x}^{(T)}$ is the sequence of produced iterates. 
\end{definition}

\begin{proposition}[\citep{Syrgkanis15:Fast}]
    \label{proposition:rvu}
    \eqref{eq:OFTRL} with learning rate $\eta > 0$ and $\vec{m}^{(t)} = \vec{u}^{(t-1)}$ satisfies the \rvu property with $(\alpha, \beta, \gamma) = (\frac{\Omega}{\eta}, \eta, \frac{1}{4\eta})$, where $\Omega \defeq \sup_{\vec{x}} \cR(\vec{x}) - \inf_{\Vec{x}} \cR(\vec{x})$. Moreover, \eqref{eq:OMD} with learning rate $\eta > 0$ and $\vec{m}^{(t)} = \vec{u}^{(t-1)}$ satisfies the \rvu property with $(\alpha, \beta, \gamma) = (\frac{\Omega}{\eta}, \eta, \frac{1}{8\eta}) $, where $\Omega \defeq \sup_{\vec{x}} D_{\cR}(\vec{x}, \widehat{\Vec{x}}^{(0)})$.
\end{proposition}

\section{Optimistic Learning in Games}
\label{section:optimistic}

In this section we study optimistic learning dynamics. We first show that the \rvu bound implies a bounded second-order path length for dynamics such as \eqref{eq:OMD} under a broad class of games.

\begin{restatable}{theorem}{boundedtraj}
    \label{theorem:bounded_traj}
    Suppose that every player $i \in [n]$ employs a regret minimizing algorithm satisfying the $(\alpha_i, \beta_i, \gamma_i)$-\rvu property with $\gamma_i \geq 2 (n-1) \sum_{j \neq i} \beta_j$, for all $i \in [n]$, under the pair of dual norms $(\|\cdot\|_1, \|\cdot\|_{\infty})$. If the game is such that $\sum_{i=1}^n \reg_i^T \geq 0$, for any $T \in \N$, then
    \begin{equation}
        \label{eq:boundtraj}
        \sum_{i=1}^n \sum_{t=1}^T \gamma_i \| \vec{x}_i^{(t)} - \vec{x}_i^{(t-1)}\|_1^2 \leq 2 \sum_{i=1}^n \alpha_i.
    \end{equation}
\end{restatable}

The $\ell_1$-norm in \Cref{theorem:bounded_traj} is only used for convenience; \eqref{eq:boundtraj} immediately extends under any equivalent norm. Now the requirement of \Cref{theorem:bounded_traj} related to the \rvu property can be satisfied by OMD and OFTRL, as implied by \Cref{proposition:rvu}. Concretely, if all players employ \eqref{eq:OMD} with the same $\eta > 0$, it would suffice to take $\eta \leq \frac{1}{4(n-1)}$. In light of this, applying \Cref{theorem:bounded_traj} only requires that the game is such that $\sum_{i=1}^n \reg_i^T \geq 0$. Next, we show that this is indeed the case for certain important classes of games.

\subsection{Games with Nonnegative Sum of Regrets}
In \emph{polymatrix games} there is an underlying graph so that every node is associated with a given player and every edge corresponds to a two-person game. The utility of a player is the sum of the utilities obtained from each individual game played with its neighbors on the graph~\citep{Kearns01:Graphical}. In a \emph{polymatrix zero-sum game}, every two-person game is zero-sum. More generally, in a \emph{zero-sum polymatrix game} the sum of \emph{all} players' utilities has to be $0$.

\emph{Strategically zero-sum games}, introduced by \citet{Moulin78:Strategically}, is the subclass of bimatrix games which are \emph{strategically equivalent} to a zero-sum game (see \Cref{def:strat-zero_sum} for a formal statement). For this class of games, \citet{Moulin78:Strategically} showed that \emph{fictitious play} converges to a Nash equilibrium. We will extend their result under a broad class of no-regret learning dynamics. An important feature of strategically zero-sum games is that it constitutes the \emph{exact class of games} whose fully-mixed Nash equilibria cannot be improved by, \emph{e.g.}, a \emph{correlation scheme}~\citep{Aumann74:Subjectivity}.

\begin{proposition}[Abridged; Full Version in \Cref{proposition:positive_regret-full}]
    \label{proposition:positive_regret}
    For the following classes of games it holds that $\sum_{i=1}^n \reg_i^T \geq 0$:
    \begin{enumerate}[nosep]
        \item Two-player zero-sum games; 
        \item Polymatrix zero-sum games; 
        \item Constant-sum polymatrix games; 
        \item Strategically zero-sum games; 
        \item Polymatrix strategically zero-sum games;\footnote{See \Cref{remark:scale}.} 
    \end{enumerate}
\end{proposition}

As suggested in the full version (\Cref{proposition:positive_regret-full}), this class of games appears to be intricately connected with the admission of a \emph{minimax theorem}. As such, it also captures certain nonconvex-nonconcave landscapes such as \emph{quasiconvex-quasiconcave games}~\citep{Sion:On} and \emph{stochastic games}~\citep{Shapley1095}, but establishing last-iterate convergence for those settings goes beyond the scope of this paper. Note that Von Neumann's minimax theorem for two-player zero-sum games can be generalized to the polymatrix games we are considering~\citep{Daskalakis09:On, Cai11:On,Cai:Zero}. Interestingly, our framework also has implications if the \emph{duality gap} is small; see \Cref{remark:duality_gap}.

\subsection{Implications} Having established the richness of the class of games we are capturing, we present implications deriving from \Cref{theorem:bounded_traj}, and extensions thereof. First, we show that it implies optimal $O(1)$ \emph{individual regret}.

\begin{restatable}[Optimal Regret Bound]{corollary}{optreg}
    \label{corollary:optreg}
    In the setting of \Cref{theorem:bounded_traj} with $\alpha_i = \alpha$, $\beta_i = \beta$, and $\gamma_i = \gamma$, the individual regret of each player $i \in [n]$ can be bounded as
    \begin{equation*}
        \reg_i^T \leq \alpha + \frac{2 n (n-1) \alpha \beta}{\gamma}.
    \end{equation*}
\end{restatable}

For example, under \eqref{eq:OMD} with $\eta = \frac{1}{4(n-1)}$ this corollary implies that $\reg_i^T \leq 8 n \Omega_i \leq 8 n \log |\cA_i|$, for any $i \in [n]$. Moreover, we also obtain an $O(1/\epsilon^2)$ bound on the number of iterations so that the last iterate is an $\epsilon$-Nash equilibrium. 

\begin{theorem}[Abridged; Full Version in \Cref{theorem:rate-OMD-full}]
    \label{theorem:rate-OMD}
    Suppose that each player employs \eqref{eq:OMD} with a suitable regularizer. If $\sum_{i=1}^n \reg_i^T \geq 0$, then for any $\epsilon > 0$ and after $T = O(1/\epsilon^2)$ iterations there is joint strategy $\Vec{x}^{(t)}$, with $t \in [T]$, which is an $\epsilon$-approximate Nash equilibrium.
\end{theorem}
This theorem requires smoothness of the regularizer used by each player (which could be different), satisfied by a broad family that includes the Euclidean DGF. Interestingly, the argument does not appear to hold under non-smooth regularizers such as the negative entropy DGF, even if one works with common local norms. We also remark that while the bound in \Cref{theorem:rate-OMD} applies for the best iterate, it is always possible to terminate once a desired precision has been reached. Asymptotic last-iterate convergence follows immediately from our techniques; see \Cref{remark:last-iterate}.

\paragraph{Advanced Predictions} Our last-iterate guarantees are the first to apply even if players employ more advanced prediction mechanisms. Indeed, \citet{Syrgkanis15:Fast} showed that a qualitatively similar \rvu bound can be derived under
\begin{enumerate}[nosep]
    \item $H$-\emph{step recency bias}: Given a window of size $H$, we define $\vec{m}^{(t)} \defeq \sum_{\tau = t - H}^{t-1} \vec{u}^{(\tau)}/H $; \label{item:H-step}
    \item \emph{Geoemetrically discounted recency bias}: For $\delta \in (0, 1)$, we define $\vec{m}^{(t)} \defeq \frac{1}{\sum_{\tau = 0}^{t-1} \delta^{-\tau} } \sum_{\tau = 0}^{t-1} \delta^{-\tau} \vec{u}^{(\tau)}$.\label{item:discounting}
\end{enumerate}
In \Cref{proposition:advanced} we show the boundedness property for the trajectories under such prediction mechanisms, and we experimentally evaluate their performance in \Cref{appendix:experiments}.

\paragraph{Bilinear Saddle-Point Problems} Our framework also has direct implications for \emph{extensive-form games (EFGs)} where the strategy space of each player is no longer a probability simplex. 
In particular, computing a Nash equilibrium in zero-sum EFGs can be formulated as a \emph{bilinear saddle-point problem (BSPP)}
\begin{equation}
    \label{eq:BSPP}
    \min_{\vec{x} \in \cX} \max_{\Vec{y} \in \cY} \Vec{x}^\top \mat{A} \Vec{y},
\end{equation}
where $\cX$ and $\cY$ are convex and compact sets associated with the sequential decision process faced by each player~\citep{Romanovskii62:Reduction,VonStengel96:Efficient,Koller96:Efficient}.

\begin{proposition}[Abdriged; Full Version in \Cref{proposition:rate-EFGs-full}]
    \label{proposition:rate-EFGs}
    Suppose that both players in a BSPP employ \eqref{eq:OMD} with $\eta \leq \frac{1}{4\|\mat{A} \|_2}$, where $\| \cdot\|_2$ is the spectral norm. Then, after $T = O(1/\epsilon^2)$ iterations there is a pair $(\vec{x}^{(t)}, \vec{y}^{(t)})$ with $t \in [T]$ which is an $O(\epsilon)$-Nash equilibrium.
\end{proposition}

We are not aware of any prior \emph{polynomial-time} guarantees for the last iterate of no-regret learning dynamics in EFGs. Asymptotic last-iterate convergence also follows immediately. \Cref{proposition:rate-EFGs} is verified through experiments on benchmark EFGs in \Cref{section:experiments} (\Cref{fig:experiments}).

\begin{remark}[Convex-Concave Games]
    \label{remark:convex-concave}
    Our framework also applies to smooth min-max optimization via a well-known linearization trick; see \Cref{appendix:convex-concave}. Implications for unconstrained setups are also immediate; see \Cref{remark:unconstrained}.
\end{remark}


\subsection{Convergence of the Social Welfare}

We conclude this section with a surprising new result presented in \Cref{theorem:smooth}. Let us first recall the concept of a \emph{smooth game}.

\begin{definition}[Smooth Game; \citep{Roughgarden15:Intrinsic}]
    A utility-maximizing game is $(\lambda, \mu)$-\emph{smooth} if for any two action profiles $\vec{a}, \Vec{a}^* \in \cA$,
    \begin{equation*}
        \sum_{i=1}^n u_i(a_i^*, \vec{a}_{-i}) \geq \lambda \sw(\Vec{a}^*) - \mu \sw(\vec{a}).
    \end{equation*}
\end{definition}
The smoothness condition imposes a bound on the ``externality'' of any unilateral deviation. The \emph{robust price of anarchy (PoA)} of a game $\Gamma$ is defined as $\rho \defeq \sup_{\lambda, \mu} \lambda/(1 + \mu)$, where the supremum is taken over all possible $\lambda, \mu$ for which $\Gamma$ is $(\lambda, \mu)$-smooth. Smooth games with favorable smoothness parameters include simultaneous second-price auctions \citep{Christodoulou16:Bayesian}, valid utility games \citep{Vetta02:Nash}, and congestion games with affine costs \citep{Awerbuch13:The,Christodoulou05:The}. We refer to \citep{Roughgarden15:Intrinsic} for an additional discussion.

The importance of Roughgarden's smoothness framework is that it guarantees convergence (in a time average sense) of no-regret learning to outcomes with \emph{approximately optimal welfare}, where the approximation depends on the robust PoA. In particular, the convergence is controlled by the sum of the players' regrets \citep[Proposition 2]{Syrgkanis15:Fast}. We use this property to obtain the following result.
\begin{theorem}[Abridged; Full Version in \Cref{theorem:smooth-full}]
    \label{theorem:smooth}
    Suppose that each player in a game $\Gamma$ employs \eqref{eq:OMD} with a suitable regularizer and learning rate $\eta > 0$. For any $\gamma > 0$ and a sufficiently large number of iterations $T = O(1/\gamma^2)$, either of the following occurs:
    \begin{enumerate}[nosep]
        \item There is an iterate $\vec{x}^{(t)}$ with $t \in [T]$ which is an $O(\gamma)$-Nash equilibrium;
        \item Otherwise,
              \begin{equation*}
                  \frac{1}{T} \sum_{t=1}^T \sw(\vec{x}^{(t)}) \geq \frac{\lambda}{1 + \mu} \opt + \frac{1}{1 + \mu} \frac{\gamma^2}{16\eta},
              \end{equation*}
              where $(\lambda, \mu)$ are the smoothness parameters of $\Gamma$.
    \end{enumerate}
\end{theorem}

In words, the dynamics either approach arbitrarily close to a Nash equilibrium, or the time average of the social welfare outperforms the robust price of anarchy. Either of these implications is remarkable. 

\section{Convergence with the Potential Method}
\label{section:potential}

In this section we show optimal regret bounds (\Cref{theorem:optregpot}) and last-iterate rates (\Cref{theorem:ratepotential}) for the fundamental class of \emph{potential games} \cite{Monderer96:Potential,Rosenthal73:A} under mirror descent (MD) with suitable regularizers. In fact, our approach is general enough to capture under a unifying framework distributed learning in Fisher's market model~\citep{Birnbaum11:Distributed}, while we expect that further applications will be identified in the future. Finally, we show that our approach can also be applied in \emph{near-potential} games, in the precise sense of \citep{Candogan13:Dynamics}, showing convergence to approximate Nash equilibria.


We commence by formally describing the class of games we are considering in this section; in \Cref{proposition:potential} we show that it incorporates typical potential games.


\begin{definition}
    \label{def:weighted_potential}
    Let $\Phi: \prod_{i=1}^n \Delta(\cA_i) \ni (\vec{x}_1, \dots, \vec{x}_n) \mapsto \R$ be a bounded function, with $\max_{\vec{x}} |\Phi(\vec{x})| \leq \phimax$, for which there exists $L > 0$ such that for any $\Vec{x}, \widetilde{\Vec{x}}$,
    \begin{equation}
        \label{eq:onesided-smooth}
        \Phi(\vec{x})  \leq \Phi(\widetilde{\vec{x}})  - \langle \nabla_{\vec{x}} \Phi(\vec{x}), \widetilde{\vec{x}} - \vec{x} \rangle + L \| \widetilde{\vec{x}} - \vec{x} \|_2^2.
    \end{equation}
    Moreover, let $g_i$ be a strictly increasing function for each $i \in [n]$. A game $\Gamma$ is $(g_1, \dots, g_n)$-\emph{potential} if for all $i \in [n]$, $a_i \in \cA_i$, and $\vec{x}_{-i} \in \prod_{j \neq i} \Delta(\cA_j)$ we have that
    \begin{equation}
        \label{eq:global-align}
        \frac{\partial \Phi(\vec{x})}{\partial \vec{x}_i(a_i)} = g_i ( u_i(a_i, \vec{x}_{-i})).
    \end{equation}
\end{definition}
A few remarks are in order. First, \eqref{eq:onesided-smooth} imposes a ``one-sided'' smoothness condition. This relaxation turns out to be crucial to encompass settings such as linear Fisher markets \citep{Birnbaum11:Distributed}. Moreover, \eqref{eq:global-align} prescribes applying a monotone transformation to the utility. While the identity mapping $g_i : x \mapsto x$ suffices to cover typical potential games, a logarithmic transformation is required to capture the celebrated \emph{proportional response} dynamics \citep{Birnbaum11:Distributed, Wu07:Proportional} in markets.

\begin{restatable}{proposition}{potential}
    \label{proposition:potential}
    Let $\Gamma$ be a game for which there exists a function $\Phi : \prod_{i=1}^n \cA_i \to \R$ with
    \begin{equation*}
        \Phi(\vec{a}) - \Phi(a_i', \vec{a}_{-i}) = \vec{w}_i( u_i(\vec{a}) - u_i(a_i', \vec{a}_{-i})),
    \end{equation*}
    where $\vec{w} \in \R^n_{> 0}$. Then, $\Gamma$ is a potential game in the sense of \Cref{def:weighted_potential} with $L = \frac{1}{2} \phimax \sum_{i=1}^n |\cA_i|$.
\end{restatable}

In this context, we will assume that all players update their strategies using mirror descent (MD) for $t \geq 0$:
\begin{equation}
    \tag{MD}
    \label{eq:MD}
    \vec{x}_i^{(t+1)} \defeq \argmax_{\vec{x} \in \Delta(\cA_i)} \eta \langle g_i(\vec{u}_i^{(t)}), \vec{x} \rangle - D_{\mathcal{R}_i}(\vec{x}, \vec{x}_i^{(t)}).
\end{equation}
In accordance to \Cref{def:weighted_potential}, players apply (coordinate-wise) the transformation $g_i$ to the observed utility. Importantly, we will allow players to employ different regularizers, as long as $\cR_i$ is $1$-strongly convex with respect to $\|\cdot\|_2$. This trivially holds under the Euclidean DGF, while it also holds under negative entropy (Pinsker's inequality). 
We are now ready to establish that the potential function increases along non-stationary orbits:

\begin{restatable}{theorem}{potentialmonotone}
    \label{theorem:monotone}
    Suppose that each player employs \eqref{eq:MD} with a $1$-strongly convex regularizer with respect to $\|\cdot\|_2$, and learning rate $\eta = \frac{1}{2L}$, where $L$ is defined as in \Cref{proposition:potential}. Then, for any $t \geq 1$,
    \begin{equation}
        \label{eq:monotone}
        \Phi(\vec{x}^{(t+1)}) - \Phi(\vec{x}^{(t)}) \geq \frac{1}{2\eta} \sum_{i=1}^n \| \vec{x}_i^{(t+1)} - \vec{x}_i^{(t)} \|_2^2 \geq 0.
    \end{equation}
\end{restatable}
We recall that for large values of learning rate $\eta$ variants of \eqref{eq:MD} are known to exhibit chaotic behavior in potential games~\citep{Bielawski21:Follow,Palaiopanos17:Multiplicative}. \Cref{theorem:monotone} also implies the following boundedness property for the trajectories from a direct telescopic summation of \eqref{eq:monotone}:
\begin{corollary}
    \label{corollary:boundedtraj-pot}
    In the setting of \Cref{theorem:monotone},
    \begin{equation*}
        \frac{1}{2\eta} \sum_{t=1}^{T-1} \| \vec{x}^{(t+1)} - \vec{x}^{(t)} \|_2^2 \leq \Phi(\vec{x}^{(T)}) - \Phi(\vec{x}^{(1)}) \leq 2 \phimax.
    \end{equation*}
\end{corollary}
In the sequel, to argue about the regret incurred by each player and the convergence to Nash equilibria, we tacitly assume that $g_i$ is the identity map. The first implication of \Cref{corollary:boundedtraj-pot} is an $O(\sqrt{T})$ bound on the individual regrets.
\begin{restatable}{corollary}{regpot}
    \label{corollary:optreg-potential}
    In the setting of \Cref{theorem:monotone}, with $\cR_i \defeq \frac{1}{2} \|\vec{x}\|_2^2$, it holds that the regret of each player $i \in [n]$ is such that $\reg_i^T = O(\sqrt{T})$.
\end{restatable}
Note that we establish vanishing $O(1/\sqrt{T})$ average regret under \emph{constant} learning rate, thereby deviating from the traditional regret analysis of \eqref{eq:MD}. 
More importantly, with a more involved argument we show optimal individual regret under \emph{optimistic multiplcative weights update (OMWU)}:
\begin{restatable}[Optimal Regret for Potential Games]{theorem}{optregpot}
    \label{theorem:optregpot}
    Suppose that each player employs OMWU with a sufficiently small learning rate $\eta > 0$. Then, the regret of each player $i \in [n]$ is such that $\reg_i^T = O(1)$.
\end{restatable}
This theorem is based on showing a suitable boundedness property (\Cref{theorem:bounded_traj-opt}), which is then appropriately combined with \Cref{proposition:rvu}. While \Cref{theorem:optregpot} is stated for OMWU, the proof readily extends well-beyond. In this way, when the underlying game is potential, we substantially strengthen and simplify the result of \citet{Daskalakis21:Near}. Moreover, we also obtain a bound on the number of iterations required to reach an approximate Nash equilibrium.
\begin{theorem}[Abridged; Full Version in \Cref{theorem:ratepotential-full}]
    \label{theorem:ratepotential}
    Suppose that each player employs \eqref{eq:MD} with a suitable regularizer. Then, after $O(1/\epsilon^2)$ iterations there is a strategy $\vec{x}^{(t)}$ which is an $\epsilon$-Nash equilibrium.
\end{theorem}

This theorem has a similar flavour to \Cref{theorem:rate-OMD}, and it is the first of its kind for the class of potential games. Finally, in settings where the potential function is concave we show a rate of convergence of $O(1/T)$.

\begin{restatable}{proposition}{concaverate}
    \label{proposition:concaverate}
    In the setting of \Cref{theorem:monotone}, if the potential function $\Phi$ is also concave, then
    \begin{equation*}
        \Phi(\vec{x}^{*}) - \Phi(\vec{x}^{(T+1)}) \leq \frac{2L}{T} \sum_{i=1}^n D_{\cR_i}(\vec{x}_i^*, \vec{x}_i^{(1)}).
    \end{equation*}
\end{restatable}
This result is stronger than the standard convergence guarantee in smooth convex optimization since \Cref{def:weighted_potential} only imposes a one-sided condition. While concavity does not hold in typical potential games, it applies to games such as Fisher's market model; see our remark below.

\begin{remark}
    \label{remark:Fisher}
    In \Cref{appendix:Fisher} we explain how our framework naturally captures the analysis of \citet{Birnbaum11:Distributed} for distributed dynamics in Fisher's market model. 
\end{remark}

\subsection{Near-Potential Games}

Finally, we illustrate the robustness of our framework by extending our results to \emph{near-potential} games \citep{Candogan13:Dynamics}. Roughly speaking, a game is near-potential if it is close---in terms of the maximum possible utility improvement through unilateral deviations (MPD)---to a potential game; we defer the precise definition to \Cref{appendix:nearpotential}. It is also worth noting that our result immediately extends under different distance measures. In this context, we prove the following theorem.
\begin{theorem}[Abridged; Full Version in \Cref{theorem:nearpot-full}]
    \label{theorem:nearpot}
    Consider a $\delta$-near-potential game where each player employs \eqref{eq:MD} with suitable regularizer. Then, there exists a potential function $\Phi$ which increases as long as $\vec{x}^{(t)}$ is not an $O(\sqrt{\delta})$-Nash equilibrium.
\end{theorem}

\section{Continuous Games}
\label{section:continuous}


\Cref{section:optimistic,section:potential} primarily focused on classes of games stemming from applications in game theory. In this section we shift our attention to \emph{continuous ``games''}, strategic interactions motivated by applications such as GANs. Specifically, we study the convergence properties of \emph{optimistic gradient descent (OGD)} beyond two-player zero-sum settings. 
Recall that OGD in unconstrained settings can be expressed using the following update rule for $t \geq 1$:
\begin{equation}
    \label{eq:OGD}
    \tag{OGD}
    \!\vec{x}_i^{(t+1)} = \vec{x}_i^{(t)} + 2\eta \nabla_{\vec{x}_i} u_i(\vec{x}^{(t)}) - \eta \nabla_{\vec{x}_i} u_i(\vec{x}^{(t-1)}),\!\!
\end{equation}
for any player $i \in [n]$, where $u_i: \prod_{i \in [n]} \cX_i \to \R$ is assumed to be continuously differentiable.

\subsection{Two-Player Games}

We first study two-player games. Let $\mat{A}, \mat{B} \in \R^{d \times d}$ be matrices so that under strategies $(\vec{x}, \vec{y}) \in \mathcal{X} \times \mathcal{Y}$ the utilities of players $\cX$ and $\cY$ are given by the bilinear form $\vec{x}^\top \mat{A} \vec{y}$ and $\vec{x}^\top \mat{B} \vec{y}$ respectively. As is common in this line of work, the matrices are assumed to be square and non-singular. In line with the application of interest, we mostly focus on the \emph{unconstrained setting}, where $\mathcal{X} = \R^d$ and $\mathcal{X} = \R^d$; some of our results are also applicable when $\cX$ and $\cY$ are balls in $\R^d$, as we make clear in the sequel. A point $(\vec{x}^*, \vec{y}^*)$ is an equilibrium if
\begin{equation}
    \label{eq:equilibria-bi}
    \begin{split}
        \vec{x}^\top \mat{A} \vec{y}^* \leq (\vec{x}^*)^\top \mat{A} \vec{y}^*, \quad \forall \vec{x} \in \cX; \\
        (\vec{x}^*)^\top \mat{B} \vec{y} \leq (\vec{x}^*)^\top \mat{B} \vec{y}^*, \quad \forall \vec{y} \in \cY.
    \end{split}
\end{equation}

Even when $\mat{B} = - \mat{A}$, this seemingly simple saddle-point problem is---with the addition of appropriate regularization---powerful enough to capture problems such as linear regression, empirical risk minimization, and robust optimization~\citep{Du19:Linear}. Further, studying such games relates to the \emph{local convergence} of complex games encountered in practical applications \citep{Liang19:Interaction}. 
We also point out that, although not explicitly formalized, our techniques can address the addition of quadratic regularization. In this regime, our first contribution is to extend the known regime for which OGD retains stability:

\begin{theorem}[Abridged; Full Version in \Cref{theorem:two_player-convergence-full}]
    \label{theorem:two_player-convergence}
    Suppose that the matrix $\mat{A}^\top \mat{B}$ has strictly negative (real) eigenvalues. Then, for a sufficiently small learning rate $\eta > 0$, \eqref{eq:OGD} converges linearly to an equilibrium.
\end{theorem}
The proof of this theorem is based on transforming the dynamics of \eqref{eq:OGD} to the \emph{frequency domain} via the $Z$-transform, in order to derive the characteristic equation of the dynamical system. When the game is zero-sum, the condition of the theorem holds since the matrix $- \mat{A}^\top \mat{A}$ is symmetric and negative definite. As such, \Cref{theorem:two_player-convergence} extends the known results in the literature. The technique we employ can also reveal the rate of convergence in terms of the eigenvalues of $\mat{A}^\top \mat{B}$. The first question stemming from \Cref{theorem:two_player-convergence} is whether the condition of stability only captures games which are, in some sense, fully competitive. Our next result answers this question in the negative.
%
%

Indeed, it turns out that the condition of stability in \Cref{theorem:two_player-convergence} also includes games with a coordination aspect, but under \eqref{eq:OGD} the players will fail to coordinate. In turn, we show that this implies that \eqref{eq:OGD} can be \emph{arbitrarily inefficient}. More precisely, for strategies $\vec{x} \in \cX$ and $\vec{y} \in \cY$, we define the social welfare as $\sw(\vec{x}, \vec{y}) \triangleq \vec{x}^\top \mat{A} \vec{y} + \vec{x}^\top \mat{B} \vec{y}$. We show the following result (we refer to \Cref{fig:efficiency} for an illustration).

\begin{restatable}{proposition}{inefficiency}
    \label{proposition:inefficiency}
    For any sufficiently large $R > 0$, there exist games such that \eqref{eq:OGD} converges under any initialization to an equilibrium $(\vec{x}^{(\infty)}, \vec{y}^{(\infty)})$ such that $\sw(\vec{x}^{(\infty)}, \vec{y}^{(\infty)}) = 0$, while there exist an equilibrium $(\vec{x}^*, \vec{y}^*)$ with $\sw(\vec{x}^*, \vec{y}^*) \geq 2R^2$.
\end{restatable}
This holds when $\cX$ and $\cY$ are compact balls on $\R^d$ and parameter $R > 0$ controls their radius. \Cref{proposition:inefficiency} seems to suggest that the stabilizing properties of \eqref{eq:OGD} come at a dramatic loss of efficiency beyond zero-sum games.


Another interesting implication is that an arbitrarily small perturbation from a zero-sum game can destabilize \eqref{eq:OGD}: 

\begin{restatable}{proposition}{robustness}
    \label{proposition:robustness}
    For any $\epsilon > 0$ there exists a game $(\mat{A}, \mat{B})$ with $\|\mat{A} + \mat{B}\|_F \leq \epsilon$ for which \eqref{eq:OGD} diverges, while the dynamics converge for the game $(\mat{A},-\mat{A})$.
\end{restatable}

Thus, even if a game $(\mat{A}, \mat{B})$ is arbitrarily close to a zero-sum in the Frobenius norm, \eqref{eq:OGD} may still diverge. To put it differently, a small noise in one of the payoff matrices can dramatically alter the behavior of the system. This phenomenon is illustrated and further discussed in \Cref{fig:robustness}.

\subsection{Multiplayer Games}

Moreover, we also characterize the \eqref{eq:OGD} dynamics in polymatrix games. Such multiplayer interactions have already received considerable attention in the literature on GANs (\emph{e.g.}, see \citep{Hoang17:Multi}, and references therein), but the behavior of the dynamics is poorly understood even in structured games~\citep{Kalogiannis21:Teamwork}. Our next theorem characterizes the important case where a single player $1$ plays against $n-1$ different players, numbered from $2$ to $n$. In particular, the utility of player $1$ under strategies $\Vec{x} = (\Vec{x}_1, \dots, \Vec{x}_n)$ is given by $\sum_{j \neq i} \Vec{x}_i^\top \mat{A}_{1, j} \Vec{x}_j$, while the utility of player $j$ is given by $\Vec{x}_j^\top \mat{A}_{j, 1} \Vec{x}_1$. Our next result extends \Cref{theorem:two_player-convergence}.

\begin{theorem}[Abridged; Full Version in \Cref{theorem:one-many-convergence-full}]
    \label{theorem:one-many-convergence}
    If $\mat{M} \defeq \sum_{j \neq 1} \mat{A}_{1, j} \mat{A}_{j, 1}$ has strictly negative (real) eigenvalues, there exists a sufficiently small learning rate $\eta > 0$ depending only on the spectrum of the matrix $\mat{M}$ such that \eqref{eq:OGD} converges with linear rate to an equilibrium.
\end{theorem}

This condition is again trivially satisfied in polymatrix zero-sum games. In fact, \Cref{theorem:one-many-convergence} is an instance of a more general characterization for polymatrix unconstrained games, as we further explain in \Cref{remark:game-matrix}.

\begin{remark}[Beyond OGD: Stability of First-Order Methods]
    In light of the result established in \Cref{theorem:two_player-convergence}, it is natural to ask whether its condition is necessary. In \Cref{theorem:impossibility} we show that the presence of positive eigenvalues in the spectrum of matrix $\mat{A}^\top \mat{B}$ necessarily implies instability even under a \emph{generic class} of first-order methods.
\end{remark}

\section{Experiments}
\label{section:experiments}

In this section we numerically investigated the last-iterate convergence of \eqref{eq:OMD} in two benchmark zero-sum extensive-form games (EFGs). Recall that computing a Nash equilibrium in this setting can be expressed as the bilinear saddle-point problem of \eqref{eq:BSPP}. When both players employ \eqref{eq:OMD} with Euclidean regularization and $\eta \leq \frac{1}{4\|\mat{A}\|_2}$, where $\|\mat{A}\|_2$ is the spectral norm of $\mat{A}$, our results guarantee last-iterate convergence in terms of the saddle-point gap. 
\begin{wrapfigure}{r}{9cm}
    \centering
    \begin{tikzpicture}
        \draw[gray,line width=.08mm] (-4.0123456789, 1.9) -- (4.0123456789, 1.9);
        \node[anchor=center,fill=white] at (.4,  1.9) {\small Kuhn poker};
        \node at (0,  0) {\includegraphics[scale=0.7]{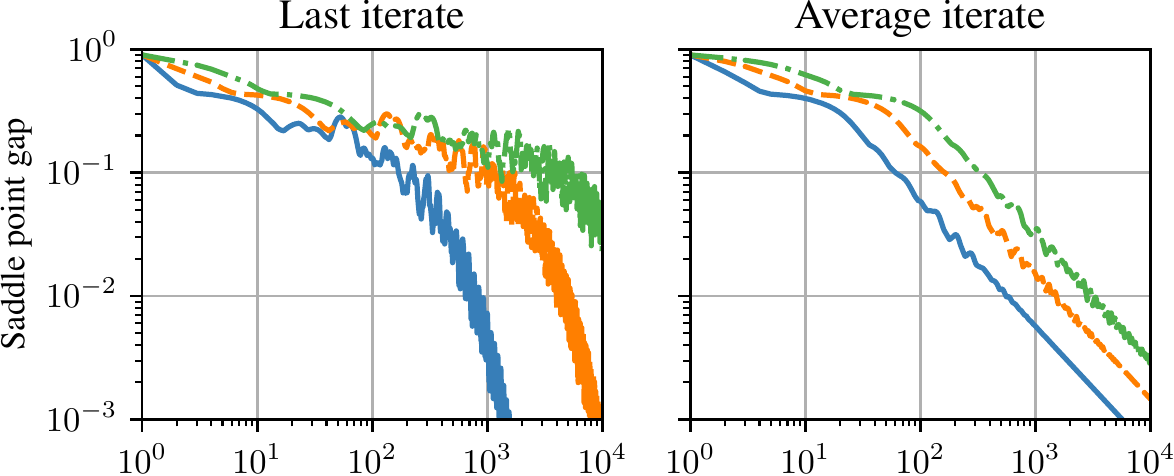}};
        \draw[gray,line width=.08mm] (-4.0123456789, -2.1) -- (4.0123456789, -2.1);
        \node[anchor=center,fill=white] at (.4,  -2.1) {\small Leduc poker};
        \node at (0, -4.4) {\includegraphics[scale=0.7]{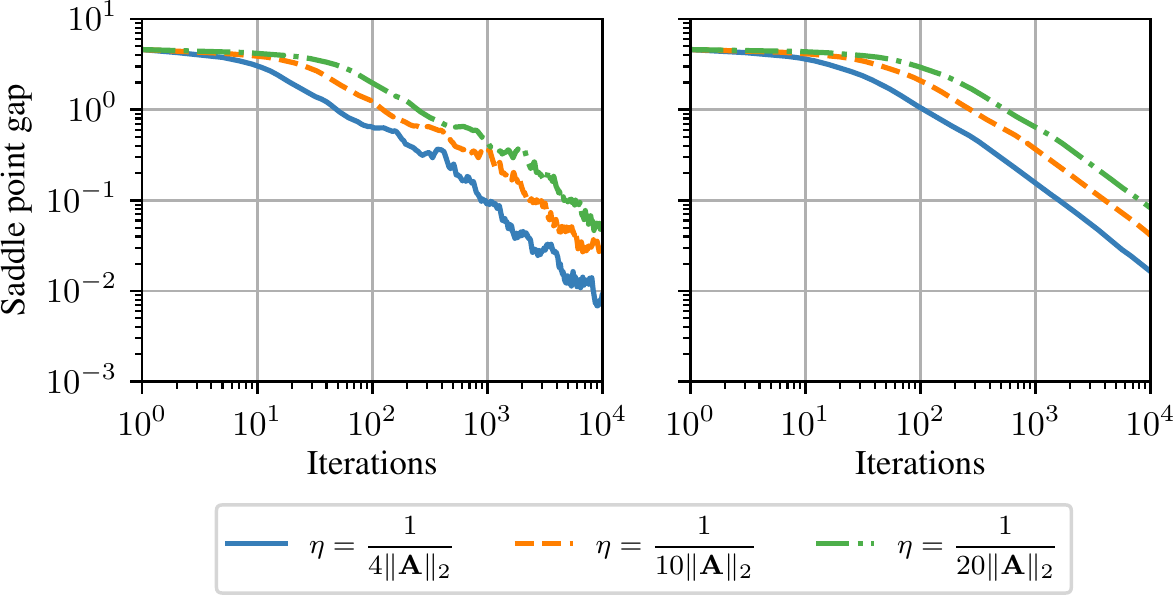}};
    \end{tikzpicture}
    \caption{The saddle-point gap of the last and average iterates of \eqref{eq:OMD} in Kuhn and Leduc poker.}
    \label{fig:experiments}
\end{wrapfigure}
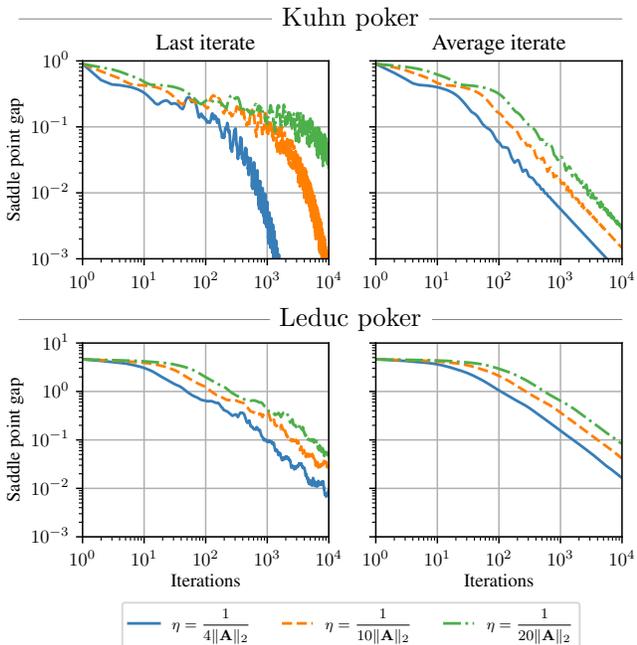
We experimented on two standard poker benchmarks known as \emph{Kuhn poker}~\citep{Kuhn50:Simplified} and \emph{Leduc poker}~\citep{Southey05:Bayes}. For each benchmark game, we ran \eqref{eq:OMD} with Euclidean regularization and three different values of $\eta$ for $10000$ iterations. After each iteration $t$, we computed the saddle-point gap corresponding to the iterates at time $t$, as well as to the average iterates up to time $t$. The results are shown in \Cref{fig:experiments}.
As expected, we observe that both the average and the last iterate converge in terms of the saddle-point gap. Moreover, we see that larger values of learning rate lead to substantially faster convergence, illustrating the importance of obtaining sharp theoretical bounds.

Finally, additional experiments related to our theoretical results have been included in \Cref{appendix:experiments}.

%


\vspace{-65pt}

\begin{center}
\end{center}
\begin{center}
\end{center}
\begin{center}
\end{center}

\section*{Acknowledgements}

Ioannis Panageas is supported by a start-up grant. Part of this work was done while Ioannis Panageas was visiting the Simons Institute for the Theory of Computing. Tuomas Sandholm is supported by the National Science Foundation under grants IIS-1901403 and CCF-1733556.

\printbibliography

\clearpage
\appendix

\section{Proofs from \texorpdfstring{\Cref{section:optimistic}}{Section 3}}

In this section we give complete proofs for our results in \Cref{section:optimistic}. For the convenience of the reader we will restate the claims before proceeding with their proofs. We begin with \Cref{theorem:bounded_traj}.

\boundedtraj*

\begin{proof}
    The proof commences similarly to \citep[Theorem 4]{Syrgkanis15:Fast}. By assumption, we know that for any player $i \in [n]$ it holds that
    \begin{equation}
        \label{eq:rvu-i}
        \reg_i^T \leq \alpha_i + \beta_i \sum_{t=1}^T \| \vec{u}_i^{(t)} - \vec{u}_i^{(t-1)} \|_{\infty}^2 - \gamma_i \sum_{t=1}^T \| \vec{x}_i^{(t)} - \vec{x}_i^{(t-1)}\|_1^2.
    \end{equation}
    Furthermore, the middle term in the \rvu property can be bounded using the following simple claim.

    \begin{claim}
        \label{claim:util-infty}
        For any player $i \in [n]$ it holds that
        \begin{equation*}
            \| \vec{u}_i^{(t)} - \vec{u}_i^{(t-1)} \|_{\infty} \leq \sum_{j \neq i} \|\vec{x}_j^{(t)} - \vec{x}_j^{(t-1)}\|_1.
        \end{equation*}
    \end{claim}

    \begin{proof}
        By the normalization assumption on the utilities we know that $|\vec{u}_i(\cdot)| \leq 1$. Thus, from the triangle inequality it follows that
        \begin{equation*}
            \| \vec{u}_i^{(t)} - \vec{u}_i^{(t-1)} \|_{\infty} \leq \sum_{\vec{a}_{-i} \in \cA_{-i} } \left| \prod_{j \neq i} \vec{x}_j^{(t)}(a_j) - \prod_{j \neq i} \vec{x}_j^{(t-1)}(a_j) \right|.
        \end{equation*}
        The induced term can be recognized as the total variation distance of two product distributions. Hence, a standard application of the sum-product inequality implies that
        \begin{equation*}
            \sum_{\vec{a}_{-i} \in \cA_{-i} } \left| \prod_{j \neq i} \vec{x}_j^{(t)}(a_j) - \prod_{j \neq i} \vec{x}_j^{(t-1)}(a_j) \right| \leq \sum_{j \neq i} \|\vec{x}_j^{(t)} - \vec{x}_j^{(t-1)}\|_1.
        \end{equation*}
    \end{proof}
    Moreover, a direct application of Young's inequality to the bound of the previous claim gives us that
    \begin{equation}
        \label{eq:util-bound}
        \|\vec{u}_i^{(t)} - \vec{u}_i^{(t-1)}\|^2_{\infty} \leq (n-1) \sum_{j \neq i} \| \vec{x}_j^{(t)} - \vec{x}_j^{(t-1)}\|^2_1.
    \end{equation}
    As a result, if we plug-in this bound to \eqref{eq:rvu-i} we may conclude that for all players $i \in [n]$ it holds that
    \begin{equation*}
        \reg_i^T \leq \alpha_i + (n-1) \beta_i \sum_{t=1}^T \sum_{j \neq i} \| \vec{x}_j^{(t)} - \vec{x}_j^{(t-1)}\|_1^2 - \gamma_i \sum_{t=1}^T \| \vec{x}_i^{(t)} - \vec{x}_i^{(t-1)}\|_1^2.
    \end{equation*}
    Summing these inequalities for all players $i \in [n]$ yields that
    \begin{align*}
        \sum_{i=1}^n \reg_i^T &\leq \sum_{i=1}^n \alpha_i + \sum_{i=1}^n \left( (n-1) \sum_{j \neq i} \beta_j - \gamma_i \right) \sum_{t=1}^T \|\vec{x}_i^{(t)} - \vec{x}_i^{(t-1)}\|_1^2 \\ 
        &\leq \sum_{i=1}^{n} \alpha_i - \frac{1}{2} \sum_{i=1}^n \gamma_i \sum_{t=1}^T \|\vec{x}_i^{(t)} - \vec{x}_i^{(t-1)}\|_1^2,
    \end{align*}
    where we used the assumption that $\gamma_i \geq 2(n-1) \sum_{j \neq i} \beta_j$ for any $i \in [n]$. Finally, the theorem follows from a rearrangement of the final bound using the assumption that $\sum_{i=1}^n \reg_i^T \geq 0$.
\end{proof}

\begin{proposition}
    \label{proposition:advanced}
    Suppose that each player employs \eqref{eq:OFTRL} with the same learning rate $\eta > 0$. Moreover, suppose that the game is such that $\sum_{i=1}^n \reg_i^T \geq 0$. Then,
    \begin{enumerate}
        \item \label{item:prop-H} If all players use $H$-step recency bias (\Cref{item:H-step}) and $\eta \leq \frac{1}{4(n-1)H}$,
              \begin{equation*}
                  \sum_{t=1}^T \sum_{i=1}^n \| \Vec{x}_i^{(t)} - \Vec{x}_i^{(t-1)}\|_1^2 \leq 8 \sum_{i=1}^n \Omega_i.
              \end{equation*}
        \item \label{item:prop-disc} If all players use geometrically $\delta$-discounted recency bias (\Cref{item:discounting}) and $\eta \leq \frac{(1 - \delta)^{3/2}}{4(n-1)}$,
              \begin{equation*}
                  \sum_{t=1}^T \sum_{i=1}^n \| \Vec{x}_i^{(t)} - \Vec{x}_i^{(t-1)}\|_1^2 \leq 16 \sum_{i=1}^n \Omega_i.
              \end{equation*}
    \end{enumerate}
\end{proposition}

\begin{proof}
    Let us first establish \Cref{item:H-step}. We will use the following \rvu bound shown by \citet[Proposition 9]{Syrgkanis15:Fast}.

    \begin{proposition}[\citep{Syrgkanis15:Fast}]
        \eqref{eq:OFTRL} with learning rate $\eta >0$ and $H$-step recency bias (\Cref{item:H-step}) satisfies the \rvu property with $(\alpha, \beta, \gamma) = (\frac{\Omega}{\eta}, \eta H^2, \frac{1}{4\eta})$, where $\Omega \defeq \sup_{\Vec{x} \in \cX} \cR(\Vec{x}) - \inf_{\Vec{x} \in \cX} \cR(\Vec{x})$.
    \end{proposition}
    As a result, we may conclude from this bound that the regret of each player $i \in [n]$ can be bounded as
    \begin{equation*}
        \reg_i^T \leq \frac{\Omega_i}{\eta} + \eta H^2 \sum_{t=1}^T \| \Vec{u}_i^{(t)} - \Vec{u}_i^{(t-1)}\|_{\infty}^2 - \frac{1}{4\eta} \sum_{t=1}^T \| \Vec{x}_i^{(t)} - \Vec{x}_i^{(t-1)}\|_1^2,
    \end{equation*}
    where $\Omega_i \defeq \sup_{\Vec{x} \in \cX_i} \cR(\Vec{x}) - \inf_{\Vec{x} \in \cX} \cR(\Vec{x})$. Next,
    using \Cref{claim:util-infty} the previous bound implies that
    \begin{equation*}
        \reg_i^T \leq \frac{\Omega_i}{\eta} + \eta (n-1) H^2 \sum_{t=1}^T \sum_{j \neq i} \| \Vec{x}_j^{(t)} - \Vec{x}_j^{(t-1)}\|_{1}^2 - \frac{1}{4\eta} \sum_{t=1}^T \| \Vec{x}_i^{(t)} - \Vec{x}_i^{(t-1)}\|_1^2.
    \end{equation*}
    Summing these bounds for all players $i \in [n]$ gives us that
    \begin{align*}
        \sum_{i=1}^n \reg_i^T &\leq \frac{1}{\eta} \sum_{i=1}^n \Omega_i + \sum_{i=1}^n \left( \eta (n-1)^2 H^2 - \frac{1}{4\eta} \right) \sum_{t=1}^T \| \Vec{x}_i^{(t)} - \Vec{x}_i^{(t-1)}\|_1^2 \\
        &\leq \frac{1}{\eta} \sum_{i=1}^n \Omega_i - \frac{1}{8\eta} \sum_{t=1}^T \sum_{i=1}^n \| \Vec{x}_i^{(t)} - \Vec{x}_i^{(t-1)}\|_1^2,
    \end{align*}
    for $\eta \leq \frac{1}{4(n-1)H}$. Thus, \Cref{item:prop-H} follows immediately under the assumption that $\sum_{i=1}^n \reg_i^T \geq 0$. Next, we proceed with the proof of \Cref{item:prop-disc}. We use the following \rvu bound shown by \citet[Proposition 10]{Syrgkanis15:Fast}.

    \begin{proposition}[\citep{Syrgkanis15:Fast}]
        \eqref{eq:OFTRL} with learning rate $\eta >0$ and geometrically $\delta$-discounted recency bias (\Cref{item:discounting}) satisfies the \rvu property with $(\alpha, \beta, \gamma) = (\frac{\Omega}{\eta}, \frac{\eta}{(1-\delta)^3}, \frac{1}{8\eta})$, where $\Omega \defeq \sup_{\Vec{x} \in \cX} \cR(\Vec{x}) - \inf_{\Vec{x} \in \cX} \cR(\Vec{x})$.
    \end{proposition}
    Thus, if player $i \in [n]$ uses geometrically $\delta$-discounted recency bias it follows that its regret can be bounded as
    \begin{equation*}
        \reg_i^T \leq \frac{\Omega_i}{\eta} + \frac{\eta}{(1 - \delta)^3} \sum_{t=1}^T \| \Vec{u}_i^{(t)} - \Vec{u}_i^{(t-1)}\|_{\infty}^2 - \frac{1}{8\eta} \sum_{t=1}^T \| \Vec{x}_i^{(t)} - \Vec{x}_i^{(t-1)}\|_1^2.
    \end{equation*}
    Combining this bound with \Cref{claim:util-infty} yields that
    \begin{equation*}
        \reg_i^T \leq \frac{\Omega_i}{\eta} + \frac{\eta(n-1)}{(1 - \delta)^3} \sum_{j \neq i} \sum_{t=1}^T \| \Vec{x}_i^{(t)} - \Vec{x}_i^{(t-1)}\|_{1}^2 - \frac{1}{8\eta} \sum_{t=1}^T \| \Vec{x}_i^{(t)} - \Vec{x}_i^{(t-1)}\|_1^2.
    \end{equation*}
    Summing these bounds for all players $i \in [n]$ gives us that
    \begin{align*}
        \sum_{i=1}^n \reg_i^T &\leq \frac{1}{\eta} \sum_{i=1}^n \Omega_i + \sum_{i=1}^n \left( \frac{\eta (n-1)^2}{(1-\delta)^3} - \frac{1}{8\eta} \right) \sum_{t=1}^T \| \Vec{x}_i^{(t)} - \Vec{x}_i^{(t-1)}\|_1^2 \\
        &\leq \frac{1}{\eta} \sum_{i=1}^n \Omega_i - \frac{1}{16\eta} \sum_{t=1}^T \sum_{i=1}^n \| \Vec{x}_i^{(t)} - \Vec{x}_i^{(t-1)}\|_1^2,
    \end{align*}
    for $\eta \leq \frac{(1 - \delta)^{3/2}}{4(n-1)}$. Finally, the fact that $\sum_{i=1}^n \reg_i^T \geq 0$ along with a rearrangement of the previous bound concludes the proof.
\end{proof}

Next, we proceed with the proof of \Cref{proposition:positive_regret}. First, we state some preliminary facts about strategically zero-sum games, a subclass of bimatrix games. It should be stressed that bimatrix games is already a very general class of games. For example, \citet{Daskalakis06:The} have shown that for \emph{succinctly representable} multiplayer games, the problem of computing Nash equilibria can be reduced to the two-player case. Below we give the formal definition of a strategically zero-sum game.

\begin{definition}[Strategically Zero-Sum Games~\citep{Moulin78:Strategically}]
    \label{def:strat-zero_sum}
    The bimatrix games $(\mat{A}, \mat{B})$ and $(\mat{C}, \mat{D})$, defined on the same action space, are \emph{strategically equivalent} if for all $\vec{x}, \vec{x}' \in \cX$ and $\vec{y}, \vec{y}' \in \cY$,
    \begin{equation*}
        \begin{split}
            \vec{x}^\top \mat{A} \vec{y} \geq (\vec{x}')^\top \mat{A} \vec{y} &\iff \vec{x}^\top \mat{C} \vec{y} \geq (\vec{x}')^\top \mat{C} \vec{y}; \\
            \vec{x}^\top \mat{B} \vec{y} \geq \vec{x}^\top \mat{B} \vec{y}' &\iff \vec{x}^\top \mat{D} \vec{y} \geq \vec{x}^\top \mat{D} \vec{y}'.
        \end{split}
    \end{equation*}
    A game is \emph{strategically zero-sum} if it is strategically equivalent to some zero-sum game.
\end{definition}
In words, consider a pair of strategies $(\vec{x}, \vec{y}) \in \mathcal{X} \times \mathcal{Y}$. Player $\cX$ can order her strategy space based on the (expected) payoff if player $\cY$ was to play $\vec{y}$, and similarly, player $\cY$ can order her strategy space under the assumption that player $\cX$ will play $\vec{x}$. In strategically equivalent games this ordering is preserved.

Without any loss, we will consider \emph{non-trivial games} in the sense of \citep[Definition 2]{Moulin78:Strategically}. The following characterization \citep[Theorem 2]{Moulin78:Strategically} will be crucial for the proof of \Cref{proposition:positive_regret}.

\begin{theorem}[\citep{Moulin78:Strategically}]
    \label{theorem:strat-zero_sum}
    Let $(\mat{A}, \mat{B})$ be a non-trivial $n \times m$\footnote{Here we overload $n$ (which typically stands for the number of players) in order to remain consistent with the prelevant notation used in this setting.} strategically zero-sum game. Then, there exist $\lambda > 0, \mu > 0$ and a compatible matrix $\mat{C}$ such that
    \begin{itemize}
        \item $\mat{A} = \lambda \mat{C} + \vec{v}_a \vec{1}_m^\top $, for some $\vec{v}_a \in \R^n $;
        \item $\mat{B} = - \mu \mat{C} + \vec{1}_n \vec{v}_b^\top $, for some $\vec{v}_b \in \R^m$.
    \end{itemize}
\end{theorem}
Here we used the notation $\Vec{1}_n$ to denote the all-ones vector in $\R^n$. The converse of \Cref{theorem:strat-zero_sum} also holds. It is also worth pointing out strategically zero-sum games can be ``embedded'' in a zero-sum game with imperfect information \citep{Moulin78:Strategically}.

\begin{remark}
    \label{remark:scale}
    For the purpose of \Cref{proposition:positive_regret} (and \Cref{proposition:positive_regret-full}) it will be assumed that $\lambda = \mu$, neutralizing any ``scale imbalances'' in the payoff matrices of the two players. We remark that if this is not the case, our result for strategically zero-sum games (\Cref{item:stra-zero_sum}) still holds by considering an appropriately \emph{weighted} sum of regrets. In this case one should select different learning rates for the two players in order to cancel out the underlying scale imbalance. Nonetheless, such an extension does not appear to work for polymatrix strategically zero-sum games (\Cref{item:polymatrix-rank}). The latter fact could be partly justified by the surprising result that even polymatrix \emph{strictly competitive games} are $\ppad$-hard \citep[Theorem 1.2]{Cai11:On}.
\end{remark}

\begin{remark}[Strictly Competitive Games]
    \label{remark:strictly_competitive}
    Strictly competitive games form a subclass of strategically zero-sum games. More precisely, a two-person game is \emph{strictly competitive} if it has the following property: if both players change their mixed strategies, then either there is no change in the expected payoffs, or one of the two expected payoffs increases and the other decreases. It was formally shown by \citet{Adler09:A} that a game $(\mat{A}, \mat{B})$ is strictly competitive if and only if $\mat{B}$ is an \emph{affine variant} of matrix $\mat{A}$; that is, $\mat{B} = \lambda \mat{A} + \mu \mat{1}$, where $\lambda > 0$, $\mu$ is unrestricted, and $\mat{1}$ denotes the all-ones matrix.
\end{remark}

Before we proceed with the proof of \Cref{proposition:positive_regret}, let us first make the following simple observation which derives from \Cref{theorem:strat-zero_sum}.

\begin{observation}
    FTRL and MD, and optimistic variants thereof, employed on a strategically zero-sum game with suitable learning rates are equivalent to the dynamics in a ``hidden'' zero-sum game.
\end{observation}

Thus, this observation reassures us that \eqref{eq:OFTRL} and \eqref{eq:OMD} dynamics (with suitable learning rates) in strategically zero-sum games inherit all of the favorable last-iterate convergence guarantees known for zero-sum games. We are now ready to prove \Cref{proposition:positive_regret}, the extended version of which is included below.

\begin{proposition}{Full Version of \Cref{proposition:positive_regret}}
    \label{proposition:positive_regret-full}
    For the following classes of games it holds that $\sum_{i=1}^n \reg_i^T \geq 0$:
    \begin{enumerate}[nosep]
        \item Two-player zero-sum games; \label{item:zero-sum}
        \item Polymatrix zero-sum games; \label{item:polymatrix}
        \item Constant-sum Polymatrix games; \label{item:global}
        \item Strategically zero-sum games; \label{item:stra-zero_sum}
        \item Polymatrix strategically zero-sum games;\footnote{See \Cref{remark:scale}.} \label{item:polymatrix-rank}
        \item Convex-concave (zero-sum) games; \label{item:convex-concave}
        \item Zero-sum games with objective $f(\Vec{x}, \Vec{y})$ such that \label{item:minimax}
              \begin{equation*}
                  \min_{\vec{x} \in \cX} \max_{\Vec{y} \in \cY} f(\Vec{x}, \Vec{y}) = \max_{\Vec{y} \in \cY} \min_{\vec{x} \in \cX} f(\Vec{x}, \Vec{y});
              \end{equation*}
        \item Quasiconvex-quasiconcave (zero-sum) games; \label{item:quasiconvex-quasiconcave}
        \item Zero-sum stochastic games. \label{item:stochastic}
    \end{enumerate}
\end{proposition}

\begin{proof}
    We proceed separately for each class of games.
    \begin{itemize}[leftmargin=*]
        \item For \Cref{item:zero-sum} let us denoted by $\mat{A}$ the matrix associated with the zero-sum game, so that player $\cX$ is the ``minimizer'' and player $\cY$ is the ``maximizer''. Further, let us denote by $\Bar{\vec{x}} \defeq (\vec{x}^{(1)} + \dots \vec{x}^{(T)})/T$ and $\Bar{\vec{y}} \defeq (\vec{y}^{(1)} + \dots \vec{y}^{(T)})/T$ the time average of the strategies employed by the two players respectively up to time $T \in \N$. Then, we have that $\reg_{\cX}^T + \reg_{\cY}^T$ is equal to
              \begin{align*}
                  &\max_{\vec{x}^* \in \cX} \left\langle \vec{x}^*, - \mat{A} \sum_{t=1}^T \vec{y}^{(t)} \right\rangle + \sum_{t=1}^T \langle \vec{x}^{(t)}, \mat{A} \vec{y}^{(t)} \rangle + \max_{\vec{y}^* \in \cY} \left\langle \vec{y}^*, \mat{A}^\top \sum_{t=1}^T \vec{x}^{(t)} \right\rangle - \sum_{t=1}^T \langle \vec{y}^{(t)}, \mat{A}^\top \vec{x}^{(t)} \rangle \\
                                              & = T \left( \max_{\vec{y}^* \in \cY} \langle \vec{y}^*, \mat{A}^\top \Bar{\vec{x}} \rangle - \min_{\vec{x}^* \in \cX} \left\langle \vec{x}^*, \mat{A} \Bar{\vec{y}} \right\rangle \right) \geq T \left( \langle \Bar{\vec{y}}, \mat{A}^\top \Bar{\vec{x}} \rangle - \langle \Bar{\vec{x}}, \mat{A} \Bar{\vec{y}} \rangle \right) = 0.
              \end{align*}
        \item For \Cref{item:polymatrix} let us denoted by $\cN_i \subseteq [n]$ the neighborhood of the node associated with each player $i \in [n]$ so that the (expected) utility of player $i \in [n]$ under a joint (mixed) strategy vector $\Vec{x} \defeq (\Vec{x}_1, \dots, \Vec{x}_n)$ can be expressed as $\sum_{j \in \cN_i} \vec{x}_i^\top \mat{A}_{i, j} \vec{x}_j$, where $\mat{A}_{i, j}$ is the payoff matrix associated with the edge $i \to j$. Observe that since every edge corresponds to a zero-sum game, it holds that $\mat{A}_{i, j} = - \mat{A}_{j, i}^\top$. Moreover, we let $\Bar{\vec{x}}_i \defeq (\vec{x}_i^{(1)} + \dots + \vec{x}_i^{(T)})/T$ be the time average of player $i$'s strategies up to time $T$. We have that
              \begin{align}
                  \frac{1}{T} \sum_{i=1}^n \reg_i^T & = \frac{1}{T} \sum_{i=1}^n \left( \max_{\vec{x}_i^* \in \cX_i} \left\{ \left\langle \vec{x}_i^*, \sum_{t=1}^T  \sum_{j \in \cN_i} \mat{A}_{i, j} \vec{x}_j^{(t)} \right\rangle \right\} - \sum_{t=1}^T \left\langle \vec{x}_i^{(t)}, \sum_{j \in \cN_i} \mat{A}_{i, j} \vec{x}_j^{(t)} \right\rangle \right) \notag \\
                                                    & = \sum_{i=1}^n \max_{\vec{x}_i^* \in \cX_i} \left\{ \left\langle \vec{x}_i^*, \sum_{j \in \cN_i} \mat{A}_{i, j} \Bar{\vec{x}}_j \right\rangle \right\} - \frac{1}{T} \sum_{i=1}^n \sum_{t=1}^T \left\langle \vec{x}_i^{(t)}, \sum_{j \in \cN_i} \mat{A}_{i, j} \vec{x}_j^{(t)} \right\rangle \notag                 \\
                                                    & = \sum_{i=1}^n \max_{\vec{x}_i^* \in \cX_i} \left\{ \left\langle \vec{x}_i^*, \sum_{j \in \cN_i} \mat{A}_{i, j} \Bar{\vec{x}}_j \right\rangle \right\} \label{align:zero_sum-canc}                                                                                                                                  \\
                                                    & \geq \sum_{i=1}^n \left\langle \Bar{\vec{x}}_i, \sum_{j \in \cN_i} \mat{A}_{i, j} \Bar{\vec{x}}_j \right \rangle = 0, \label{align:zero_sum-canc-2}
              \end{align}
              where \eqref{align:zero_sum-canc} and \eqref{align:zero_sum-canc-2} follow from the fact that $\mat{A}_{i, j} = - \mat{A}_{j, i}^\top$ for any $i \neq j$.
        \item For \Cref{item:global} the proof is analogous to \Cref{item:polymatrix} using the assumption that
              \begin{equation*}
                  \sum_{i=1}^n u_i(\Vec{x}) = \sum_{i=1}^n \left\langle \Vec{x}_i, \sum_{j \in \cN_i} \mat{A}_{i, j} \Vec{x}_j \right\rangle = C,
              \end{equation*}
              for any $\Vec{x} = (\Vec{x}_1, \dots, \Vec{x}_n) \in \prod_{i \in [n]} \Delta(\cA_i)$, where $C$ is some arbitrary constant.

        \item For \Cref{item:stra-zero_sum} we let $(\mat{A}, \mat{B})$ be the underlying (bimatrix) strategically zero-sum game, with $\mat{A}, \mat{B} \in \R^{n \times m}$. We will use \Cref{theorem:strat-zero_sum} under the assumption that $\lambda = \mu$; see \Cref{remark:scale}. That is, we know that there exists a (compatible) matrix $\mat{C}$ and $\lambda > 0$ such that $\mat{A} = \lambda \mat{C} + \vec{v}_a \vec{1}_m^\top$ and $\mat{B} = - \lambda \mat{C} + \vec{1}_n \vec{v}_b^\top$, where $\vec{v}_a \in \R^n$ and $\vec{v}_b \in \R^m$. Similarly to the proof for \Cref{item:zero-sum}, we have that $\reg_{\cX}^T + \reg_{\cY}^T$ is equal to
              \begin{align}
                  &\max_{\vec{x}^* \in \Delta^n} \left\langle \vec{x}^*, \mat{A} \sum_{t=1}^T \vec{y}^{(t)} \right\rangle - \sum_{t=1}^T \langle \vec{x}^{(t)}, \mat{A} \vec{y}^{(t)} \rangle + \max_{\vec{y}^* \in \Delta^m} \left\langle \vec{y}^*, \mat{B}^\top \sum_{t=1}^T \vec{x}^{(t)} \right\rangle - \sum_{t=1}^T \langle \vec{y}^{(t)}, \mat{B}^\top \vec{x}^{(t)} \rangle \notag                                                                                                                \\
                                              & = \max_{\vec{x}^* \in \Delta^n} \left\{ \lambda \left\langle \vec{x}^*, \mat{C} \sum_{t=1}^T \vec{y}^{(t)}  \right\rangle + T \left\langle \vec{x}^*, \vec{v}_a \right\rangle \right\} - \lambda \sum_{t=1}^T \left\langle \vec{x}^{(t)}, \mat{C} \vec{y}^{(t)}  \right\rangle - \sum_{t=1}^T \left\langle \vec{x}^{(t)}, \vec{v}_a \right\rangle \label{algin:vec1-sim}                                                                                                                  \\
                                              & + \max_{\vec{y}^* \in \Delta^m} \left\{ -\lambda \left\langle \vec{y}^*, \mat{C}^\top \sum_{t=1}^T \vec{x}^{(t)}  \right\rangle + T \left\langle \vec{y}^*, \vec{v}_b \right\rangle \right\} + \lambda \sum_{t=1}^T \left\langle \vec{y}^{(t)}, \mat{C}^\top \vec{x}^{(t)}  \right\rangle - \sum_{t=1}^T \left\langle \vec{y}^{(t)}, \vec{v}_b \right\rangle \label{algin:vec1-sim2}                                                                                                      \\
                                              & = T \left( \max_{\vec{x}^* \in \Delta^n} \left\{ \lambda \left\langle \vec{x}^*, \mat{C} \bar{\vec{y}}  \right\rangle + \left\langle \vec{x}^*, \vec{v}_a \right\rangle \right\} - \left\langle \bar{\vec{x}}, \vec{v}_a \right\rangle + \max_{\vec{y}^* \in \Delta^m} \left\{ -\lambda \left\langle \vec{y}^*, \mat{C}^\top \bar{\vec{x}}  \right\rangle + \left\langle \vec{y}^*, \vec{v}_b \right\rangle \right\} - \left\langle \bar{\vec{y}}, \vec{v}_b \right\rangle \right) \notag \\
                                              & \geq T (\lambda \langle \bar{\vec{x}}, \mat{C} \bar{\vec{y}} \rangle - \lambda \langle \bar{\vec{y}}, \mat{C}^\top \bar{\vec{x}}\rangle ) = 0. \notag
              \end{align}
              where in \eqref{algin:vec1-sim} we used the fact that $\mat{A} = \lambda \mat{C} + \vec{v}_a \vec{1}_m^\top$ and that $\vec{1}^\top_m \vec{y}^{(t)} = 1$ since $\vec{y}^{(t)} \in \Delta^m$, and \eqref{algin:vec1-sim2} follows given that $\mat{B} = -\lambda \mat{C} + \vec{v}_b \vec{1}_m^\top$ and that $\vec{1}^\top_n \vec{x}^{(t)} = 1$ since $\vec{x}^{(t)} \in \Delta^n$.
        \item For \Cref{item:polymatrix-rank} the proof follows analogously to \Cref{item:polymatrix} and \Cref{item:stra-zero_sum}; see \Cref{remark:scale}.
        \item For \Cref{item:convex-concave} let $f(\vec{x}, \vec{y})$ be any convex-concave function; that is, $f(\Vec{x}, \Vec{y})$ is convex with respect to $\Vec{x} \in \cX$ for any fixed $\Vec{y} \in \cY$ and concave with respect to $\Vec{y} \in \cY$ for any fixed $\Vec{x} \in \cX$. Moreover, let us assume that player $\cX$ is the ``minimizer'' and player $\cX$ the ``maximizer''. We have that $\reg_{\cX}^T + \reg_{\cY}^T $ is equal to
              \begin{align}
                  &- \min_{\vec{x}^* \in \cX} \left\{ \sum_{t=1}^T f(\vec{x}^*, \vec{y}^{(t)}) \right\} + \sum_{t=1}^T f(\vec{x}^{(t)}, \vec{y}^{(t)}) + \max_{\vec{y}^* \in \cY} \left\{ \sum_{t=1}^T f(\vec{x}^{(t)}, \vec{y}^*) \right\} - \sum_{t=1}^T f(\vec{x}^{(t)}, \vec{y}^{(t)}) \notag \\
                                              & \geq - \sum_{t=1}^T f(\bar{\vec{x}}, \vec{y}^{(t)}) + \sum_{t=1}^T f(\vec{x}^{(t)}, \bar{\vec{y}}) \geq - T f(\bar{\vec{x}}, \bar{\vec{y}}) + T f(\bar{\vec{x}}, \bar{\vec{y}}) = 0, \notag
              \end{align}
              where the last inequality follows from Jensen's inequality, applicable since $f(\Vec{x}, \Vec{y})$ is convex-concave.
        \item For \Cref{item:minimax} we assume that $\cX$ and $\cY$ are convex and compact sets. Then, we have that
              \begin{align}
                  \frac{\reg_{\cX}^T + \reg_{\cY}^T}{T} & = \max_{\Vec{y}^* \in \cY} \left\{ \frac{1}{T} \sum_{t=1}^T f(\Vec{x}^{(t)}, \Vec{y}^*) \right\} - \min_{\Vec{x} \in \cX} \left\{ \frac{1}{T} \sum_{t=1}^T f(\Vec{x}^*, \Vec{y}^{(t)}) \right\} \notag \\
                                                        & \geq \max_{\vec{y}^* \in \cY} \min_{\Vec{x}^* \in \cX} f(\Vec{x}^*, \Vec{y}^*) - \min_{\Vec{x}^* \in \cX} \max_{\Vec{y} \in \cY} f(\Vec{x}^*, \Vec{y}^*) = 0, \label{align:minimax}
              \end{align}
              where \eqref{align:minimax} uses the following inequalities:
              \begin{equation*}
                  \frac{f(\Vec{x}^{(1)}, \Vec{y}^*) + \dots + f(\Vec{x}^{(T)}, \Vec{y}^*)}{T} \geq \min_{\Vec{x}^* \in \cX} f(\Vec{x}^*, \Vec{y}^*);
              \end{equation*}
              \begin{equation*}
                  \frac{f(\Vec{x}^*, \Vec{y}^{(1)}) + \dots + f(\Vec{x}^{*}, \Vec{y}^{(T)})}{T} \leq \max_{\Vec{y}^* \in \cY} f(\Vec{x}^*, \Vec{y}^*).
              \end{equation*}
        \item For \Cref{item:quasiconvex-quasiconcave} it is assumed that $f(\Vec{x}, \Vec{y})$ is lower semicontinuous and quasiconvex with respect to $\Vec{x} \in \cX$ for any fixed $\Vec{y} \in \cY$, and upper semicontinuous and quasiconcave with respect to $\Vec{y} \in \cY$ for any fixed $\Vec{x} \in \cX$, where $\cX$ and $\cY$ are convex and compact sets. By Sion's minimax theorem~\citep{Sion:On} we know that
              \begin{equation*}
                  \max_{\vec{y}^* \in \cY} \min_{\Vec{x}^* \in \cX} f(\Vec{x}^*, \Vec{y}^*) = \min_{\Vec{x}^* \in \cX} \max_{\Vec{y} \in \cY} f(\Vec{x}^*, \Vec{y}^*),
              \end{equation*}
              and the conclusion follows from \Cref{item:minimax}.
        \item For \Cref{item:stochastic} the claim follows directly by \Cref{item:minimax} by virtue of Shapley's theorem~\citep{Shapley1095}.
    \end{itemize}
\end{proof}

\begin{remark}[Duality Gap]
    \label{remark:duality_gap}
    Consider a function $f : \cX \times \cY \to \R$. Weak duality implies that
    \begin{equation*}
        \min_{\Vec{x}^* \in \cX} \max_{\Vec{y} \in \cY} f(\Vec{x}^*, \Vec{y}^*) \geq \max_{\vec{y}^* \in \cY} \min_{\Vec{x}^* \in \cX} f(\Vec{x}^*, \Vec{y}^*).
    \end{equation*}
    The value $\gap \defeq \min_{\Vec{x}^* \in \cX} \max_{\Vec{y} \in \cY} f(\Vec{x}^*, \Vec{y}^*) - \max_{\vec{y}^* \in \cY} \min_{\Vec{x}^* \in \cX} f(\Vec{x}^*, \Vec{y}^*) \geq 0$ is called the \emph{duality gap} of $f$. Analogously to the proof of \Cref{proposition:positive_regret} (\Cref{item:minimax}) we may conclude that
    \begin{equation*}
        \reg_{\cX}^T + \reg_{\cY}^T \geq - \gap  T.
    \end{equation*}
    Now observe that if we relax the condition of \Cref{theorem:bounded_traj} so that $\sum_{i=1}^n \reg_i^T \geq - \gap T$, then for $\alpha_i = \alpha$, $\beta_i = \beta$, and $\gamma_i = \gamma$ it follows that
    \begin{equation*}
        \sum_{t=1}^T \sum_{i=1}^n \| \Vec{x}_i^{(t)} - \Vec{x}_i^{(t-1)}\|_1^2 \leq \frac{2n \alpha}{\gamma} + \frac{2 \gap}{\gamma} T.
    \end{equation*}
    This observation, along with the argument of \Cref{theorem:rate-OMD}, imply that for suitable regularizers and for a sufficiently large $T$ there will exist a joint strategy vector $\Vec{x}^{(t)}$ with $t \in [T]$ which is an $O(\sqrt{\gap})$-Nash equilibrium, as long as $\sum_{i=1}^n \reg_i^T \geq -\gap T$, where $\gap > 0$. In fact, this argument is analogous to that we provide for near-potential games in \Cref{theorem:nearpot}.
\end{remark}

\optreg*

\begin{proof}
    First of all, when $\gamma_i = \gamma$ and $\alpha_i = \alpha$ for all $i \in [n]$, the bound we obtained in \Cref{theorem:bounded_traj} can be simplified as
    \begin{equation}
        \label{eq:simple-bounded_traj}
        \sum_{i=1}^n \sum_{t=1}^T \| \vec{x}_i^{(t)} - \vec{x}_i^{(t-1)} \|_1^2 \leq \frac{2n \alpha}{\gamma}.
    \end{equation}
    Moreover, the \rvu property implies that the regret of each individual player $i \in [n]$ can be bounded as
    \begin{equation}
        \label{eq:beta-rvu}
        \reg_i^T \leq \alpha + \beta \sum_{t=1}^T \| \vec{u}_i^{(t)} - \vec{u}_i^{(t-1)}\|^2_{\infty} - \gamma \sum_{t=1}^T \| \vec{x}_i^{(t)} - \vec{x}_i^{(t-1)}\|_1^2 \leq \alpha + (n-1) \beta \sum_{j \neq i} \sum_{t=1}^T \| \vec{x}_i^{(t)} - \vec{x}_i^{(t-1)}\|^2_1,
    \end{equation}
    where we used \Cref{claim:util-infty} and the fact that $\gamma > 0$. As a result, plugging-in bound \eqref{eq:simple-bounded_traj} (which follows from \Cref{theorem:bounded_traj}) to \eqref{eq:beta-rvu} concludes the proof.
\end{proof}

Next, we proceed with the proof of \Cref{theorem:rate-OMD}, the detailed version of which is given below.

\begin{theorem}[Full Version of \Cref{theorem:rate-OMD}]
    \label{theorem:rate-OMD-full}
    Suppose that each player employs \eqref{eq:OMD} with (i) pair of norms $(\|\cdot\|, \|\cdot\|_*)$ such that $\|\Vec{x}\| \geq C \|\Vec{x}\|_1$ and $\|\Vec{x}\|_* \leq C_* \|\Vec{x}\|_\infty$ for any $\Vec{x} \in \cX_i$; (ii) $G_i$-smooth regularizer $\cR_i$; and (iii) learning rate $\eta \leq \frac{C}{4C_* (n-1)}$. Moreover, suppose that the game is such that $\sum_{i=1}^n \reg_i^T \geq 0$ for any $T \in \N$. Then, for any $\epsilon > 0$, after $T > \left\lceil \frac{8}{\epsilon^2} \sum_{i=1}^n \Omega_i \right\rceil$ iterations there exists an iterate $\Vec{x}^{(t)}$ with $t \in [T]$ which is an
    \begin{equation*}
        \epsilon \left( C_* + 2 \frac{\max_{i \in [n]} \left\{ G_i \Omega_i'\right\}}{\eta} \right)
    \end{equation*}
    approximate Nash equilibrium, where $\Omega_i \defeq \sup_{\Vec{x}, \Vec{y} \in \cX_i} D_{\cR_i}(\Vec{x},\Vec{y})$ and $\Omega_i' \defeq \sup_{\Vec{x}, \Vec{y} \in \cX_i} \| \Vec{x} - \Vec{y} \|$.
\end{theorem}

\begin{proof}
    We will use the following refined \rvu bound, extracted from \citep[Proof of Theorem 18]{Syrgkanis15:Fast}:
    \begin{proposition}[\nakedcite{Syrgkanis15:Fast}]
        \label{proposition:refined-rvu}
        If a player $i$ employs \eqref{eq:OMD}, it holds that
        \begin{equation*}
            \reg_i^T \leq \frac{\Omega_i}{\eta} + \eta \sum_{t=1}^{T} \| \vec{u}_i^{(t)} - \vec{u}_i^{(t-1)} \|_*^2 - \frac{1}{4\eta} \sum_{t=1}^{T} \left( \| \vec{x}_i^{(t)} - \widehat{\vec{x}}_i^{(t)}\|^2 + \| \vec{x}_i^{(t)} - \widehat{\vec{x}}_i^{(t-1)}\|^2 \right).
        \end{equation*}
    \end{proposition}
    Using the norm-equivalence bounds $\|\Vec{x}\| \geq C \|\Vec{x}\|_1$ and $\|\Vec{x}\|_* \leq C_* \|\Vec{x}\|_\infty$ to the bound of \Cref{proposition:refined-rvu} yields that
    \begin{align*}
        \reg_i^T \leq \frac{\Omega_i}{\eta} + \eta C_*^2 \sum_{t=1}^{T} \| \vec{u}_i^{(t)} - \vec{u}_i^{(t-1)} \|_{\infty}^2 & - \frac{1}{8\eta} C^2 \sum_{t=1}^{T} \| \vec{x}_i^{(t)} - \widehat{\vec{x}}_i^{(t)}\|_1^2 + \| \vec{x}_i^{(t)} - \widehat{\vec{x}}_i^{(t-1)}\|_1^2 \\
                                                                                                                             & - \frac{1}{8\eta} \sum_{t=1}^T \| \Vec{x}_i^{(t)} - \widehat{\Vec{x}}_i^{(t)} \|^2 + \| \Vec{x}_i^{(t)} - \widehat{\Vec{x}}_i^{(t-1)} \|^2
    \end{align*}

    Moreover, combining this bound with \Cref{claim:util-infty} implies that

    \begin{align*}
        \reg_i^T & \leq \frac{\Omega_i}{\eta} + \eta C_*^2 (n-1) \sum_{t=1}^{T} \sum_{j \neq i} \| \vec{x}_j^{(t)} - \vec{x}_j^{(t-1)} \|_{1}^2 - \frac{1}{8\eta} C^2 \sum_{t=1}^{T} \left( \| \vec{x}_i^{(t)} - \widehat{\vec{x}}_i^{(t)}\|_1^2 + \| \vec{x}_i^{(t)} - \widehat{\vec{x}}_i^{(t-1)}\|_1^2 \right) \\
                 & - \frac{1}{8\eta} \sum_{i=1}^T \left( \| \Vec{x}_i^{(t)} - \widehat{\Vec{x}}_i^{(t)} \|^2 + \| \Vec{x}_i^{(t)} - \widehat{\Vec{x}}_i^{(t-1)} \|^2 \right)                                                                                                                                      \\
                 & \leq \frac{\Omega_i}{\eta} + \eta C_*^2 (n-1) \sum_{t=1}^{T} \sum_{j \neq i} \| \vec{x}_j^{(t)} - \vec{x}_j^{(t-1)} \|_{1}^2 - \frac{C^2}{16 \eta} \sum_{t=1}^{T} \|\Vec{x}_i^{(t)} - \Vec{x}_i^{(t-1)}\|_1^2                                                                                  \\
                 & - \frac{1}{8\eta} \sum_{i=1}^T \left( \| \Vec{x}_i^{(t)} - \widehat{\Vec{x}}_i^{(t)} \|^2 + \| \Vec{x}_i^{(t)} - \widehat{\Vec{x}}_i^{(t-1)} \|^2 \right),
    \end{align*}
    where we used the fact that
    \begin{equation*}
        \sum_{t=1}^{T} \|\Vec{x}_i^{(t)} - \Vec{x}_i^{(t-1)}\|^2 \leq 2 \sum_{t=1}^T \|\Vec{x}_i^{(t)} - \widehat{\Vec{x}}_i^{(t)}\|^2 + 2 \sum_{t=1}^T \|\Vec{x}_i^{(t)} - \widehat{\Vec{x}}_i^{(t-1)} \|_1^2,
    \end{equation*}
    which in turn follows from Young's inequality. As a result, we may conclude that
    \begin{align*}
        \sum_{i=1}^n \reg_i^T & \leq \sum_{i=1}^n \frac{\Omega_i}{\eta} + \left( \eta C_*^2(n-1)^2 - \frac{C^2}{16\eta} \right) \sum_{i=1}^n \sum_{t=1}^{T} \|\Vec{x}_i^{(t)} - \Vec{x}_i^{(t-1)} \|_1^2 \\
                              & - \frac{1}{8\eta} \sum_{i=1}^n \sum_{i=1}^T \left( \| \Vec{x}_i^{(t)} - \widehat{\Vec{x}}_i^{(t)} \|^2 + \| \Vec{x}_i^{(t)} - \widehat{\Vec{x}}_i^{(t-1)} \|^2 \right).
    \end{align*}
    Thus, for learning rate $\eta \leq \frac{C}{4 C_* (n-1)}$ it follows that
    \begin{equation*}
        0 \leq \sum_{i=1}^n \reg_i^T \leq \frac{1}{\eta} \sum_{i=1}^n \Omega_i - \frac{1}{8\eta} \sum_{i=1}^n \sum_{i=1}^T \left( \| \Vec{x}_i^{(t)} - \widehat{\Vec{x}}_i^{(t)} \|^2 + \| \Vec{x}_i^{(t)} - \widehat{\Vec{x}}_i^{(t-1)} \|^2 \right),
    \end{equation*}
    implying that
    \begin{equation}
        \label{eq:bounded-traj-rvu}
        \sum_{t=1}^{T} \sum_{i=1}^n  \left( \| \vec{x}_i^{(t)} - \widehat{\vec{x}}_i^{(t)}\|^2 + \| \vec{x}_i^{(t)} - \widehat{\vec{x}}_i^{(t-1)}\|^2 \right) \leq 8 \sum_{i=1}^n \Omega_i.
    \end{equation}
    Now assume that for all $t \in [T]$ it holds that $\sum_{i=1}^n \left( \| \vec{x}_i^{(t)} - \widehat{\vec{x}}_i^{(t)}\|^2 + \| \vec{x}_i^{(t)} - \widehat{\vec{x}}_i^{(t-1)}\|^2 \right) > \epsilon^2$. In this case, it follows from \eqref{eq:bounded-traj-rvu} that
    \begin{equation*}
        \epsilon^2 T \leq 8 \sum_{i=1}^n \Omega_i \implies T \leq \frac{8}{\epsilon^2} \sum_{i=1}^n \Omega_i.
    \end{equation*}
    Thus, for $T > \left\lceil \frac{8}{\epsilon^2} \sum_{i=1}^n \Omega_i \right\rceil$ it must be the case that there exists $t \in [T]$ such that
    \begin{equation*}
        \sum_{i=1}^n \left( \| \vec{x}_i^{(t)} - \widehat{\vec{x}}_i^{(t)}\|^2 + \| \vec{x}_i^{(t)} - \widehat{\vec{x}}_i^{(t-1)}\|^2 \right) \leq \epsilon^2.
    \end{equation*}

    In turn, this implies that for any $i \in [n]$,
    \begin{itemize}
        \item[(i)] $\| \vec{x}_i^{(t)} - \widehat{\vec{x}}_i^{(t)}\| \leq \epsilon$;
        \item[(ii)] $\| \widehat{\Vec{x}}_i^{(t)} - \widehat{\Vec{x}}_i^{(t-1)} \|^2 \leq 2 \| \vec{x}_i^{(t)} - \widehat{\vec{x}}_i^{(t)}\|^2 + 2 \| \vec{x}_i^{(t)} - \widehat{\vec{x}}_i^{(t-1)}\|^2 \leq 2 \epsilon^2 \implies \| \widehat{\Vec{x}}_i^{(t)} - \widehat{\Vec{x}}_i^{(t-1)} \| \leq 2 \epsilon$.
    \end{itemize}
    Finally, we show the following claim which will conclude the proof.

    \begin{claim}
        \label{claim:close-Nash}
        In the setting of \Cref{theorem:rate-OMD-full}, suppose that
        \begin{equation*}
            \sqrt{\sum_{i=1}^n \left( \| \vec{x}_i^{(t)} - \widehat{\vec{x}}_i^{(t)}\|^2 + \| \vec{x}_i^{(t)} - \widehat{\vec{x}}_i^{(t-1)}\|^2 \right)} \leq \epsilon.
        \end{equation*}
        Then, it follows that  $\Vec{x}^{(t)}$ is an
        \begin{equation*}
            \epsilon \left( C_* + 2 \frac{\max_{i \in [n]} \{ G_i \Omega_i' \}}{\eta} \right)
        \end{equation*}
        approximate Nash equilibrium.
    \end{claim}

    \begin{proof}
        Observe that the maximization problem associated with \eqref{eq:OMD} can be expressed in the following variational inequality form:
        \begin{equation*}
            \left\langle \vec{u}_i^{(t)} - \frac{1}{\eta} \left(\nabla \cR_i ( \widehat{\vec{x}}^{(t)}_i) - \nabla \cR_i ( \widehat{\vec{x}}_i^{(t-1)})\right), \widehat{\vec{x}}_i - \widehat{\vec{x}}_i^{(t)} \right\rangle \leq 0, \quad \forall \widehat{\vec{x}}_i \in \cX_i,
        \end{equation*}
        for any $i \in [n]$. Thus, it follows that
        \begin{align}
            \langle \vec{u}_i^{(t)}, \widehat{\vec{x}}_i - \widehat{\vec{x}}_i^{(t)} \rangle & \leq \frac{1}{\eta} \langle \nabla \cR_i ( \widehat{\vec{x}}^{(t)}_i) - \nabla \cR_i ( \widehat{\vec{x}}_i^{(t-1)}), \widehat{\vec{x}}_i - \widehat{\vec{x}}_i^{(t)} \rangle \notag                     \\
                                                                                             & \leq \frac{1}{\eta} \| \nabla \cR_i ( \widehat{\vec{x}}^{(t)}_i) - \nabla \cR_i ( \widehat{\vec{x}}_i^{(t-1)}) \|_* \| \widehat{\vec{x}}_i - \widehat{\vec{x}}_i^{(t)} \| \label{align:cauchyschartz-1} \\
                                                                                             & \leq 2 \epsilon \frac{G_i \Omega_i'}{\eta}, \label{align:nash-approx-1}
        \end{align}
        where \eqref{align:cauchyschartz-1} follows from the Cauchy-Schwarz inequality, and \eqref{align:nash-approx-1} uses the fact that $\|\nabla \cR_i ( \widehat{\vec{x}}^{(t)}_i) - \nabla \cR_i ( \widehat{\vec{x}}_i^{(t-1)}) \|_* \leq G_i \|\widehat{\vec{x}}_i^{(t)} - \widehat{\vec{x}}_i^{(t-1)}\| \leq 2 \epsilon G_i$, which follows from the smoothness assumption on the regularizer $\cR_i$. \eqref{align:nash-approx-1} also uses the notation $\Omega_i' \defeq \sup_{\Vec{x}, \Vec{y} \in \cX_i} \| \Vec{x} - \Vec{y} \|$. As a result, we have established that for any player $i \in [n]$ it holds that for any $\widehat{\vec{x}}_i \in \cX_i$,
        \begin{equation}
            \label{eq:Nash-approx}
            \langle \vec{u}_i^{(t)}, \widehat{\vec{x}}_i^{(t)} \rangle \geq \langle \vec{u}_i^{(t)}, \widehat{\vec{x}}_i \rangle - 2 \epsilon \frac{G_i \Omega_i'}{\eta}.
        \end{equation}
        Moreover, we also have that
        \begin{equation*}
            \left| \langle \Vec{u}_i^{(t)}, \Vec{x}_i^{(t)} - \widehat{\Vec{x}}_i^{(t)} \rangle \right| \leq \| \Vec{u}_i^{(t)} \|_* \| \Vec{x}_i^{(t)} - \widehat{\Vec{x}}_i^{(t)}\| \leq \epsilon C_*,
        \end{equation*}
        where we used the fact that $\| \Vec{x}_i^{(t)} - \widehat{\Vec{x}}_i^{(t)}\| \leq \epsilon$, and that $\| \Vec{u}_i^{(t)}\|_{\infty} \leq 1$ (by the normalization hypothesis). Plugging-in the last bound to \eqref{eq:Nash-approx} gives us that
        \begin{equation*}
            \langle \vec{x}_i^{(t)}, \vec{u}_i^{(t)} \rangle \geq \langle \widehat{\vec{x}}_i^{(t)}, \vec{u}_i^{(t)} \rangle - \epsilon C_* \geq \langle \widehat{\vec{x}}_i, \vec{u}_i^{(t)} \rangle - \epsilon C_* - 2 \epsilon \frac{G_i \Omega_i'}{\eta},
        \end{equation*}
        for any $\widehat{\vec{x}}_i \in \cX_i$ and player $i \in [n]$. As a result, the proof follows by definition of approximate Nash equilibria (\Cref{def:Nash}).
    \end{proof}
\end{proof}

\begin{remark}[Asymptotic Last-Iterate Convergence]
    \label{remark:last-iterate}
    Leveraging our argument in the proof of \Cref{theorem:rate-OMD-full} it follows that for any $\epsilon > 0$ there exists a sufficiently large $T = T(\epsilon)$ so that $\| \Vec{x}_i^{(t)} - \widehat{\Vec{x}}_i^{(t)}\|_1 \leq \epsilon$ and $\| \Vec{x}_i^{(t)} - \widehat{\Vec{x}}_i^{(t-1)}\| \leq \epsilon$ , for any $i \in [n]$ and $t \geq T$. In turn, under smooth regularizers this implies that any $\Vec{x}^{(t)}$ with $t \geq T = T(\epsilon)$ will be an $O(\epsilon)$-Nash equilibrium by virtue of \Cref{claim:close-Nash}, establishing asymptotic last-iterate convergence. On the other hand, our techniques do not seem to imply \emph{pointwise convergence}.
\end{remark}

Next, we give the proof of \Cref{theorem:smooth}, the detailed version of which is given in \Cref{theorem:smooth-full}. To this end, we will require the following proposition shown by \citet{Syrgkanis15:Fast}, which is a slight refinement of a result due to \citet{Roughgarden15:Intrinsic}.

\begin{proposition}[\citep{Syrgkanis15:Fast}]
    \label{proposition:smooth}
    Consider a $(\lambda, \mu)$-smooth game. If each player $i \in [n]$ incurs regret at most $\reg_i^T$, then
    \begin{equation*}
        \frac{1}{T} \sum_{t=1}^T \sw(\vec{x}^{(t)}) \geq \frac{\lambda}{1 + \mu} \opt - \frac{1}{1 + \mu} \frac{1}{T} \sum_{i=1}^n \reg_i^T.
    \end{equation*}
\end{proposition}

\begin{theorem}[Full Version of \Cref{theorem:smooth}]
    \label{theorem:smooth-full}
    Suppose that each player employs \eqref{eq:OMD} with (i) pair of norms $(\|\cdot\|, \|\cdot\|_*)$ such that $\|\Vec{x}\| \geq C \|\Vec{x}\|_1$ and $\|\Vec{x}\|_* \leq C_* \|\Vec{x}\|_\infty$ for any $\Vec{x} \in \cX_i$; (ii) $G_i$-smooth regularizer $\cR_i$; and (iii) learning rate $\eta \leq \frac{C}{4C_* (n-1)}$. Moreover, fix any $\gamma > 0$ and consider $T$ iterations of the dynamics with $T \geq \frac{16 \sum_{i=1}^n \Omega_i}{\gamma^2}$, where $\Omega_i \defeq \sup_{\Vec{x}, \Vec{y} \in \cX_i} D_{\cR_i}(\Vec{x}, \Vec{y})$. Then, either of the following occurs:
    \begin{enumerate}
        \item There exists an iterate $\Vec{x}^{(t)}$ with $t \in [T]$ which is a
              \begin{equation*}
                  \gamma \left( C_* + 2 \frac{\max_{i \in [n]} \{ G_i \Omega_i'\}}{\eta} \right)
              \end{equation*}
              approximate Nash equilibrium;
        \item Otherwise, it holds that
              \begin{equation*}
                  \frac{1}{T} \sum_{t=1}^T \sw(\Vec{x}^{(t)}) \geq \frac{\lambda}{1 + \mu} \opt + \frac{\gamma^2}{16\eta}.
              \end{equation*}
    \end{enumerate}
\end{theorem}

\begin{proof}
    Similarly to the proof of \Cref{theorem:rate-OMD-full}, we may conclude that for $\eta \leq \frac{C}{4C_* (n-1)}$ it holds that
    \begin{equation}
        \label{eq:sum-reg}
        \sum_{i=1}^n \reg_i^T \leq \frac{1}{\eta} \sum_{i=1}^n \Omega_i - \frac{1}{8\eta} \sum_{i=1}^T \sum_{i=1}^n \left( \| \Vec{x}_i^{(t)} - \widehat{\Vec{x}}_i^{(t)} \|^2 + \| \Vec{x}_i^{(t)} - \widehat{\Vec{x}}_i^{(t-1)} \|^2 \right).
    \end{equation}
    Now if there exists $t \in [T]$ such that
    \begin{equation*}
        \sum_{i=1}^n \left( \| \Vec{x}_i^{(t)} - \widehat{\Vec{x}}_i^{(t)} \|^2 + \| \Vec{x}_i^{(t)} - \widehat{\Vec{x}}_i^{(t-1)} \|^2 \right) \leq \gamma^2,
    \end{equation*}
    it follows from \Cref{claim:close-Nash} that $\Vec{x}^{(t)}$ is a
    \begin{equation*}
        \gamma \left( C_* + 2 \frac{\max_{i \in [n]} \{ G_i \Omega_i'\}}{\eta} \right)
    \end{equation*}
    approximate Nash equilibrium. In the contrary case, it follows from \eqref{eq:sum-reg} that\footnote{Cf., see \citep[Theorem 9]{Hsieh21:Adaptive}; that result is only asymptotic.}

    \begin{equation*}
        \sum_{i=1}^n \reg_i^T \leq \frac{1}{\eta} \sum_{i=1}^n \Omega_i - \frac{1}{8\eta} \sum_{t=1}^T \gamma^2 \leq - \frac{1}{16\eta} \gamma^2 T,
    \end{equation*}
    as long as $T \geq \frac{16 \sum_{i=1}^n \Omega_i}{\gamma^2}$. Thus, we may conclude from \Cref{proposition:smooth} that
    \begin{equation*}
        \frac{1}{T} \sum_{t=1}^T \sw(\Vec{x}^{(t)}) \geq \frac{\lambda}{1 + \mu} \opt + \frac{1}{1 + \mu} \frac{\gamma^2}{16\eta}.
    \end{equation*}
\end{proof}

\subsection{Smooth Convex-Concave Games}
\label{appendix:convex-concave}

In this subsection we explain how our framework can also be applied in the context of \emph{smooth min-max optimization}. To be precise, let $f(\vec{x}, \vec{y}) : \cX \times \cY \to \R$ be a continuously differentiable \emph{convex-concave} function; that is, $f(\vec{x}, \vec{y})$ is convex with respect to $\vec{x} \in \cX$ for any fixed $\vec{y} \in \cY$ and concave with respect to $\vec{y} \in \cY$ for any fixed $\vec{x} \in \cX$. We will also make the following standard $\ell_2$-smoothness assumptions:
\begin{equation}
    \label{eq:smooth-x}
    \begin{split}
        \| \nabla_{\vec{x}}f(\vec{x}, \vec{y}) - \nabla_{\vec{x}}f(\vec{x}, \vec{y}') \|_2 \leq L \| \vec{y} - \vec{y}'\|_2, \quad \forall \vec{x} \in \cX, \vec{y}, \vec{y}' \in \cY; \\
        \| \nabla_{\vec{x}}f(\vec{x}, \vec{y}) - \nabla_{\vec{x}}f(\vec{x}
        ', \vec{y}) \|_2 \leq L \| \vec{x} - \vec{x}'\|_2, \quad \forall \vec{x}, \vec{x}' \in \cX, \vec{y} \in \cY;
    \end{split}
\end{equation}
and
\begin{equation}
    \label{eq:smooth-y}
    \begin{split}
        \| \nabla_{\vec{y}}f(\vec{x}, \vec{y}) - \nabla_{\vec{y}}f(\vec{x}, \vec{y}') \|_2 \leq L \| \vec{y} - \vec{y}'\|_2, \quad \forall \vec{x} \in \cX, \vec{y}, \vec{y}' \in \cY; \\
        \| \nabla_{\vec{y}}f(\vec{x}, \vec{y}) - \nabla_{\vec{y}}f(\vec{x}
        ', \vec{y}) \|_2 \leq L \| \vec{x} - \vec{x}'\|_2, \quad \forall \vec{x}, \vec{x}' \in \cX, \vec{y} \in \cY.
    \end{split}
\end{equation}

For notational simplicity we are using a common smoothness parameter $L$ in \eqref{eq:smooth-x} and \eqref{eq:smooth-y}, but a more refined analysis follows directly from our techniques.

The key idea is the well-known fact that one can construct a regret minimizer for convex utility functions via a regret minimizer for linear utilities (\emph{e.g.}, see \citep{McMahan11:Follow}). Indeed, we claim that the incurred regret $\reg^T$ under convex utility functions can be bounded by the regret $\reg_{\cL}$ of an algorithm observing the tangent plane of the utility function at every decision point. To see this, let $\nabla_{\vec{x}} u^{(t)}(\vec{x}^{(t)})$ be the gradient of a convex and continuously differentiable utility function $u^{(t)}$ on some convex domain $\cX$.\footnote{The same technique applies more generally using any \emph{subgradient} at the decision point.} Then, by convexity, we have that
\begin{equation*}
    u^{(t)}(\vec{x}) \geq \vec{u}^{(t)}(\vec{x}^{(t)}) + \langle \nabla_{\vec{x}} u^{(t)}(\vec{x}^{(t)}), \vec{x} - \vec{x}^{(t)} \rangle, \quad \forall \vec{x} \in \cX.
\end{equation*}
From this inequality it is easy to conclude that
\begin{equation}
    \label{eq:convex-linear}
    \reg^T \leq \reg_{\cL}^T,
\end{equation}
as we claimed.

Now let us assume, for concreteness, that each player employs \eqref{eq:OMD} with Euclidean regularization, while observing the \emph{linearized utility functions} based on the aforementioned scheme. Then, \Cref{proposition:rvu} implies that the individual regret of each player can be bounded as

\begin{equation}
    \label{eq:rvu-unconstrained}
    \begin{split}
        \reg_{\cX, \cL} &\leq \frac{\Omega_{\cX}}{\eta} + \eta \sum_{t=1}^T \| \nabla_{\vec{x}}f(\vec{x}^{(t)}, \vec{y}^{(t)}) - \nabla_{\vec{x}}f(\vec{x}^{(t-1)}, \vec{y}^{(t-1)}) \|_2^2 - \frac{1}{8\eta} \sum_{t=1}^T \| \vec{x}^{(t)} - \vec{x}^{(t-1)}\|_2^2; \\
        \reg_{\cY, \cL} &\leq \frac{\Omega_{\cY}}{\eta} + \eta \sum_{t=1}^T \| \nabla_{\vec{y}}f(\vec{x}^{(t)}, \vec{y}^{(t)}) - \nabla_{\vec{y}}f(\vec{x}^{(t-1)}, \vec{y}^{(t-1)}) \|_2^2 - \frac{1}{8\eta} \sum_{t=1}^T \| \vec{y}^{(t)} - \vec{y}^{(t-1)}\|_2^2,
    \end{split}
\end{equation}
where $\Omega_{\cX}$ and $\Omega_{\cY}$ are defined as in \Cref{proposition:rvu}. Now it follows from \eqref{eq:smooth-x} and Young's inequality that
\begin{align}
    \| \nabla_{\vec{x}}f(\vec{x}^{(t)}, \vec{y}^{(t)}) - \nabla_{\vec{x}}f(\vec{x}^{(t-1)}, \vec{y}^{(t-1)}) \|_2^2 \leq & 2 \| \nabla_{\vec{x}}f(\vec{x}^{(t)}, \vec{y}^{(t)}) - \nabla_{\vec{x}}f(\vec{x}^{(t)}, \vec{y}^{(t-1)}) \|_2^2 + \notag   \\
                                                                                                                         & 2 \| \nabla_{\vec{x}}f(\vec{x}^{(t)}, \vec{y}^{(t-1)}) - \nabla_{\vec{x}}f(\vec{x}^{(t-1)}, \vec{y}^{(t-1)}) \|_2^2 \notag \\
    \leq                                                                                                                 & 2L^2 \| \vec{x}^{(t)} - \vec{x}^{(t-1)}\|_2^2 + 2L^2 \| \vec{y}^{(t)} - \vec{y}^{(t-1)}\|_2^2. \label{align:x-L}
\end{align}
Similarly, it follows from \eqref{eq:smooth-y} and Young's inequality that
\begin{align}
    \| \nabla_{\vec{y}}f(\vec{x}^{(t)}, \vec{y}^{(t)}) - \nabla_{\vec{y}}f(\vec{x}^{(t-1)}, \vec{y}^{(t-1)}) \|_2^2 \leq & 2 \| \nabla_{\vec{y}}f(\vec{x}^{(t)}, \vec{y}^{(t)}) - \nabla_{\vec{y}}f(\vec{x}^{(t)}, \vec{y}^{(t-1)}) \|_2^2 + \notag   \\
                                                                                                                         & 2 \| \nabla_{\vec{y}}f(\vec{x}^{(t)}, \vec{y}^{(t-1)}) - \nabla_{\vec{y}}f(\vec{x}^{(t-1)}, \vec{y}^{(t-1)}) \|_2^2 \notag \\
    \leq                                                                                                                 & 2L^2 \| \vec{x}^{(t)} - \vec{x}^{(t-1)}\|_2^2 + 2L^2 \| \vec{y}^{(t)} - \vec{y}^{(t-1)}\|_2^2. \label{align:y-L}
\end{align}
As a result, plugging \eqref{align:x-L} and \eqref{align:y-L} to the \rvu bound of \eqref{eq:rvu-unconstrained} gives us that $\reg_{\cX, \cL} + \reg_{\cY, \cL}$ can be upper bounded by
\begin{equation}
    \label{eq:uncon-sum}
    \frac{\Omega_{\cX} + \Omega_{\cY}}{\eta} + \left( 4\eta L^2 - \frac{1}{8\eta} \right) \sum_{t=1}^T \| \vec{x}^{(t)} - \vec{x}^{(t-1)} \|_2^2 + \left( 4\eta L^2 - \frac{1}{8\eta} \right) \sum_{t=1}^T \| \vec{y}^{(t)} - \vec{y}^{(t-1)} \|_2^2.
\end{equation}
Finally, letting $\eta = \frac{1}{8L}$ in \eqref{eq:uncon-sum} implies that
\begin{equation}
    \label{eq:uncon-sum-2}
    \reg_{\cX, \cL} + \reg_{\cY, \cL} \leq 8L (\Omega_{\cX} + \Omega_{\cY}) - \frac{L}{2} \sum_{t=1}^T \| \vec{x}^{(t)} - \vec{x}^{(t-1)}\|_2^2 - \frac{L}{2} \sum_{t=1}^T \|\vec{y}^{(t)} - \vec{y}^{(t-1)}\|_2^2.
\end{equation}
Moreover, we know from \Cref{proposition:positive_regret} that $\reg^T_{\cX} + \reg^T_{\cY} \geq 0$, implying that $\reg^T_{\cX, \cL} + \reg^T_{\cY, \cL} \geq 0$ by virtue of \eqref{eq:convex-linear}. As a result, we may conclude from \eqref{eq:uncon-sum-2} the following theorem.
\begin{theorem}
    \label{theorem:convex-concave}
    Let $f(\vec{x}, \vec{y})$ be a continuously differentiable convex-concave function satisfying the $L$-smoothness condition of \eqref{eq:smooth-x} and \eqref{eq:smooth-y}. If both players employ \eqref{eq:OMD} with Euclidean regularization and $\eta = \frac{1}{8L}$, then
    \begin{equation*}
        \sum_{t=1}^T \| \vec{x}^{(t)} - \vec{x}^{(t-1)}\|_2^2 + \sum_{t=1}^T \|\vec{y}^{(t)} - \vec{y}^{(t-1)}\|_2^2 \leq 16 (\Omega_{\cX} + \Omega_{\cY}).
    \end{equation*}
\end{theorem}

This theorem combined with \eqref{eq:rvu-unconstrained} directly gives us that every player incurs $O(1)$ individual regret. Further, a bound on the number of iterations required to reach an approximate equilibrium can be derived similarly to \Cref{theorem:rate-OMD}.

\begin{remark}[Unconstrained Setting]
    \label{remark:unconstrained}
    While our framework based on regret minimization---and in particular \Cref{theorem:convex-concave}---requires $\Omega_{\cX}$ and $\Omega_{\cY}$ to be bounded, extensions are possible to the unconstrained setup as well. Indeed, as was pointed out by \citet[Remark 4]{Golowich20:Tight}, it suffices to use the fact that the iterates of (unconstrained) optimistic gradient descent remain within a bounded ball that depends only the initialization; the later property was shown in \citep[Lemma 4]{Mokhtari20:Convergence}. This can be combined with \Cref{theorem:convex-concave} to bound the norm of the gradients $\Lambda^{(t)} \defeq \| \nabla_{\Vec{x}} f(\Vec{x}^{(t)}, \Vec{y}^{(t)}) \| + \| \nabla_{\Vec{y}} f(\Vec{x}^{(t)}, \Vec{y}^{(t)}) \|$ at time $t$. More precisely, \Cref{theorem:convex-concave} implies that after $T$ iterations there exists time $t \in [T]$ such that $\Lambda^{(t)} = O(1/\sqrt{T})$; cf., see \citet{Golowich20:Tight}. We point out that $\Lambda^{(t)}$ is the most common measure for last-iterate convergence in min-max optimization~\citep{Diakonikolas:Potential}.
\end{remark}

\begin{remark}[Curvature Exploitation]
    The linearization trick we employed can be suboptimal as the regret minimization algorithm could fail to exploit the \emph{curvature} of the utility functions; \emph{e.g.}, this would be the case if the objective function $f$ is strongly-convex-strongly-concave. Extending our framework to address this issue is an interesting direction for the future.
\end{remark}

\subsection{Bilinear Saddle-Point Problems}
\label{appendix:EFGs}

In this subsection we show that our framework also has direct implications for extensive-form games (EFGs). A comprehensive overview of extensive-form games would lead us beyond the scope of this paper. Instead, we will focus on two-player zero-sum games wherein the computation of a Nash equilibrium can be formulated as a \emph{bilinear saddle-point problem} (BSPP). Namely, a BSPP can be expressed as
\begin{equation*}
    \min_{\vec{x} \in \cX} \max_{\vec{y} \in \cY} \Vec{x}^\top \mat{A} \Vec{y},
\end{equation*}
where $\cX$ and $\cY$ are convex and compact sets, and $\mat{A} \in \R^{n \times m}$. In the case of EFGs, $\cX$ and $\cY$ are the \emph{sequence-form strategy polytopes} of the sequential decision process faced by the two players, and $\mat{A}$ is a matrix with the leaf payoffs of the game. We first show the following.

\begin{proposition}
    \label{proposition:EFGs}
    Suppose that both players in a BSPP employ \eqref{eq:OMD} with Euclidean regularizer and $\eta \leq \frac{1}{4\|\mat{A}\|_2}$, where $\| \mat{A}\|_2$ is the spectral norm of $\mat{A}$. Then, 
    \begin{equation*}
        \sum_{t=1}^T \left( \| \Vec{x}^{(t)} - \Vec{x}^{(t-1)} \|_2^2 + \| \Vec{y}^{(t)} - \Vec{y}^{(t-1)} \|_2^2 \right) \leq 16 (\Omega_{\cX} + \Omega_{\cY}).
    \end{equation*}
\end{proposition}

\begin{proof}
    Given that player $\cX$ employs \eqref{eq:OMD} with Euclidean regularization, \Cref{proposition:rvu} implies that\footnote{Although \Cref{proposition:rvu} was only stated for simplexes by \citet{Syrgkanis15:Fast}, the proof readily extends for arbitrary convex and compact sets.}
    \begin{align}
        \reg^T_{\cX} & \leq \frac{\Omega_{\cX}}{\eta} + \eta \sum_{t=1}^T \| \mat{A} \Vec{y}^{(t)} - \mat{A} \Vec{y}^{(t-1)} \|_2^2 - \frac{1}{8\eta} \sum_{t=1}^T \| \Vec{x}^{(t)} - \Vec{x}^{(t-1)} \|_2^2 \notag              \\
                     & \leq \frac{\Omega_{\cX}}{\eta} + \eta \| \mat{A}\|_2^2 \sum_{t=1}^T \|\Vec{y}^{(t)} - \Vec{y}^{(t-1)} \|_2^2 - \frac{1}{8\eta} \sum_{t=1}^T \| \Vec{x}^{(t)} - \Vec{x}^{(t-1)} \|_2^2. \label{align:regx}
    \end{align}
    Similarly, we have that
    \begin{align}
        \reg^T_{\cY} & \leq \frac{\Omega_{\cY}}{\eta} + \eta \sum_{t=1}^T \| \mat{A}^\top \Vec{x}^{(t)} - \mat{A}^\top \Vec{x}^{(t-1)} \|_2^2 - \frac{1}{8\eta} \sum_{t=1}^T \| \Vec{y}^{(t)} - \Vec{y}^{(t-1)} \|_2^2 \notag    \\
                     & \leq \frac{\Omega_{\cY}}{\eta} + \eta \| \mat{A}\|_2^2 \sum_{t=1}^T \|\Vec{x}^{(t)} - \Vec{x}^{(t-1)} \|_2^2 - \frac{1}{8\eta} \sum_{t=1}^T \| \Vec{y}^{(t)} - \Vec{y}^{(t-1)} \|_2^2, \label{align:regy}
    \end{align}
    where we used the fact that $\| \mat{A}^\top \|_2 = \| \mat{A} \|_2$. Thus, summing \eqref{align:regx} and \eqref{align:regy} yields that $\reg_{\cX}^T + \reg_{\cY}^T$ can be upper bounded as
    \begin{align*}
        &\frac{\Omega_{\cX} + \Omega_{\cY}}{\eta} + \left( \eta \| \mat{A}\|_2^2 - \frac{1}{8\eta} \right) \sum_{t=1}^T \| \Vec{x}^{(t)} - \Vec{x}^{(t-1)} \|_2^2 + \left( \eta \| \mat{A}\|_2^2 - \frac{1}{8\eta} \right) \sum_{t=1}^T \| \Vec{y}^{(t)} - \Vec{y}^{(t-1)} \|_2^2. \\
                                           & \leq \frac{\Omega_{\cX} + \Omega_{\cY}}{\eta} - \frac{1}{16\eta} \sum_{t=1}^T \| \Vec{x}^{(t)} - \Vec{x}^{(t-1)} \|_2^2 - \frac{1}{16\eta} \sum_{t=1}^T \| \Vec{y}^{(t)} - \Vec{y}^{(t-1)} \|_2^2,
    \end{align*}
    where we used the fact that $\eta \leq \frac{1}{4 \| \mat{A} \|_2}$. A rearrangement of the final bound along with the fact that $\reg_{\cX}^T + \reg_{\cY}^T \geq 0$ completes the proof.
\end{proof}

Moreover, we can employ the argument of \Cref{theorem:rate-OMD} to also bound the number of iterations required to reach an $\epsilon$-approximate Nash equilibrium of the BSPP. In the following claim we use the $O(\cdot)$ notation to suppress (universal) constants.

\begin{proposition}[Full Version of \Cref{proposition:rate-EFGs}]
    \label{proposition:rate-EFGs-full}
    Suppose that both players in a BSPP employ \eqref{eq:OMD} with Euclidean regularization and $\eta \leq \frac{1}{4\|\mat{A}\|_2}$. Then, after $T = O\left(\frac{\Omega_{\cX} + \Omega_{\cY}}{\epsilon^2}\right)$ iterations there is a joint strategy $(\Vec{x}^{(t)}, \Vec{y}^{(t)})$ with $t \in [T]$ which is an $O(\epsilon \max\{\Omega_{\cX}, \Omega_{\cY} \} \|\mat{A} \|_2)$-Nash equilibrium of the BSPP.
\end{proposition}

\section{Proofs from \texorpdfstring{\Cref{section:potential}}{Section 4}}

In this section we give complete proofs for our results in \Cref{section:potential}. We commence with the simple proof of \Cref{proposition:potential}, which is recalled below.

\potential*

\begin{proof}
    By assumption, we know that there exists a (bounded) function $\Phi$ such that
    \begin{equation}
        \label{eq:potential}
        \Phi(\vec{a}) - \Phi(a_i', \vec{a}_{-i}) = \vec{w}_i ( u_i(\vec{a}) - u_i(a_i', \vec{a}_{-i})).
    \end{equation}
    That is, $\Gamma$ is a \emph{weighted potential game}: the difference in utility deriving from a deviation of a player $i \in [n]$ is translated to the exact same deviation in the potential function, modulo the scaling factor $\Vec{w}_i$. We will first show that the condition of \eqref{eq:potential} can be translated for the mixed extensions as well. As is common, in the sequel we slightly abuse notation by using the same symbols for the mixed extensions of the potential function and the utility function of each player.

    \begin{claim}
        \label{claim:mixedpotential}
        For any $(\Vec{x}_1, \dots, \Vec{x}_n) \in \prod_{i \in [n]} \Delta(\cA_i)$ it holds that
        \begin{equation}
            \Phi(\vec{x}) - \Phi(\vec{x}_i', \vec{x}_{-i}) = \vec{w}_i ( u_i(\vec{x}) - u_i(\vec{x}_i', \vec{x}_{-i})),
        \end{equation}
        where $\Phi(\vec{x}) \defeq \E_{\vec{a} \sim \vec{x}}[\Phi(\vec{a})] = \sum_{(a_1, \dots, a_n) \in \cA} \Phi(a_1, \dots, a_n) \prod_{i = 1}^n \vec{x}_i(a_i)$.
    \end{claim}

    \begin{proof}
        We have that
        \begin{align}
            \Phi(\vec{x}) - \Phi(\vec{x}_i', \vec{x}_{-i}) & = \sum_{\vec{a} \in \cA} \Phi(\vec{a}) (\vec{x}_i(a_i) - \vec{x}_i'(a_i) ) \prod_{j\neq i}  \vec{x}_j(a_j) \notag                                                                         \\
                                                           & = \sum_{a_i \in \cA_i} (\vec{x}_i(a_i) - \vec{x}_i'(a_i) ) \sum_{\vec{a}_{-i} \in \cA_{-i}} \Phi(a_i, \vec{a}_{-i}) \prod_{j \neq i}  \vec{x}_j(a_j) \notag                               \\
                                                           & = \vec{w}_i \sum_{a_i \in \cA_i} (\vec{x}_i(a_i) - \vec{x}_i'(a_i) ) \sum_{\vec{a}_{-i} \in \cA_{-i}} u_i(a_i, \vec{a}_{-i}) \prod_{j \neq i}  \vec{x}_j(a_j) \label{align:differ}        \\
                                                           & = \vec{w}_i \left( \sum_{\vec{a} \in \cA} u_i(\vec{a}) \prod_{j=1}^n \vec{x}_j(a_j) - \sum_{\vec{a} \in \cA} u_i(\vec{a}) \vec{x}_i'(a_i) \prod_{k\neq i}^n \vec{x}_k(a_k) \right) \notag \\
                                                           & =\vec{w}_i ( u_i(\vec{x}) - u_i(\vec{x}_i', \vec{x}_{-i})) \notag,
        \end{align}
        where \eqref{align:differ} follows from a rearrangement of the terms and the assumption of \eqref{eq:potential}.
    \end{proof}
    Now observe that from \Cref{claim:mixedpotential} we may conclude that

    \begin{align*}
        \frac{\partial \Phi(\vec{x})}{\partial \vec{x}_i(a_i)} = \vec{w}_i \frac{\partial u_i(\vec{x})}{\partial \vec{x}_i(a_i)} & = \vec{w}_i \frac{\partial}{\partial \vec{x}_i(a_i)} \left( \sum_{\vec{a} \in \cA} u_i(\vec{a}) \prod_{j=1}^n \vec{x}_j(a_j) \right) \notag \\
                                                                                                                                 & = \vec{w}_i \sum_{\vec{a}_{-i} \in \cA_{-i}} u_i(a_i, \vec{a}_{-i}) \prod_{j \neq i} \vec{x}_j(a_j) = \vec{w}_i u_i(a_i, \vec{x}_{-i}).
    \end{align*}
    This verifies condition \eqref{eq:global-align} from \Cref{def:weighted_potential}. Thus, it remains to establish the smoothness of the potential, in the sense of \eqref{eq:onesided-smooth}. To this end, we first recall the following simple fact.
    \begin{fact}
        \label{fact:smooth}
        Let $g : \R^d \to \R$ be a twice continuously differentiable function such that the spectral norm of the Hessian $\nabla^2 g$ is upper bounded by some $M > 0$. Then, for any $\vec{x}, \vec{y} \in \R^d$,
        \begin{equation*}
            \| \nabla g(\vec{x}) - \nabla g(\vec{y})\|_2 \leq M \| \vec{x} - \vec{y}\|_2.
        \end{equation*}
    \end{fact}
    Hence, it suffices to bound the operator norm of the Hessian of $\Phi$. This is shown in the following lemma, where we crucially leverage the multilinearity of the (mixed) potential function.
    \begin{lemma}
        It holds that $\|\nabla^2 \Phi \|_2 \leq \phimax \sum_{i=1}^n |\cA_i|$.
    \end{lemma}
    \begin{proof}
        First of all, it is easy to see that for all $i \in [n]$ and $a_i, a_i' \in \cA_i$,
        \begin{equation*}
            \frac{\partial^2 \Phi}{\partial \vec{x}_i(a_i) \partial \vec{x}_i(a_i')} = 0.
        \end{equation*}
        On the other hand, for $i \neq j \in [n]$ and $a_i \in \cA_i$ and $a_j \in \cA_j$ it holds that
        \begin{equation*}
            \frac{\partial^2 \Phi}{\partial \vec{x}_i(a_i) \partial \vec{x}_j(a_j)} = \sum_{\vec{a}_{-i, -j} \in \cA_{-i, -j}} \Phi(a_i, a_j, \vec{a}_{-i, -j}) \prod_{k \neq i, j} \vec{x}_k(a_k);
        \end{equation*}
        here we used the shorthand notation $\vec{a}_{-i, -j} \defeq \prod_{k \neq i, j} a_k$ and $\cA_{-i, -j} \defeq \prod_{k \neq i, j} \cA_k$. Thus, applying the triangle inequality yields that
        \begin{equation}
            \label{eq:phimax}
            \left| \frac{\partial^2 \Phi}{\partial \vec{x}_i(a_i) \partial \vec{x}_j(a_j)} \right| = \left| \sum_{\vec{a}_{-i, -j}} \Phi(a_i, a_j, \vec{a}_{-i, -j}) \prod_{k \neq i, j} \vec{x}_k(a_k) \right| \leq \phimax \sum_{\vec{a}_{-i, -j}} \prod_{k \neq i, j} \vec{x}_k(a_k) = \phimax,
        \end{equation}
        where the final equality holds by the normalization constraint of the induced product distribution. Thus, for any $\Vec{z} \in \R^{d}$, where $d = \sum_{i=1}^n |\cA_i|$, we have that
        \begin{equation*}
            \| \nabla^2 \Phi \Vec{z}\|_2 \leq \phimax \sqrt{d \left( \sum_{r = 1}^d \Vec{z}(r) \right)^2} \leq \phimax \sqrt{d^2 \sum_{r=1}^d \Vec{z}^2(r)} = \phimax d \|\Vec{z}\|_2,
        \end{equation*}
        where we used \eqref{eq:phimax} in the first bound, and Jensen's inequality in the second one. As a result, we have shown that $\|\nabla^2 \Phi\|_2 \leq \phimax \sum_{i=1}^n |\cA_i|$, concluding the proof.
    \end{proof}

    Finally, combining this lemma with \Cref{fact:smooth} we arrive at the following corollary, which concludes the proof of \Cref{proposition:potential}.

    \begin{corollary}
        \label{corollary:onesided}
        Let $L = \frac{1}{2} \phimax \sum_{i=1}^n |\cA_i|$. Then, for any $\Vec{x}, \widetilde{\Vec{x}} \in \prod_{i \in [n]} \Delta(\cA_i)$,
        \begin{equation*}
            \begin{split}
                \Phi(\widetilde{\vec{x}}) \leq \Phi(\vec{x}) + \langle \nabla_{\vec{x}} \Phi(\vec{x}), \widetilde{\vec{x}} - \vec{x} \rangle + L \| \widetilde{\vec{x}} - \vec{x} \|_2^2; \\
                - \Phi(\widetilde{\vec{x}}) \leq - \Phi(\vec{x}) - \langle \nabla_{\vec{x}} \Phi(\vec{x}), \widetilde{\vec{x}} - \vec{x} \rangle + L \| \widetilde{\vec{x}} - \vec{x} \|_2^2.
            \end{split}
        \end{equation*}
    \end{corollary}
\end{proof}

\potentialmonotone*

\begin{proof}
    First of all, since every regularizer $\cR_i$ is $1$-strongly convex with respect to $\|\cdot\|_2$, we obtain from \Cref{def:weighted_potential} that
    \begin{align}
        \Phi(\vec{x}) & \leq \Phi(\widetilde{\vec{x}}) - \langle \nabla_{\vec{x}} \Phi(\vec{x}), \widetilde{\vec{x}} - \vec{x} \rangle + L \sum_{i=1}^n \| \widetilde{\vec{x}}_i - \vec{x}_i \|_2^2 \notag                    \\
                      & \leq \Phi(\widetilde{\vec{x}}) - \langle \nabla_{\vec{x}} \Phi(\vec{x}), \widetilde{\vec{x}} - \vec{x} \rangle + 2L \sum_{i=1}^n D_{\cR_i}(\widetilde{\vec{x}}_i, \vec{x}_i). \label{eq:new-onesided}
    \end{align}
    Moreover, we know from the update rule of \eqref{eq:MD} that for any player $i \in [n]$,
    \begin{equation}
        \label{eq:totalupdate}
        \vec{x}_i^{(t+1)} = \argmax_{\vec{x}_i \in \Delta(\cA_i) } \left\{ \left\langle \nabla_{\vec{x}_i} \Phi(\vec{x}^{(t)}), \vec{x}_i - \vec{x}_i^{(t)} \right\rangle - \frac{1}{\eta} D_{\mathcal{R}_i}(\vec{x}_i, \vec{x}_i^{(t)})  \right\},
    \end{equation}
    where we used the fact that $g(\vec{u}_i) = \nabla_{\vec{x}_i} \Phi(\vec{x})$ (\Cref{def:weighted_potential}). Also, from \eqref{eq:new-onesided} we obtain that
    \begin{align}
        \Phi(\vec{x}^{(t)}) & \leq \Phi(\vec{x}^{(t+1)}) - \left\langle \nabla_{\vec{x}} \Phi(\vec{x}^{(t)}), \vec{x}^{(t+1)} - \vec{x}^{(t)} \right\rangle + 2L \sum_{i=1}^n D_{\mathcal{R}_i} (\vec{x}_i^{(t+1)}, \vec{x}_i^{(t)}). \label{eq:smooth-pot}
    \end{align}
    Now the update rule of \eqref{eq:totalupdate} implies that
    \begin{equation}
        \label{eq:per-agent-bound}
        \left\langle \nabla_{\vec{x}_i} \Phi(\vec{x}^{(t)}), \vec{x_i}^{(t+1)} - \vec{x}_i^{(t)} \right\rangle - \frac{1}{\eta} D_{\mathcal{R}_i} (\vec{x}_i^{(t+1)}, \vec{x}_i^{(t)}) \geq \frac{1}{2\eta} \| \vec{x}_i^{(t+1)} - \vec{x}_i^{(t)} \|_2^2,
    \end{equation}
    for all $i \in [n]$, where we used the $1$-strong convexity of $\cR_i$ with respect to $\|\cdot\|_2$ (quadratic growth). As a result, summing \eqref{eq:per-agent-bound} for all $i \in [n]$ yields that
    \begin{equation*}
        \left\langle \nabla_{\vec{x}} \Phi(\vec{x}^{(t)}), \vec{x}^{(t+1)} - \vec{x}^{(t)} \right\rangle - \frac{1}{\eta} \sum_{i=1}^n D_{\mathcal{R}_i} (\vec{x}_i^{(t+1)}, \vec{x}_i^{(t)}) \geq \frac{1}{2\eta} \sum_{i=1}^n \| \vec{x}_i^{(t+1)} - \vec{x}_i^{(t)} \|_2^2.
    \end{equation*}
    Finally, plugging this bound with $\eta = \frac{1}{2L} $ to \eqref{eq:smooth-pot} and rearranging the terms concludes the proof.
\end{proof}

\regpot*

\begin{proof}
    It is well-known (\emph{e.g.}, see \citep{Shalev-Shwartz12:Online}) that the cumulative regret $\reg^T$ of \eqref{eq:MD} can be bounded as
    \begin{equation}
        \label{eq:regret-MD}
        \reg^T \leq \frac{\Omega}{\eta} + \sum_{t=1}^{T} \| \vec{u}^{(t)}\|_* \|\vec{x}^{(t)} - \vec{x}^{(t-1)}\|,
    \end{equation}
    where $\Omega \defeq \sup_{\vec{x} \in \cX} D_{\cR} (\vec{x}, \vec{x}^{(0)})$, and $\vec{u}^{(1)}, \dots, \vec{u}^{(t)}$ represents the sequence of utility vectors observed by the regret minimizer. Now we know from \Cref{corollary:boundedtraj-pot} that for $\eta = \frac{1}{2L}$,
    \begin{equation*}
        \sum_{t=2}^{T} \| \vec{x}^{(t)} - \vec{x}^{(t-1)}\|_2^2 \leq \sum_{i=1}^n \sum_{t=2}^{T} \| \vec{x}^{(t)} - \vec{x}^{(t-1)}\|_2^2  \leq \frac{2 \phimax}{L}.
    \end{equation*}
    Thus, an application of the Cauchy-Schwarz inequality implies that
    \begin{equation}
        \label{eq:single-traj}
        \sum_{t=2}^{T} \| \vec{x}^{(t)} - \vec{x}^{(t-1)} \|_2 \leq \sqrt{(T-1) \sum_{t=2}^{T} \| \vec{x}^{(t)} - \vec{x}^{(t-1)} \|_2^2} \leq \sqrt{T} \sqrt{\frac{2\phimax}{L}}.
    \end{equation}
    Finally, from \eqref{eq:regret-MD} it follows that for any player $i \in [n]$,
    \begin{equation*}
        \reg_i^T \leq \frac{\Omega_i}{\eta} + \|\Vec{u}_i\|_2 \sum_{t=1}^{T} \|\vec{x}^{(t)} - \vec{x}^{(t-1)}\|_2 = O(\sqrt{T}),
    \end{equation*}
    where we used the notation $\|\vec{u}_i \|_2 \defeq \max_{t \in [T]} \|\Vec{u}_i^{(t)}\|_2$. This concludes the proof.
\end{proof}

\optregpot*

\begin{proof}
    First of all, it is well-known that OMWU on the simplex can be expressed with the following update rule for $t \geq 1$:

    \begin{equation}
        \tag{OMWU}
        \label{eq:OMWU}
        \vec{x}_i^{(t+1)}(a_i) = \frac{\exp\left(2 \eta \vec{u}_i^{(t)}(a_i) - \eta \vec{u}_i^{(t-1)}(a_i)\right)}{\sum_{a_i' \in \cA_i} \exp\left(2 \eta \vec{u}_i^{(t)}(a_i') - \eta \vec{u}^{(t-1)}_i(a_i')\right) \Vec{x}_i^{(t)}(a_i') } \vec{x}_i^{(t)}(a_i),
    \end{equation}
    for any action $a_i \in \cA_i$ and player $i \in [n]$. By convention, $\vec{x}_i^{(0)}$ and $\Vec{x}_i^{(1)}$ are initialized from the uniform distribution over $\cA_i$, for all $i \in [n]$. Now we claim that the update rule of \eqref{eq:OMWU} is equivalent to
    \begin{equation*}
        \vec{x}_i^{(t+1)} = \argmax_{\vec{x}_i \in \Delta(\cA_i) } \left\{ \left\langle 2 \nabla_{\vec{x}_i} \Phi(\vec{x}^{(t)}) - \nabla_{\vec{x}_i} \Phi(\vec{x}^{(t-1)}) , \vec{x}_i - \vec{x}_i^{(t)} \right\rangle - \frac{1}{\eta} D_{\mathcal{R}_i}(\vec{x}_i, \vec{x}_i^{(t)})  \right\},
    \end{equation*}
    assuming that $\cR_i$ is the negative entropy DGF for all $i \in [n]$. By $1$-strong convexity of $\cR_i$, this implies that
    \begin{equation}
        \label{eq:opt-QG}
        \left\langle 2 \nabla_{\vec{x}_i} \Phi(\vec{x}^{(t)}) - \nabla_{\vec{x}_i} \Phi(\vec{x}^{(t-1)}) , \vec{x_i}^{(t+1)} - \vec{x}_i^{(t)} \right\rangle - \frac{1}{\eta} D_{\mathcal{R}_i} (\vec{x}_i^{(t+1)}, \vec{x}_i^{(t)}) \geq \frac{1}{2\eta} \| \vec{x}_i^{(t+1)} - \vec{x}_i^{(t)} \|_2^2.
    \end{equation}
    Moreover, the term $\left| \left\langle \nabla_{\vec{x}_i} \Phi(\vec{x}^{(t)}) - \nabla_{\vec{x}_i} \Phi(\vec{x}^{(t-1)}), \vec{x}_i^{(t+1)} - \vec{x}_i^{(t)} \right\rangle \right|$ can be upper bounded by
    \begin{align}
        &\| \nabla_{\vec{x}_i} \Phi(\vec{x}^{(t)}) - \nabla_{\vec{x}_i} \Phi(\vec{x}^{(t-1)}) \|_2 \| \vec{x}_i^{(t+1)} - \vec{x}_i^{(t)} \|_2 \label{align:cauchy-s}                \\
                                                                                                                                                                         & = \| \vec{u}_i^{(t)} - \vec{u}_i^{(t-1)} \|_2 \| \vec{x}_i^{(t+1)} - \vec{x}_i^{(t)} \|_2 \label{align:pot-def}                                                                  \\
                                                                                                                                                                         & \leq \sqrt{|\cA_i|} \| \vec{u}_i^{(t)} - \vec{u}_i^{(t-1)} \|_\infty \| \vec{x}_i^{(t+1)} - \vec{x}_i^{(t)} \|_2 \label{align:equiv-norms}                                       \\
                                                                                                                                                                         & \leq \sqrt{|\cA_i|} \sum_{j \neq i} \| \vec{x}_j^{(t)} - \vec{x}_j^{(t-1)}\|_2 \| \vec{x}_i^{(t+1)} - \vec{x}_i^{(t)} \|_2 \label{align:claim}                                   \\
                                                                                                                                                                         & \leq \frac{1}{2} \sqrt{|\cA_i|} \sum_{j \neq i} \left( \| \vec{x}_j^{(t)} - \vec{x}_j^{(t-1)}\|^2_2 + \| \vec{x}_i^{(t+1)} - \vec{x}_i^{(t)} \|_2^2 \right) \label{align:penult} \\
                                                                                                                                                                         & = \frac{1}{2} \sqrt{|\cA_i|} \left( (n-1) \| \vec{x}_i^{(t+1)} - \vec{x}_i^{(t)} \|_2^2 + \sum_{j \neq i} \|\vec{x}_j^{(t)} - \vec{x}_j^{(t-1)}\|^2_2 \right), \notag
    \end{align}
    where \eqref{align:cauchy-s} follows from Cauchy-Schwarz; \eqref{align:pot-def} is a consequence of \Cref{def:weighted_potential}; \eqref{align:equiv-norms} is derived from H\"older's inequality (equivalence of norms); \eqref{align:claim} uses \Cref{claim:util-infty}; and \eqref{align:penult} is an application of Young's inequality. Thus, combining this bound with \eqref{eq:opt-QG} yields that
    \begin{align*}
        \left\langle \nabla_{\vec{x}_i} \Phi(\vec{x}^{(t)}), \vec{x_i}^{(t+1)} - \vec{x}_i^{(t)} \right\rangle - \frac{1}{\eta} D_{\mathcal{R}_i} (\vec{x}_i^{(t+1)}, \vec{x}_i^{(t)}) \geq \left( \frac{1}{2\eta} - \frac{1}{2} \sqrt{|\cA_i|} (n-1)  \right) \| \vec{x}_i^{(t+1)} - \vec{x}_i^{(t)} \|_2^2 \\
        - \frac{1}{2} \sqrt{|\cA_i|} \sum_{j \neq i} \| \vec{x}_j^{(t)} - \vec{x}_j^{(t-1)} \|_2^2.
    \end{align*}
    Summing these inequalities for all $i \in [n]$ gives us that
    \begin{align*}
        \left\langle \nabla_{\vec{x}} \Phi(\vec{x}^{(t)}), \vec{x}^{(t+1)} - \vec{x}^{(t)} \right\rangle - \frac{1}{\eta} \sum_{i=1}^n D_{\mathcal{R}_i} (\vec{x}_i^{(t+1)}, \vec{x}_i^{(t)}) & \geq                                                                                                      \\
        \sum_{i=1}^n \left( \frac{1}{2\eta} - \frac{1}{2} \sqrt{|\cA_i|} (n-1)  \right) \| \vec{x}_i^{(t+1)} - \vec{x}_i^{(t)} \|_2^2
                                                                                                                                                                                              & - \frac{1}{2} \left( \sum_{j \neq i} \sqrt{|\cA_j|} \right) \| \vec{x}_i^{(t)} - \vec{x}_i^{(t-1)}\|_2^2.
    \end{align*}
    Thus, from \eqref{eq:new-onesided} for $\eta \leq \frac{1}{2L}$ we conclude that $\Phi(\vec{x}^{(t+1)}) - \Phi(\vec{x}^{(t)})$ can be lower bounded by
    \begin{equation*}
        \sum_{i=1}^n \left( \left( \frac{1}{2\eta} - \frac{1}{2} \sqrt{|\cA_i|} (n-1)  \right) \| \vec{x}_i^{(t+1)} - \vec{x}_i^{(t)} \|_2^2 - \frac{1}{2} \left( \sum_{j \neq i} \sqrt{|\cA_j|} \right) \| \vec{x}_i^{(t)} - \vec{x}_i^{(t-1)}\|_2^2 \right).
    \end{equation*}
    As a result, as long as
    \begin{equation*}
        \eta \leq \min \left\{ \frac{1}{2L}, \frac{1}{2 \sqrt{|\cA_i|} (n-1)} \right\},
    \end{equation*}
    for all $i \in [n]$, the previous bound implies that
    \begin{equation*}
        \Phi(\vec{x}^{(t+1)}) - \Phi(\vec{x}^{(t)}) \geq \sum_{i=1}^n \left( \frac{1}{4\eta} \| \vec{x}_i^{(t+1)} - \vec{x}_i^{(t)} \|_2^2 - \frac{1}{2} \left( \sum_{j \neq i} \sqrt{|\cA_j|} \right) \| \vec{x}_i^{(t)} - \vec{x}_i^{(t-1)}\|_2^2 \right).
    \end{equation*}
    Thus, a telescopic summation leads to the following conclusion.

    \begin{theorem}
        \label{theorem:bounded_traj-opt}
        Suppose that each player $i \in [n]$ employs \eqref{eq:OMWU} with learning rate $\eta > 0$ such that
        \begin{equation*}
            \eta \leq \min \left\{ \frac{1}{2L}, \frac{1}{2 \sqrt{|\cA_i|} (n-1)}, \frac{1}{4 \sum_{j \neq i} \sqrt{|\cA_j|}} \right\},
        \end{equation*}
        for all $i \in [n]$, where $L$ is defined as in \Cref{proposition:potential}. Then,
        \begin{equation*}
            \frac{1}{8\eta} \sum_{i=1}^n \sum_{t=1}^{T-1} \| \vec{x}_i^{(t+1)} - \vec{x}_i^{(t)}\|_2^2 \leq \Phi(\Vec{x}^{(T)}) - \Phi(\Vec{x}^{(1)}) \leq 2 \phimax.
        \end{equation*}
    \end{theorem}
    Finally, \Cref{theorem:bounded_traj-opt} along with the \rvu bound (which holds for \eqref{eq:OMWU} as specified in \Cref{proposition:rvu}) and \Cref{claim:util-infty} conclude the proof of the theorem.
\end{proof}

Next, we proceed with the proof of \Cref{theorem:ratepotential}, the detailed version of which is given below.

\begin{theorem}[Full Version of \Cref{theorem:ratepotential}]
    \label{theorem:ratepotential-full}
    Suppose that each player $i$ employs \eqref{eq:MD} with $\cR_i(\vec{x})$ such that $\nabla \cR_i$ is $G$-Lipschitz continuous with respect to the $\ell_2$-norm, and $\eta = \frac{1}{2L}$, where $L$ is defined as in \Cref{proposition:potential}. Then, after $O(1/\epsilon^2)$ iterations there exists a joint strategy $\vec{x}^{(t)}$  which is an $\epsilon$-approximate Nash equilibrium.
\end{theorem}
For the proof of this theorem we will use the following simple claim.
\begin{claim}
    \label{claim:approx-nash}
    Suppose that each player employs \eqref{eq:MD} with learning rate $\eta > 0$ and regularizer $\cR_i$ such that $\nabla \cR_i$ is $G$-Lipschitz continuous with respect to the $\ell_2$-norm. If $\|\vec{x}^{(t+1)} - \vec{x}^{(t)}\|_2 \leq \epsilon$, it holds that $\vec{x}^{(t)}$ is an
    \begin{equation*}
        \epsilon \left( \frac{G \Omega}{\eta} + \sqrt{|\cA|} \right)
    \end{equation*}
    approximate Nash equilibrium, where $\Omega \defeq \max_{i \in [n]} \sup_{\Vec{x}_i, \Vec{x}_i' \in \Delta(\cA_i)} \| \Vec{x}_i - \vec{x}_i'\|_2$, and $|\cA| \defeq \max_{i \in [n]} |\cA_i|$.
\end{claim}

\begin{proof}
    The argument proceeds similarly to the proof of \Cref{theorem:rate-OMD-full}. First, observe that for each $i \in [n]$ it holds that $\|\vec{x}_i^{(t+1)} - \vec{x}_i^{(t)}\|_{2} \leq \epsilon$. By definition of \eqref{eq:MD} we have that
    \begin{equation*}
        \vec{x}_i^{(t+1)} = \argmax_{\vec{x}_i \in \Delta(\cA_i)} \left\{ \langle \vec{x}_i, \vec{u}_i^{(t)} \rangle - \frac{1}{\eta} D_{\cR_i}(\Vec{x}_i, \Vec{x}_i^{(t)}) \right\}.
    \end{equation*}
    This maximization problem can be equivalently expressed in the following variational inequality form:
    \begin{equation*}
        \left\langle \vec{u}_i^{(t)} - \frac{1}{\eta} \left(\nabla \cR_i ( \vec{x}^{(t+1)}_i) - \nabla \cR_i ( \vec{x}_i^{(t)})\right), \widehat{\vec{x}}_i - \vec{x}_i^{(t+1)} \right\rangle \leq 0, \quad \forall \widehat{\vec{x}}_i \in \Delta(\cA_i),
    \end{equation*}
    for any $i \in [n]$. Thus, it follows that
    \begin{align}
        \langle \vec{u}_i^{(t)}, \widehat{\vec{x}}_i - \vec{x}_i^{(t+1)} \rangle & \leq \frac{1}{\eta} \langle \nabla \cR_i ( \vec{x}^{(t+1)}_i) - \nabla \cR_i ( \vec{x}_i^{(t)}), \widehat{\vec{x}}_i - \vec{x}_i^{(t+1)} \rangle \notag                     \\
                                                                                 & \leq \frac{1}{\eta} \| \nabla \cR_i ( \vec{x}^{(t+1)}_i) - \nabla \cR_i ( \vec{x}_i^{(t)}) \|_2 \| \widehat{\vec{x}}_i - \vec{x}_i^{(t+1)} \|_2 \label{align:cauchyschartz} \\
                                                                                 & \leq \epsilon \frac{G \Omega_i}{\eta}. \label{align:nash-approx}
    \end{align}
    where \eqref{align:cauchyschartz} follows from the Cauchy-Schwarz inequality, and \eqref{align:nash-approx} uses the fact that $\|\nabla \cR_i ( \vec{x}^{(t+1)}_i) - \nabla \cR_i ( \vec{x}_i^{(t)}) \|_2 \leq G \|\vec{x}_i^{(t+1)} - \vec{x}_i^{(t)}\|_{2} \leq \epsilon G$, which follows from the assumption that $\nabla \cR_i$ is $G$-Lipschitz continuous. Also note that \eqref{align:nash-approx} uses the notation $\Omega_i \defeq \sup_{\Vec{x}_i, \Vec{x}_i' \in \Delta(\cA_i)} \| \Vec{x}_i - \vec{x}_i'\|_2$. As a result, we have established that for any player $i \in [n]$ it holds that for any $\widehat{\vec{x}}_i \in \Delta(\cA_i)$,
    \begin{equation}
        \label{eq:approx-nash}
        \langle \vec{u}_i^{(t)}, \vec{x}_i^{(t+1)} \rangle \geq \langle \vec{u}_i^{(t)}, \widehat{\vec{x}}_i \rangle - \epsilon \frac{G \Omega_i}{\eta}.
    \end{equation}
    Moreover, it also follows that
    \begin{equation*}
        \left| \langle \Vec{u}_i^{(t)}, \Vec{x}_i^{(t+1)} - \Vec{x}_i^{(t)} \rangle \right| \leq \| \Vec{u}_i^{(t)} \|_2 \| \Vec{x}_i^{(t+1)} - \Vec{x}_i^{(t)}\|_2 \leq \sqrt{|\cA_i|} \epsilon,
    \end{equation*}
    where we used the fact that $\| \Vec{x}_i^{(t+1)} - \Vec{x}_i^{(t)}\|_2 \leq \epsilon$, and that $\| \Vec{u}_i^{(t)}\|_{\infty} \leq 1$ (by the normalization hypothesis). Plugging the last bound to \eqref{eq:approx-nash} gives us that
    \begin{equation}
        \langle \vec{u}_i^{(t)}, \vec{x}_i^{(t)} \rangle \geq \langle \vec{u}_i^{(t)}, \vec{x}_i^{(t+1)} \rangle - \epsilon \sqrt{|\cA_i|} \geq \langle \vec{u}_i^{(t)}, \widehat{\vec{x}}_i \rangle - \epsilon \frac{G \Omega_i}{\eta} - \epsilon \sqrt{|\cA_i|},
    \end{equation}
    for any $\widehat{\vec{x}}_i \in \Delta(\cA_i)$ and player $i \in [n]$. Finally, the proof follows by definition of approximate Nash equilibrium (\Cref{def:Nash}).
\end{proof}

\begin{proof}[Proof of \Cref{theorem:ratepotential-full}]
    Suppose that $\|\vec{x}^{(t+1)} - \vec{x}^{(t)}\|_2 > \epsilon$ for all $t \in [T]$. \Cref{corollary:boundedtraj-pot} implies that
    \begin{equation}
        4 \eta \phimax \geq \sum_{t=1}^{T-1} \| \Vec{x}^{(t+1)} - \Vec{x}^{(t)}\|_2^2 \geq (T-1) \epsilon^2 \implies T \leq \frac{4\eta \phimax}{\epsilon^2} + 1.
    \end{equation}
    Hence, for $T > \left\lceil \frac{4 \eta \phimax}{\epsilon^2} \right\rceil + 1$ it must be the case that there exists $t \in [T]$ such that $\|\vec{x}^{(t+1)} - \vec{x}^{(t)}\|_2 \leq \epsilon$. Thus, the theorem follows directly from \Cref{claim:approx-nash}.
\end{proof}

\concaverate*

\begin{proof}
    The proof proceeds similarly to that in \citep[Lemma 4]{Birnbaum11:Distributed}. We will first require an auxiliary lemma regarding the following optimization problem:
    \begin{equation}
        \label{eq:aux-opt}
        \begin{split}
            &\text{maximize} \quad g(\vec{x}) - D_{\cR}(\vec{x}, \vec{y}); \\
            &\text{subject to} \quad \vec{x} \in C,
        \end{split}
    \end{equation}
    where $g$ is a concave function on a convex and compact domain $C$.
    \begin{lemma}[\emph{e.g.}, see \citep{Chen93:Convergence}]
        \label{lemma:aux-opt}
        Let $\widehat{\vec{x}}$ be the optimal solution to the optimization problem \eqref{eq:aux-opt}. Then,
        \begin{equation*}
            g(\vec{x}) - D_{\cR}(\vec{x}, \vec{y}) \leq g(\widehat{\vec{x}}) - D_{\cR}(\widehat{\vec{x}}, \vec{y}) - D_{\cR}(\vec{x}, \widehat{\vec{x}}).
        \end{equation*}
    \end{lemma}
    Next, we will apply this lemma to conclude that for all $i \in [n]$ it holds that
    \begin{align}
        \eta \langle \nabla_{\vec{x}_i} \Phi(\vec{x}^{(t)}), \vec{x}_i^* - \vec{x}_i^{(t)} \rangle - D_{\cR_i}(\vec{x}_i^*, \vec{x}_i^{(t)}) \leq \eta \langle \nabla_{\vec{x}_i} \Phi(\vec{x}^{(t)}), \vec{x}_i^{(t+1)} - \vec{x}_i^{(t)} \rangle 
        &- D_{\cR_i}(\vec{x}_i^{(t+1)}, \vec{x}_i^{(t)}) \notag \\
        &- D_{\cR_i}(\vec{x}_i^*, \vec{x}_i^{(t+1)}). \label{eq:5points}
    \end{align}
    Moreover, from \eqref{eq:new-onesided} it follows that
    \begin{align}
        - \Phi(\vec{x}^{(t+1)}) & \leq - \Phi(\vec{x}^{(t)}) - \langle \nabla_{\vec{x}} \Phi(\vec{x}^{(t)}), \vec{x}^{(t+1)} - \vec{x}^{(t)} \rangle + 2L \sum_{i=1}^n D_{\cR_i}(\vec{x}^{(t+1)}_i, \vec{x}^{(t)}_i) \notag                                                             \\
        &\leq\!\! -\Phi(\vec{x}^{(t)})\!-\!\langle \nabla_{\vec{x}} \Phi(\vec{x}^{(t)}), \vec{x}^*\!-\!\vec{x}^{(t)} \rangle\!+\!2L \sum_{i=1}^n D_{\cR_i}(\vec{x}^*, \vec{x}_i^{(t)})\!-\!2L \sum_{i=1}^n D_{\cR_i}(\vec{x}_i^*, \vec{x}_i^{(t+1)}) \label{align:5points} \\
        &\leq -\Phi(\vec{x}^*) + 2L \sum_{i=1}^n D_{\cR_i}(\vec{x}^*, \vec{x}_i^{(t)}) - 2L \sum_{i=1}^n D_{\cR_i}(\vec{x}_i^*, \vec{x}_i^{(t+1)}), \label{align:concavity}
    \end{align}
    where \eqref{align:5points} follows from \eqref{eq:5points} for $\eta = \frac{1}{2L}$, and \eqref{align:concavity} follows by concavity of $\Phi$. As a result, summing \eqref{align:concavity} for all $t \in [T]$ and removing the telescopic terms yields that
    \begin{align}
        \sum_{t=1}^T \left( \Phi(\vec{x}^*) - \Phi(\vec{x}^{(t+1)}) \right) &\leq 2L \sum_{i=1}^n D_{\cR_i}(\vec{x}_i^*, \vec{x}_i^{(1)}) - 2L \sum_{i=1}^n D_{\cR_i}(\vec{x}_i^*, \vec{x}_i^{(T+1)}) \notag \\
        &\leq 2L \sum_{i=1}^n D_{\cR_i}(\vec{x}_i^*, \vec{x}_i^{(1)}), \label{eq:penultimate}
    \end{align}
    where in the last inequality we used the well-known fact that $D_{\cR_i}(\cdot, \cdot) \geq 0$. Finally, \Cref{theorem:monotone} implies that
    \begin{equation*}
        \sum_{t=1}^T \Phi(\vec{x}^*) - \Phi(\vec{x}^{(t+1)}) \geq  T \Phi(\vec{x}^*) - T \Phi(\vec{x}^{(T+1)}),
    \end{equation*}
    and plugging this bound to \eqref{eq:penultimate} concludes the proof after a rearrangement.
\end{proof}

\subsection{Near-Potential Games}
\label{appendix:nearpotential}

In this section we extend some of the results established for (weighted) potential games to \emph{near-potential} games---in the sense of \citet{Candogan13:Dynamics}---by proving \Cref{theorem:nearpot}. Intuitively, a game $\Gamma$ is said to be near-potential if it is \emph{close} to some potential game. While there are many natural ways to measure the distance between two games, here we follow the one suggested by \citet{Candogan13:Dynamics}; namely, the \emph{maximum pairwise difference (MPD)}:

\begin{definition}[Maximum Pairwise Difference (MPD); \citep{Candogan13:Dynamics}]
    \label{definition:MPD}
    Let $\Gamma$ and $\widehat{\Gamma}$ be two (normal-form) games. The \emph{maximum pairwise difference} between $\Gamma$ and $\widehat{\Gamma}$ is defined as
    \begin{equation*}
        d(\Gamma, \widehat{\Gamma}) \triangleq \sup_{i \in [n], \vec{a} \in \cA} | (u_i(a_i, \vec{a}_{-i}) - u_i(a_i', \vec{a}_{-i})) - (\widehat{u}_i(a_i, \vec{a}_{-i}) - \widehat{u}_i(a_i', \vec{a}_{-i}))|,
    \end{equation*}
    where $u_i : \prod_{i \in [n]} \cA_i \to [-1,1] $ and $\widehat{u}_i : \prod_{i \in [n]} \cA_i \to [-1,1]$ are the utility functions associated with player $i$ in $\Gamma$ and $\widehat{\Gamma}$ respectively.
\end{definition}

This definition tacitly posits that the games are compatible in the sense that the set of actions available to each player coincide. MPD captures the difference between two games in terms of the maximum possible utility improvement through unilateral deviations. Different distance measures can be considered without qualitatively altering the rest of the analysis. Armed with \Cref{definition:MPD}, we are ready to introduce the concept of a near-potential game.

\begin{definition}[Near-Potential Game; \citep{Candogan13:Dynamics}]
    A game $\Gamma$ is $\delta$-\emph{near-potential} if there exists a (compatible) potential game $\widehat{\Gamma}$ such that $d(\Gamma, \widehat{\Gamma}) \leq \delta$.
\end{definition}

We remark that \citep{Candogan10:A,Candogan11:Flows} have devised a framework for efficiently finding the nearest potential game to a given game when the distance is measured in terms of the MPD. Specifically, they show that this can be formulated as a convex optimization problem \citep{Candogan10:A,Candogan11:Flows}. Nevertheless, it is clear that following a given update rule in a game $\Gamma$ does not require any sort of knowledge regarding the closest potential game $\widehat{\Gamma}$; the potential function of $\widehat{\Gamma}$ will only be used for the purpose of the analysis. We begin by pointing out the following simple claim.

\begin{claim}
    \label{claim:close-gradient}
    Let $\Gamma$ be any $\delta$-near-potential game with utilities $u_i : \prod_{i \in [n]} \cA_i \to [-1,1]$, for all $i \in [n]$. Moreover, let $\Phi$ be the potential function of a game $\widehat{\Gamma}$ such that $d(\Gamma, \widehat{\Gamma}) = \delta$. Then,
    \begin{equation}
        \label{eq:close-gradient}
        \frac{\partial \Phi(\vec{x})}{\partial \vec{x}_i(a_i)} = u_i(a_i, \vec{x}_{-i}) + e_i(a_i; \vec{x}_{-i}),
    \end{equation}
    where $|e_{i}(a_i, \vec{x}_{-i})| = O(\delta)$, for any $i \in [n]$, $a_i \in \cA_i$, and $\vec{x}_{-i} \in \prod_{j \neq i} \Delta(\cA_j)$.
\end{claim}

Now we are ready to prove \Cref{theorem:nearpot}, the full version of which is given below.

\begin{theorem}[Full Version of \Cref{theorem:nearpot}]
    \label{theorem:nearpot-full}
    Consider a $\delta$-nearly-potential game wherein every player employs \eqref{eq:MD} with learning rate $\eta = \frac{1}{2L}$, where $L$ is a parameter associated with the closest potential game, and regularizer $\cR_i$ which is $1$-strongly convex and $G$-smooth with respect to $\|\cdot\|_2$. Then, there exists a bounded potential function $\Phi$ which increases as long as $\vec{x}^{(t)}$ is not an $O(\sqrt{\delta})$-Nash equilibrium.
\end{theorem}

\begin{proof}
    Since players employ \eqref{eq:MD}, we know that the update rule of each player $i \in [n]$ takes the form
    \begin{equation*}
        \vec{x}_i^{(t+1)} = \argmax_{\vec{x}_i \in \Delta(\cA_i) } \left\{ \left\langle \vec{u}_i^{(t)}, \vec{x}_i - \vec{x}_i^{(t)} \right\rangle - \frac{1}{\eta} D_{\mathcal{R}_i}(\vec{x}_i, \vec{x}_i^{(t)})  \right\}.
    \end{equation*}
    By the $1$-strong convexity of $\cR_i$ with respect to $\|\cdot\|_2$, this implies that
    \begin{equation}
        \label{eq:delta-QG}
        \left\langle \vec{u}_i^{(t)}, \vec{x_i}^{(t+1)} - \vec{x}_i^{(t)} \right\rangle - \frac{1}{\eta} D_{\mathcal{R}_i} (\vec{x}_i^{(t+1)}, \vec{x}_i^{(t)}) \geq \frac{1}{2\eta} \| \vec{x}_i^{(t+1)} - \vec{x}_i^{(t)} \|_2^2.
    \end{equation}
    However, by \Cref{claim:close-gradient} we also know that
    \begin{equation*}
        \vec{u}_i^{(t)} = \nabla_{\vec{x}_i}\Phi(\vec{x}^{(t)}) + \vec{e}_i^{(t)},
    \end{equation*}
    where $\Phi$ is the potential function of the potential game within $\delta$ distance from the original game, and $\vec{e}_i^{(t)}$ is a vector such that $\| \vec{e}_i^{(t)}\|_{\infty} = C \delta$, for some parameter $C > 0$. Thus, combing this fact with \eqref{eq:delta-QG} yields that
    \begin{align}
        \left\langle \nabla_{\vec{x}_i}\Phi(\vec{x}^{(t)}), \vec{x_i}^{(t+1)} - \vec{x}_i^{(t)} \right\rangle - \frac{1}{\eta} D_{\mathcal{R}_i} (\vec{x}_i^{(t+1)}, \vec{x}_i^{(t)}) & \geq \frac{1}{2\eta} \| \vec{x}_i^{(t+1)} - \vec{x}_i^{(t)} \|_2^2 - \langle \vec{e}_i^{(t)}, \vec{x}_i^{(t+1)} - \vec{x}_i^{(t)} \rangle \notag       \\
                                                                                                                                                                                      & \geq \frac{1}{2\eta} \| \vec{x}_i^{(t+1)} - \vec{x}_i^{(t)} \|_2^2 - \| \vec{e}_i^{(t)} \|_{\infty} \| \vec{x}_i^{(t+1)} - \vec{x}_i^{(t)} \|_1 \notag \\
                                                                                                                                                                                      & \geq \frac{1}{2\eta} \| \vec{x}_i^{(t+1)} - \vec{x}_i^{(t)} \|_2^2 - 2 C \delta, \label{align:last-delta}
    \end{align}
    where we used the fact that the $\ell_1$ diameter of $\Delta(\cA_i)$ is $2$ in the last derivation. Summing \eqref{align:last-delta} for all $i \in [n]$ gives us that
    \begin{equation*}
        \left\langle \nabla_{\vec{x}}\Phi(\vec{x}^{(t)}), \vec{x}^{(t+1)} - \vec{x}^{(t)} \right\rangle - \frac{1}{\eta} \sum_{i=1}^n D_{\mathcal{R}_i} (\vec{x}_i^{(t+1)}, \vec{x}_i^{(t)}) \geq \frac{1}{2\eta} \sum_{i=1}^n \| \vec{x}_i^{(t+1)} - \vec{x}_i^{(t)} \|_2^2 - 2 C n \delta.
    \end{equation*}
    Thus, using the smoothness condition of \eqref{eq:new-onesided} with $\eta = \frac{1}{2L}$ implies that
    \begin{equation*}
        \Phi(\vec{x}^{(t+1)}) - \Phi(\vec{x}^{(t)}) \geq \frac{1}{2\eta} \sum_{i=1}^n \| \vec{x}_i^{(t+1)} - \vec{x}_i^{(t)} \|_2^2 - 2 C n \delta.
    \end{equation*}
    As a result, we conclude that $\Phi$ increases as long as
    \begin{equation*}
        \| \vec{x}^{(t+1)} - \vec{x}^{(t)} \|_2 \geq 2 \sqrt{\eta C n \delta},
    \end{equation*}
    and the theorem follows directly from \Cref{claim:approx-nash}.
\end{proof}

\subsection{Fisher Markets}
\label{appendix:Fisher}

In this subsection we illustrate how the framework we developed in \Cref{section:potential} unifies the work of \citet{Birnbaum11:Distributed} related to distributed dynamics in Fisher's classical market model. While the following exposition focuses on the linear counterpart of Fisher's model, we refer the interested reader to the work of \citet{Birnbaum11:Distributed} for an elegant extension to the \emph{spending constraints model}. Regarding the motivation for studying distributed dynamics in markets, let us quote from the work of \citet{Birnbaum11:Distributed}:
\begin{quote}
    \emph{``Algorithmic results in a centralized model of computation do not directly address the question of market dynamics: how might agents interacting in a market arrive at an equilibrium? Here, the quest is for simple and distributed algorithms that are guaranteed to converge fast. Such distributed algorithms are especially applicable when the agents involved are automated, and one has to prescribe a particular protocol for them to follow''}
\end{quote}

In this context, we commence by recalling the underlying model. The exposition in the sequel follows that in \citep{Birnbaum11:Distributed}. In the linear Fisher's market model there are $n$ agents (bidders) and $m$ (infinitely) \emph{divisible} goods. It is assumed---without any loss of generality---that there is a unit supply from each good. Every agent has an overall \emph{budget} $B_i$, circumscribing its buying power. The goal of each agent $i \in [n]$ is to maximize its own utility, which, for a given \emph{allocation vector} $\vec{x} \in \R_{\geq 0}^{n \times m}$, is defined as $\sum_{j} \vec{x}_i(j) \vec{u}_i(j)$, where $\vec{u}_i(j)$ represents the \emph{utility} of that agent for a unit of good $j$.

\paragraph{Equilibrium Conditions} Consider an allocation vector $\vec{x} \in \R_{\geq 0}^{n \times m}$ and a \emph{price vector} $\vec{p} \in \R_{\geq 0}^m$. The pair $(\vec{x}, \vec{p})$ is said to be an \emph{equilibrium} if the following conditions are met.
\begin{enumerate}
    \item \emph{Buyer optimality}: Each agent $i \in [n]$ maximizes its own utility subject to the budget constraints. Formally, for the given price vector $\vec{p}$, it has to be that the allocation vector $\vec{x}$ maximizes the following (linear) program:
          \begin{align*}
              \text{maximize} \quad   & \sum_{j \in [m]} \vec{u}_i(j) \vec{x}_i(j),        \\
              \text{subject to} \quad & \sum_{j \in [m]} \vec{p}(j) \vec{x}_i(j) \leq B_i; \\
                                      & \vec{x}_{i}(j) \geq 0, \quad \forall j \in [m].
          \end{align*}
    \item \emph{Market clearance}: It must be the case that $\sum_{i \in [n]} \vec{x}_i(j) = 1$, for all $j \in [m]$. \label{item:market-clear}
\end{enumerate}

\paragraph{Convex Formulation of Equilibria} Equilibria in linear Fisher markets can be determined via the seminal \emph{Eisenberg-Gale convex program} \citep{Eisenberg59:Consensus}. As in \citep{Birnbaum11:Distributed}, here we focus on an alternative formulation due to \citet{Shmyrev09:An}, which is given below.
\begin{equation}
    \label{eq:Smyrev}
    \begin{split}
        \text{maximize} \quad &\sum_{i, j} \vec{b}_i(j) \log \vec{u}_i(j) - \sum_{j \in [m]} \vec{p}(j) \log \vec{p}(j), \\
        \text{subject to} \quad &\sum_{i \in [n]} \vec{b}_i(j) = \vec{p}(j), \quad \forall j \in [m]; \\
        &\sum_{j \in [m]} \vec{b}_i(j) = B_i, \quad \forall \in [n]; \\
        &\vec{b}_i(j) \geq 0, \quad \forall i \in [n], j \in [m].
    \end{split}
\end{equation}
In this convex program each variable $\vec{b}_i(j)$ represents the amount spent by agent $i$ on good $j$, so that for a given solution to \eqref{eq:Smyrev}, the induced allocation $\vec{x}_i(j)$ is given by $\vec{b}_i(j)/\vec{p}(j)$. Observe that for any feasible solution to \eqref{eq:Smyrev} the markets always clears in the sense of \Cref{item:market-clear}. Moreover, \citet{Shmyrev09:An} showed that the optimal solution of \eqref{eq:Smyrev} is such that each buyer is allocated an optimal bundle of goods, guaranteeing both equilibrium conditions for the associated Fisher market. In the sequel, and for the sake of simplicity, we will let $B_i = 1$.

In this context, we will relate linear Fisher markets with the setting we presented in \Cref{section:potential}. To this end, let $\Phi(\vec{b})$ be the objective function of \eqref{eq:Smyrev}:

\begin{equation*}
    \Phi(\vec{b}) = \sum_{i, j} \vec{b}_i(j) \log \vec{u}_i(j) - \sum_{j \in [m]} \vec{p}(j) \log \vec{p}(j) = \sum_{i, j} \vec{b}_i(j) \log \left( \frac{\vec{u}_i(j)}{\vec{p}(j)} \right),
\end{equation*}
where we used the fact that $\vec{p}(j) = \sum_{i \in [n]} \vec{b}_i(j)$, which in turn follows by the feasibility constraints of \eqref{eq:Smyrev}. Now a direct calculation shows that for any $i \in [n]$, each component $j \in [m]$ on the corresponding gradient is such that
\begin{equation*}
    \left( \nabla \Phi(\vec{b}) \right)_{i, j} = \log \left( \frac{\vec{u}_i(j)}{\vec{p}(j)} \right) - 1.
\end{equation*}
Thus, $\Phi$ is not always smooth on the feasible region
\begin{equation*}
    S \defeq \left\{ \vec{b} \in \R^{n \times m} : \sum_{j \in [m]} \vec{b}_i(j) = B_i, \forall i \in [n], \vec{b}_i(j) \geq 0, \forall (i, j) \in [n] \times [m] \right\}.
\end{equation*}
Nevertheless, \citep[Lemma 7]{Birnbaum11:Distributed} establishes that $\Phi$ satisfies the weaker ``one-sided'' smoothness condition of \Cref{def:weighted_potential} with respect to the Kullback-Leibler divergence. Moreover, let $g : x \mapsto \log x - 1$ be a monotone transformation, so that
\begin{equation}
    \label{eq:per-utility}
    g^{-1} \left( ( \nabla \Phi(\vec{b}) )_{i, j} \right) = \frac{\vec{u}_i(j)}{\vec{p}(j)} = \frac{\vec{u}_i(j) \vec{x}_i(j)}{\vec{b}_i(j)}.
\end{equation}
But the latter expression can be thought of as the utility that agent $i$ received from item $j$ per fraction of budget invested to $j$. As a result, if every agent employs \eqref{eq:MD} with negative entropy DGF, monotone transformation $g : x \mapsto \log x - 1$, and utility vector the feedback suggested by \eqref{eq:per-utility}, it can be shown that bids are updated using the following update rule:
\begin{equation}
    \tag{PR}
    \label{eq:PR}
    \vec{b}^{(t+1)}_i(j) = \frac{\vec{u}_i(j) \vec{x}_i^{(t)}(j) }{\sum_{k \in [m]} \vec{u}_i(k) \vec{x}^{(t)}_i(k)} B_i.
\end{equation}
These dynamics are known us \emph{Proportional Response}, and they were introduced by \citet{Wu07:Proportional} (see also \citep{Zhang11:Proportional}). \eqref{eq:PR} can be seen as typical multiplicative weights update after applying the monotone transformation $g$. In this way, these distributed dynamics is tantamount to optimizing Smyrev's convex program \eqref{eq:Smyrev}, which converges to an equilibrium with $O(1/T)$ rate by concavity of $\Phi$; see \Cref{proposition:concaverate}.

An interesting direction would be to incorporate into this framework further and more general market models such as the \emph{Arrow-Debreu (exchange)} version; \emph{e.g.}, see \citep{Panageas:Combinatorial} for some recent developments related in spirit to the approach we presented in this section.

\section{Proofs from \texorpdfstring{\Cref{section:continuous}}{Section 5}}

In this section we provide the proofs omitted from \Cref{section:continuous}. From a technical standpoint, following~\citep{Anagnostides22:Frequency}, we will rely on
a fundamental tool from signal processing and control theory; namely, the $Z$\emph{-transform}. Recall that the (\emph{bilateral}) $Z$-transform of a sequence $(\Vec{x}^{(t)})$ in $\R^d$ is defined as
\begin{equation}
    \label{eq:Z-transform}
    \Vec{X}(z) \defeq \mathcal{Z} \left\{ \Vec{x}^{(t)} \right\} = \sum_{t=-\infty}^{\infty} \Vec{x}^{(t)} z^{-t},
\end{equation}
where $z \in \C^*$ is assumed to be in the \emph{region of convergence}:
\begin{equation*}
    \text{ROC} \defeq \left\{ z \in \C^* : \left| \sum_{t=-\infty}^{\infty} \Vec{x}^{(t)} z^{-t} \right| < \infty \right\}.
\end{equation*}

Observe that we define the $Z$-transform coordinate-wise. Our analysis will leverage the following well-known properties which follow directly from the definition of \eqref{eq:Z-transform}.

\begin{property}[Linearity]
    \label{property:linearity}
    Let $(\vec{x}^{(t)})_{t=-\infty}^{\infty}$ and $(\vec{y}^{(t)})_{t=-\infty}^{\infty}$ be sequences in $\R^d$. Then,
    \begin{equation*}
        \mathcal{Z} \left\{ \vec{x}^{(t)} + \vec{y}^{(t)} \right\} = \mathcal{Z} \left\{ \vec{x}^{(t)} \right\} + \mathcal{Z} \left\{ \vec{y}^{(t)} \right\}.
    \end{equation*}
\end{property}

\begin{property}[Time-Delay Property]
    \label{property:time-delay}
    Let $(\vec{x}^{(t)})_{t = -\infty}^{\infty}$ be a sequence in $\R^d$, and some $t_0 \in \R$. Then, it holds that
    \begin{equation*}
        \mathcal{Z} \left\{ \vec{x}^{(t - t_0)} \right\} = z^{-t_0} \mathcal{Z} \left\{ \vec{x}^{(t)} \right\}.
    \end{equation*}
\end{property}

We are now ready to prove \Cref{theorem:two_player-convergence}, the detailed version of which is given below.

\begin{theorem}[Detailed Version of \Cref{theorem:two_player-convergence}]
    \label{theorem:two_player-convergence-full}
    Let $\mat{A}$ and $\mat{B}$ be square and full-rank $d \times d$ matrices, and $\gamma \defeq \rho(\mat{A}^\top \mat{B})$, where $\rho(\cdot)$ denotes the spectral radius. If the matrix $\mat{A}^\top \mat{B}$ has strictly negative (real) eigenvalues, it holds that for any learning rate $\eta \leq \frac{1}{2\sqrt{\gamma}}$ \eqref{eq:OGD} converges with linear rate to an equilibrium.
\end{theorem}

\begin{proof}
    First, observe that $\nabla_{\vec{x}}(\vec{x}^\top \mat{A} \vec{y}) = \mat{A} \vec{y}$ and $\nabla_{\vec{y}} (\vec{x}^\top \mat{B} \vec{y}) = \mat{B}^\top \vec{x}$. Thus, \eqref{eq:OGD} can be expressed as the following linear dynamical system:

    \begin{equation}
        \label{eq:linear-OGD}
        \begin{split}
            \vec{x}^{(t+1)} &= \vec{x}^{(t)} + 2\eta \mat{A} \vec{y}^{(t)} - \eta \mat{A} \vec{y}^{(t-1)}; \\
            \vec{y}^{(t+1)} &= \vec{y}^{(t)} + 2\eta \mat{B}^\top \vec{x}^{(t)} - \eta \mat{B}^\top \vec{x}^{(t-1)}.
        \end{split}
    \end{equation}

    To analyze its convergence properties, we will transfer \eqref{eq:linear-OGD} to the $z$-space. To this end, let $\vec{X}(z)$ and $\vec{Y}(z)$ be the $Z$-transform of the sequence $(\vec{x}^{(t)})$ and $(\vec{y}^{(t)})$ respectively; here it is tacitly assumed that $z$ belongs to the region of convergence. Thus, using linearity (\Cref{property:linearity}) and the time-delay property (\Cref{property:time-delay}), it follows from \eqref{eq:linear-OGD} that

    \begin{equation}
        \label{eq:z-OGD}
        \begin{split}
            z \vec{X}(z) &= \vec{X}(z) + 2\eta \mat{A} \vec{Y}(z) - \eta z^{-1} \mat{A} \vec{Y}(z); \\
            z \vec{Y} &= \vec{Y}(z) + 2\eta \mat{B}^\top \vec{X}(z) - \eta z^{-1} \mat{B}^\top \vec{X}(z).
        \end{split}
    \end{equation}
    Note that we also used the fact that $\mathcal{Z} \left\{ \mat{A} \Vec{y}^{(t)} \right\} = \mat{A} \vec{Y}(z)$ and $\mathcal{Z} \left\{ \mat{B}^\top \Vec{x}^{(t)} \right\} = \mat{B}^\top \vec{X}(z)$, which follow from linearity of the $Z$-transform. In this way, we have transferred the \emph{difference equation} \eqref{eq:linear-OGD} to the \emph{algebraic equation} \eqref{eq:z-OGD}. From the latter algebraic system, we may uncouple these equations to conclude that
    \begin{align}
            z^2 (z-1)^2 \vec{X}(z) = (2\eta z - \eta) \mat{A} z (z-1) \vec{Y}(z) &= (2\eta z - \eta) \mat{A} (2\eta z - \eta) \mat{B}^\top \vec{X}(z) \notag \\
            &= \eta^2 (2z - 1)^2 \mat{A} \mat{B}^\top \vec{X}(z); \label{eq:uncoupled-1}
    \end{align}
    and
    \begin{align}
    z^2 (z-1)^2 \vec{Y}(z) = (2\eta z - \eta) \mat{B}^\top z (z-1) \vec{X}(z) &= (2\eta z - \eta) \mat{B}^\top (2\eta z - \eta) \mat{A} \vec{Y}(z) \notag \\
    &= \eta^2 (2z - 1)^2 \mat{B}^\top \mat{A} \vec{Y}(z). \label{eq:uncoupled-2}   
    \end{align}
    As a result, we have derived the \emph{characteristic equation} of the dynamical system \eqref{eq:linear-OGD}:

    \begin{proposition}
        \label{proposition:characteristic}
        The characteristic equation of \eqref{eq:linear-OGD} can be expressed as
        \begin{equation}
            \label{eq:charact-det}
            \det(\mat{I}_d (z(z-1))^2 - \eta^2 (2z -1)^2 \mat{A}^\top \mat{B}) = 0.
        \end{equation}
        In particular, if $\chi(z)$ is the characteristic polynomial of matrix $\mat{A}^\top \mat{B}$, it follows that $z \in \C^*$ satisfies \eqref{eq:charact-det} if and only if
        \begin{equation}
            \label{eq:charact-chi}
            \chi\left( \left( \frac{1}{\eta} \frac{z(z-1)}{2z - 1} \right)^2 \right) = 0.
        \end{equation}
    \end{proposition}
    Note that we used the fact that the matrices $\mat{A} \mat{B}^\top$ and $\mat{B}^\top \mat{A}$ have identical spectrum, which in turn follows since $\mat{A}$ and $\mat{B}$ are assumed to be non-singular, as well as the property $\det(\mat{M}) = \det(\mat{M}^\top)$ in order to derive \eqref{eq:charact-det} from \eqref{eq:uncoupled-1} and \eqref{eq:uncoupled-2}. To see the second part of the claim in \Cref{proposition:characteristic}, first observe that $z = \frac{1}{2}$ is never a solution to \eqref{eq:charact-det}. Thus, from the multilinearity of the determinant,
    \begin{equation*}
        \det(\mat{I}_d (z(z-1))^2 - \eta^2 (2z -1)^2 \mat{A}^\top \mat{B}) = 0 \iff \det \left( \left( \frac{1}{\eta} \frac{z(z-1)}{2z - 1} \right)^2 \mat{I}_d - \mat{A}^\top \mat{B} \right) = 0.
    \end{equation*}
    As a result, the equivalence asserted in \eqref{eq:charact-chi} follows by recalling that $\chi(\lambda) = \det(\lambda \mat{I}_d - \mat{A}^\top \mat{B})$ since $\chi(\lambda)$ is the characteristic polynomial of matrix $\mat{A}^\top \mat{B}$.

    Having derived the characteristic equation of the dynamical system (\Cref{proposition:characteristic}), it remains to derive its solutions. To do this, let $- \lambda$ be such that $\chi(-\lambda) = 0$; by assumption, we know that $\lambda \in \R_{> 0}$. Now this eigenvalue induces solutions of the following form:
    \begin{equation}
        \label{eq:quadratic-pm}
        \left( \frac{1}{\eta} \frac{z(z-1)}{2z - 1} \right)^2 = - \lambda \iff \frac{z(z-1)}{2z - 1} = \pm \eta \sqrt{\lambda} j \iff
        \begin{cases}
            z^2 + z(-1 - 2 \eta \sqrt{\lambda} j) + \eta \sqrt{\lambda}j = 0; \\
            z^2 + z(-1 + 2 \eta \sqrt{\lambda} j) - \eta \sqrt{\lambda}j = 0.
        \end{cases}
    \end{equation}
    where $j \in \C$ denotes the imaginary unit. Now observe that $z^2 + z(-1 - 2\eta \sqrt{\lambda} j) + \eta \sqrt{\lambda} j = 0 \iff \Bar{z}^2 + \Bar{z}(-1 + 2\eta \sqrt{\lambda} j) - \eta \sqrt{\lambda} j = 0$, where $\Bar{z}$ denotes the complex conjugate of $z$. Hence, it suffices to solve only the first equation since their solutions are equivalent in terms of magnitude---in particular, they are complex conjugates. Thus, we obtain the following solutions:
    \begin{equation}
        z_{\pm} = \frac{1 + 2\eta \sqrt{\lambda} j \pm \sqrt{1 - 4\eta^2 \lambda}}{2}.
    \end{equation}
    Moreover, by assumption we know that $\eta \leq \frac{1}{2\sqrt{\gamma}}$, which in turn implies that $\eta \leq \frac{1}{2\sqrt{\lambda}} \iff 1 - 4\eta^2 \lambda \geq 0$; this holds by definition of $\gamma := \rho(\mat{A}^\top \mat{B})$. Thus, it follows that
    \begin{equation*}
        \begin{split}
            |z_+|^2 = \frac{1}{4} \left( (1 + \sqrt{1 - 4\eta^2 \lambda})^2 + (2\eta \sqrt{\lambda})^2 \right) = \frac{1}{2} \left( 1 + \sqrt{1 - 4\eta^2 \lambda} \right) < 1; \\
            |z_-|^2 = \frac{1}{4} \left( (1 - \sqrt{1 - 4\eta^2 \lambda})^2 + (2\eta \sqrt{\lambda})^2 \right) = \frac{1}{2} \left( 1 - \sqrt{1 - 4\eta^2 \lambda} \right) < 1.
        \end{split}
    \end{equation*}
    This implies that all of the solutions to the characteristic equation \eqref{eq:charact-chi} lie within the unit circle on the complex plane. Thus, the theorem follows from the fundamental theorem of linear dynamical systems.
\end{proof}

\inefficiency*

\begin{proof}
    We consider the game described with the following matrices:
    \begin{equation}
        \label{eq:game-inefficiency}
        \mat{A} \defeq
        \begin{bmatrix}
            1  & - 2 \\
            -1 & 1
        \end{bmatrix};
        \mat{B} \defeq
        \begin{bmatrix}
            1 & 1  \\
            1 & -1
        \end{bmatrix}.
    \end{equation}
    We also assume that the constraints sets are such that $\cX = \ball_1(\vec{0}, R)$ and $\cY = \ball_1(\vec{0}, R)$, for a sufficiently large radius $R > 0$. Here $\ball_1$ denotes the closed $\ell_1$-ball; we only use the $\ell_1$-norm for mathematical convenience. We will first consider the \emph{unconstrained dynamics}. In particular, we will show that the condition of \Cref{theorem:two_player-convergence} is met. Indeed, a direct calculation reveals that
    \begin{equation*}
        \mat{A}^\top \mat{B} =
        \begin{bmatrix}
            0  & 2  \\
            -1 & -3
        \end{bmatrix}.
    \end{equation*}
    Now the characteristic equation of $\mat{A}^\top \mat{B}$ reads $\lambda(\lambda+3) + 2 = 0 \iff \lambda_1 = -2, \lambda_2 = -1$. Thus, $\mat{A}^\top \mat{B}$ has negative (real) eigenvalues. As a result, \Cref{theorem:two_player-convergence} implies that for a sufficiently small learning rate \eqref{eq:OGD} converges under any initialization---assuming unconstrained domains (see \Cref{fig:efficiency}). Moreover, the following claim characterizing its limit points is immediate.

    \begin{claim}
        Suppose that \eqref{eq:linear-OGD} converges to a point $(\vec{x}^{(\infty)}, \vec{y}^{(\infty)})$. Then, it holds that $\mat{A} \vec{y}^{(\infty)} = \vec{0}$ and $\mat{B}^\top \vec{x}^{(\infty)} = \vec{0}$.
    \end{claim}

    In particular, given that the matrices $\mat{A}$ and $\mat{B}$ are full-rank, this claim implies that $\vec{x}^{(\infty)} = \vec{0}$ and $\vec{y}^{(\infty)} = \vec{0}$. Moreover, for a sufficiently large $R > 0$ we know that the projected dynamics on $\cX$ and $\cY$ (respectively) will coincide with the unconstrained dynamics. As a result, we have shown that \eqref{eq:OGD} constrained on $\cX$ and $\cY$ will converge to $(\vec{0}, \vec{0})$, which is clearly an equilibrium point. However, there exists a much more efficient equilibrium:

    \begin{claim}
        Let $\vec{x}^* = (R, 0) \in \ball_1(\vec{0}, R) $ and $\vec{y}^* = (R, 0) \in \ball_1(\vec{0}, R)$. Then, $(\vec{x}^*, \vec{y}^*)$ is an equilibrium of the game \eqref{eq:game-inefficiency}.
    \end{claim}

    \begin{proof}
        When player $\cY$ plays $\vec{y}^*$ it follows that $\mat{A} \vec{y}^* = (R, -R)$. Thus, it follows that $\vec{x}^\top \mat{A} \vec{y}^* \leq \|\vec{x}\|_1 \| \mat{A} \vec{y}^*\|_{\infty} \leq R^2 = (\Vec{x}^*)^\top \mat{A} \Vec{y}^*$, for any $\vec{x} \in \ball_1(\vec{0}, R)$. Thus, $\vec{x}^*$ is indeed a best response to $\vec{y}^*$. Similarly, when player $\cX$ plays $\Vec{x}^*$ it follows that $\mat{B}^\top \Vec{x}^* = (R, R)$. Thus, we conclude that $(\Vec{x}^*)^\top \mat{B}\Vec{y} \leq \|\Vec{y}\|_1 \| \mat{B}^\top \Vec{x}^* \|_{\infty} \leq R^2 = (\Vec{x}^*)^\top \mat{B} \Vec{y}^*$, for any $\Vec{y} \in \ball_1(\Vec{0}, R)$. This means that $\Vec{y}^*$ is indeed a best response to $\Vec{x}^*$, verifying our claim.
    \end{proof}
    Finally, we have that $\sw(\vec{x}^*, \vec{y}^*) = (\vec{x}^*)^\top \mat{A} \vec{y}^* + (\vec{x}^*)^\top \mat{B} \vec{y}^* = 2R^2$. This concludes the proof.
\end{proof}

\robustness*

\begin{proof}
    Fix any $\epsilon > 0$. We consider the game $(\mat{A}, \mat{B})$ described with the following payoff matrices:
    \begin{equation}
        \label{eq:robustness}
        \mat{A} \defeq
        \begin{bmatrix}
            1 & 0          \\
            0 & \epsilon/2
        \end{bmatrix};
        \mat{B} \defeq
        \begin{bmatrix}
            -1 & 0          \\
            0  & \epsilon/2
        \end{bmatrix}.
    \end{equation}
    Observe that $\|\mat{A} + \mat{B}\|_F = \epsilon$.\footnote{Recall that the Frobenius norm of a matrix $\mat{M}$ is defined as $\| \mat{M} \|_F \defeq \| (\mat{M})^\flat \|_2$, where $(\mat{M})^\flat$ is the standard vectorization of $\mat{M}$.} Let us first argue about convergence of \eqref{eq:OGD} in the game $(\mat{A}, -\mat{A})$. To this end, observe that the matrix $\mat{A}^\top \mat{A}$ has only positive eigenvalues. Thus, we know from \Cref{theorem:two_player-convergence} that \eqref{eq:OGD} converges for $\eta \leq \frac{1}{2}$ since the spectral radius of $\mat{A}^\top \mat{A}$ is $1$. On the other hand, for the game $(\mat{A}, \mat{B})$ it holds that the matrix $\mat{A}^\top \mat{B}$ has a positive eigenvalue; namely, $\lambda = \epsilon^2/4$. But from \Cref{proposition:characteristic} it can be shown that this implies that the characteristic equation of the associated dynamical system has a solution $z \in \C^*$ with $|z| > 1$. In turn, this implies that \eqref{eq:OGD} will diverge under any non-trivial initialization; see \Cref{fig:robustness} for a graphical illustration and a discussion about this phenomenon.
\end{proof}

Next, we proceed with the proof of \Cref{theorem:one-many-convergence}, the detailed version of which is recalled below.

\begin{theorem}[Detailed Version of \Cref{theorem:one-many-convergence}]
    \label{theorem:one-many-convergence-full}
    Let $\{\mat{A}_{1, j}\}_{j = 2}^n$ and $\{ \mat{A}_{j, 1} \}_{j=2}^n$ be square matrices such that $\det(\mat{M}) \neq 0$, where $\mat{M} \defeq \sum_{j \neq 1} \mat{A}_{1, j} \mat{A}_{j, 1}$. Moreover, let $\gamma \defeq \rho(\mat{A})$, where $\rho(\cdot)$ denotes the spectral radius. If $\mat{M}$ has strictly negative (real) eigenvalues, it holds that for any learning rate $\eta \leq \frac{1}{2\sqrt{\gamma}}$ \eqref{eq:OGD} converges with linear rate to an equilibrium.
\end{theorem}

\begin{proof}
    First of all, recall that we have that $u_1(\vec{x}) = \sum_{j \neq 1} \vec{x}_1^\top \mat{A}_{1, j} \vec{x}_j$, while $u_j(\vec{x}) = \vec{x}_j^\top \mat{A}_{j, 1} \vec{x}_1$ for $j \neq 1$. Thus, it follows that $\nabla_{\vec{x}_1} u_1(\vec{x}) = \sum_{j \neq 1} \mat{A}_{1, j} \vec{x}_j$, and $\nabla_{\vec{x}_j} u_j(\vec{x}) = \mat{A}_{j, 1} \vec{x}_1$ for $j \neq 1$. As a result, \eqref{eq:OGD} can be cast as
    \begin{equation}
        \label{eq:1-many}
        \begin{split}
            \vec{x}_1^{(t+1)} &= \vec{x}_1^{(t)} + 2\eta \sum_{j \neq 1} \mat{A}_{1, j} \vec{x}^{(t)}_j - \eta \sum_{j \neq 1} \mat{A}_{1, j} \vec{x}^{(t-1)}_j; \\
            \vec{x}_j^{(t+1)} &= \vec{x}_j^{(t)} + 2 \eta \mat{A}_{j, 1} \vec{x}^{(t)}_1 - \eta \mat{A}_{j, 1} \vec{x}_1^{(t-1)}, \forall j \neq 1.
        \end{split}
    \end{equation}
    Transferring these dynamics to the $z$-space yields that
    \begin{equation}
        \label{eq:1-many-z}
        \begin{split}
            z \vec{X}_1(z) &= \vec{X}_1(z) + 2\eta \sum_{j \neq 1} \mat{A}_{1, j} \vec{X}_j(z) - \eta z^{-1} \sum_{j \neq 1} \mat{A}_{1, j} \vec{X}_j(z); \\
            z \vec{X}_j(z) &= \vec{X}_j(z) + 2 \eta \mat{A}_{j, 1} \vec{X}_1(z) - \eta z^{-1} \mat{A}_{j, 1} \vec{X}_1(z), \forall j \neq 1.
        \end{split}
    \end{equation}
    From the second equation it follows that $(z^2 - z) \vec{X}_j(z) = \eta (2z - 1) \mat{A}_{j, 1} \vec{X}_1(z)$, for all $j \neq 1$. Thus, plugging this to the first equation of \eqref{eq:1-many-z} yields that
    \begin{align*}
        z^2 (z-1)^2 \vec{X}_1(z) = \eta^2 (2z - 1)^2 \sum_{j \neq 1} \mat{A}_{1, j} \mat{A}_{j, 1} \vec{X}_1(z) \\
        \iff \left( (z(z-1))^2 \mat{I}_d - \eta^2 (2z-1)^2 \sum_{j \neq 1} \mat{A}_{1, j} \mat{A}_{j, 1} \right) \vec{X}_1(z) = 0.
    \end{align*}
    From this, it is possible to conclude the characteristic equation of the associated dynamical system:

    \begin{proposition}
        The characteristic equation of the dynamical system \eqref{eq:1-many} can be expressed as
        \begin{equation*}
            \det \left( (z(z-1))^2 \mat{I}_d - \eta^2 (2z-1)^2 \sum_{j \neq 1} \mat{A}_{1, j} \mat{A}_{j, 1} \right) = 0.
        \end{equation*}
    \end{proposition}
    Observe that if all the poles of $\vec{X}_1(z)$ lie within the unit circle of the complex plane, the same holds for each $\Vec{X}_j(z)$, for all $j \neq 1$, modulo at most a single pole at $z = 1$. This follows given that $(z^2 - z) \vec{X}_j(z) = \eta (2z - 1) \mat{A}_{j, 1} \vec{X}_1(z)$, for all $j \neq 1$. Finally, the rest of the theorem follows analogously to \Cref{theorem:two_player-convergence-full}, while it is direct to verify that, assuming convergence, the limit points satisfy the equilibrium conditions.
\end{proof}

\begin{remark}
    \label{remark:game-matrix}
    More broadly, the convergence of \eqref{eq:OGD} can be determined in terms of the spectrum of the matrix
    \begin{equation*}
        \mat{M} \defeq
        \begin{pmatrix}
            \mat{0}_{d}    & \mat{A}_{1, 2} & \mat{A}_{1,3}  & \dots  & \mat{A}_{1, n} \\
            \mat{A}_{2, 1} & \mat{0}_{d}    & \mat{A}_{2, 3} & \dots  & \mat{A}_{2, n} \\
            \vdots         & \vdots         & \vdots         & \ddots & \vdots         \\
            \mat{A}_{n, 1} & \mat{A}_{n, 2} & \mat{A}_{n, 3} & \dots  & \mat{0}_{d}
        \end{pmatrix}.
    \end{equation*}
    Indeed, \Cref{theorem:one-many-convergence} can be seen as an instance of such a broader characterization for \emph{separable} continuous games. As we point out in \Cref{remark:tensor}, obtaining such results in multilinear settings would require different techniques.
\end{remark}

\begin{remark}[Beyond Polymatrix Games]
    \label{remark:tensor}
    We believe that results in unconstrained multilinear $n$-player games (akin to standard NFGs, but without constraints) can be obtained using recent advances from \emph{multilinear control theory} \citep{Multilinear21:Chen}. In particular, stability could be determined in terms of the tensor of the underlying game.
\end{remark}

Our final result for \Cref{section:continuous} shows that a generic class of first-order fails to guarantee stability in all general-sum games, even with two players. More precisely, we consider the following method.
\begin{equation}
    \label{eq:HGD}
    \tag{HGD}
    \begin{split}
        \vec{x}_i^{(t+1)} = \sum_{\tau=0}^{T} \alpha^{(\tau)} \vec{x}_i^{(t-\tau)} + \sum_{\tau=0}^T \beta^{(\tau)} \nabla_{\vec{x}_i} u_i(\vec{x}^{(t-\tau)}).\!\!\!\!\!\!\!\!\!
    \end{split}
\end{equation}
This method---which will be referred to as \emph{historical gradient descent (HGD)}~\citep{Anagnostides20:Robust}---can be described with the ordered set of coefficients $A \defeq (\alpha^{(0)}, \dots, \alpha^{(T)})$ and $B \defeq (\beta^{(0)},  \dots, \beta^{(T)})$. In particular, \eqref{eq:HGD} contains \eqref{eq:OGD} when $A = (1)$ and $B = (2\eta, -\eta )$, but it goes well-beyond this method. For the purpose of our analysis we are going to represent an \eqref{eq:HGD} algorithm using the following two polynomials.
\begin{equation*}
    \begin{split}
        S(z) \defeq \alpha^{(0)} + \alpha^{(1)} z^{-1} + \dots + \alpha^{(T)} z^{-T}; \\
        G(z) \defeq \beta^{(0)} + \beta^{(1)} z^{-1} + \dots + \beta^{(T)} z^{-T}.
    \end{split}
\end{equation*}
For example, \eqref{eq:OGD} is associated with polynomials $S(z) = 1$ and $G(z) = 2 \eta - \eta z^{-1}$. However, the considered class of dynamics contains certain trivial algorithms. For example, the update rule of an \eqref{eq:HGD} algorithm for which $G(z) \equiv 0$ would not depend on the observed utility at any iteration. To address such trivialities, we impose a condition which ensures that an \eqref{eq:HGD} algorithm is interesting from a game-theoretic standpoint. In particular, we say that an \eqref{eq:HGD} algorithm is \emph{regular} if $S(1) = 1$ and $G(1) \neq 0$. Then, it is easy to see that, if both players employ a regular \eqref{eq:HGD} algorithm and the dynamics converge, the limit points will be equilibria in the sense of \eqref{eq:equilibria-bi}. We will also assume without any loss that the polynomials $G(z)$ and $S(z) - z$ have no common roots. We are now ready to state and prove our main theorem on \eqref{eq:HGD} algorithms.

\begin{restatable}{theorem}{impossibility}
    \label{theorem:impossibility}
    For any regular algorithm in \eqref{eq:HGD} there exists a game for which the method will diverge under any non-trivial initialization.
\end{restatable}
\begin{proof}
    Consider a game $(\mat{A}, \mat{B})$. The dynamics of the underlying dynamical system in the $z$-space can be expressed as follows.
    \begin{equation*}
        \begin{split}
            z \vec{X}(z) &= S(z) \vec{X}(z) + G(z) \mat{A} \Vec{Y}(z); \\
            z \vec{Y}(z) &= S(z) \vec{Y}(z) + G(z) \mat{B}^\top \Vec{X}(z).
        \end{split}
        \iff
        \begin{split}
            (z - S(z)) \vec{X}(z) &= G(z) \mat{A} \Vec{Y}(z); \\
            (z - S(z)) \vec{Y}(z) &= G(z) \mat{B}^\top \Vec{X}(z).
        \end{split}
    \end{equation*}
    From this representation, we may conclude that the characteristic equation of the dynamical system can be expressed as
    \begin{equation*}
        \det \left( (z - S(z))^2 \mat{I}_d - (G(z))^2 \mat{A}^\top \mat{B} \right) = 0 \iff \chi \left( \left( \frac{z - S(z)}{G(z)} \right)^2 \right) = 0,
    \end{equation*}
    where $\chi(z)$ represents the characteristic polynomial of matrix $\mat{A}^\top \mat{B}$. Note that our previous derivation uses the assumption that the polynomials $S(z) - z$ and $G(z)$ do not share any common roots, in turn implying that no root of $G(z)$ can satisfy the equation $\det \left( (z - S(z))^2 \mat{I}_d - (G(z))^2 \mat{A}^\top \mat{B} \right) = 0$. Now consider any $\epsilon > 0$ so that $G(1 + \epsilon) \neq 0$ and $S(1 + \epsilon) \neq 1 + \epsilon$. If the game $(\mat{A}, \mat{B})$ is such that the matrix $\mat{A}^\top \mat{B}$ has the positive eigenvalue
    \begin{equation*}
        \lambda = \left( \frac{1 + \epsilon - S(1 + \epsilon)}{G(1+\epsilon)} \right)^2,
    \end{equation*}
    then it follows that $z = 1 + \epsilon$ is a solution to the characteristic equation of the system. But given that $|z| > 1$, this necessarily implies that the dynamics will diverge---under any non-trivial initialization---by the fundamental theorem of linear dynamical systems.
\end{proof}

While there is a plethora of control-theoretic techniques that could stabilize the dynamics, we have argued (\Cref{proposition:inefficiency}) that stability might be at odds with efficiency---social welfare maximization. Understanding the fundamental tension between different solution concepts is an important question for the future.

\section{Additional Experiments}
\label{appendix:experiments}

In this section we corroborate some our theoretical results through experiments on several games. 

\subsection{Normal-Form Games}

We start by illustrating the last-iterate convergence when players employ different and potentially advanced prediction mechanisms in normal-form games. We are first going to assume that both players employ \eqref{eq:OMD} with prediction $\Vec{m}^{(t)}$ set to either $H$-step recency bias (\Cref{item:H-step}), or geometrically $\delta$-discounted recency bias (\Cref{item:discounting}). We remark that while \citet{Syrgkanis15:Fast} established \rvu bounds only for OFTRL under such predictions, it is immediate to extend these bounds with qualitatively similar results for OMD. That is, the bounds on the learning rate obtained in \Cref{proposition:advanced} for OFTRL are analogous to the corresponding ones for OMD. Inspired by the work of \citet{Daskalakis21:Near}, we will also experiment with predictions inducing $H$-order differences in the \rvu bound:
\begin{itemize}
    \item[(i)] $1$-order: $\Vec{m}^{(t)} \defeq \Vec{u}^{(t-1)}$;
    \item[(ii)] $2$-order: $\Vec{m}^{(t)} \defeq 2 \Vec{u}^{(t-1)} - \Vec{u}^{(t-2)}$;
    \item[(iii)] $3$-order: $\Vec{m}^{(t)} \defeq 3 \Vec{u}^{(t-1)} - 3 \Vec{u}^{(t-2)} + \Vec{u}^{(t-3)}$;
    \item[(iv)] $4$-order: $\Vec{m}^{(t)} \defeq 4 \Vec{u}^{(t-1)} - 6 \Vec{u}^{(t-2)} + 4 \Vec{u}^{(t-3)} - \Vec{u}^{(t-4)}$.
\end{itemize}

We point out that our techniques immediate imply last-iterate convergence under $H$-order predictions, for a sufficiently small learning rate $\eta = \eta(H)$.

\paragraph{Zero-Sum Games} We first illustrate the behavior of the dynamics for two-player zero-sum games. We let $\mat{A}_j$ represent the cost matrix for player $\cX$---and subsequently the payoff matrix for player $\cY$---for $j \in \{1, 2, 3\}$, defined as follows.

\begin{equation}
    \label{eq:zero_sum}
    \mat{A}_1 \defeq \begin{bmatrix}
        1    & -1 & -1 \\
        -1   & -1 & 0  \\
        -0.5 & 0  & -1
    \end{bmatrix};
    \mat{A}_2 \defeq \begin{bmatrix}
        1    & -2 & -1 \\
        -1   & 1  & 0  \\
        -0.5 & 1  & -1
    \end{bmatrix};
    \mat{A}_3 \defeq \begin{bmatrix}
        -1  & 1    & -1   \\
        0   & 0.5  & -1   \\
        0.3 & -0.5 & -0.5
    \end{bmatrix}.
\end{equation}

The \eqref{eq:OMD} dynamics when \emph{both} players employ $H$-step recency bias are illustrated and discussed in \Cref{fig:H_step-both}. \Cref{fig:H_step-one} illustrates the behavior of the dynamics when only \emph{one} of the two players employs $H$-step recency bias. The geometrically $\delta$-discounting prediction mechanism is depicted in \Cref{fig:discounting}, while \Cref{fig:H_order} shows the dynamics under $H$-order predictions.

\begin{figure}[!ht]
    \centering
    \includegraphics[scale=0.7]{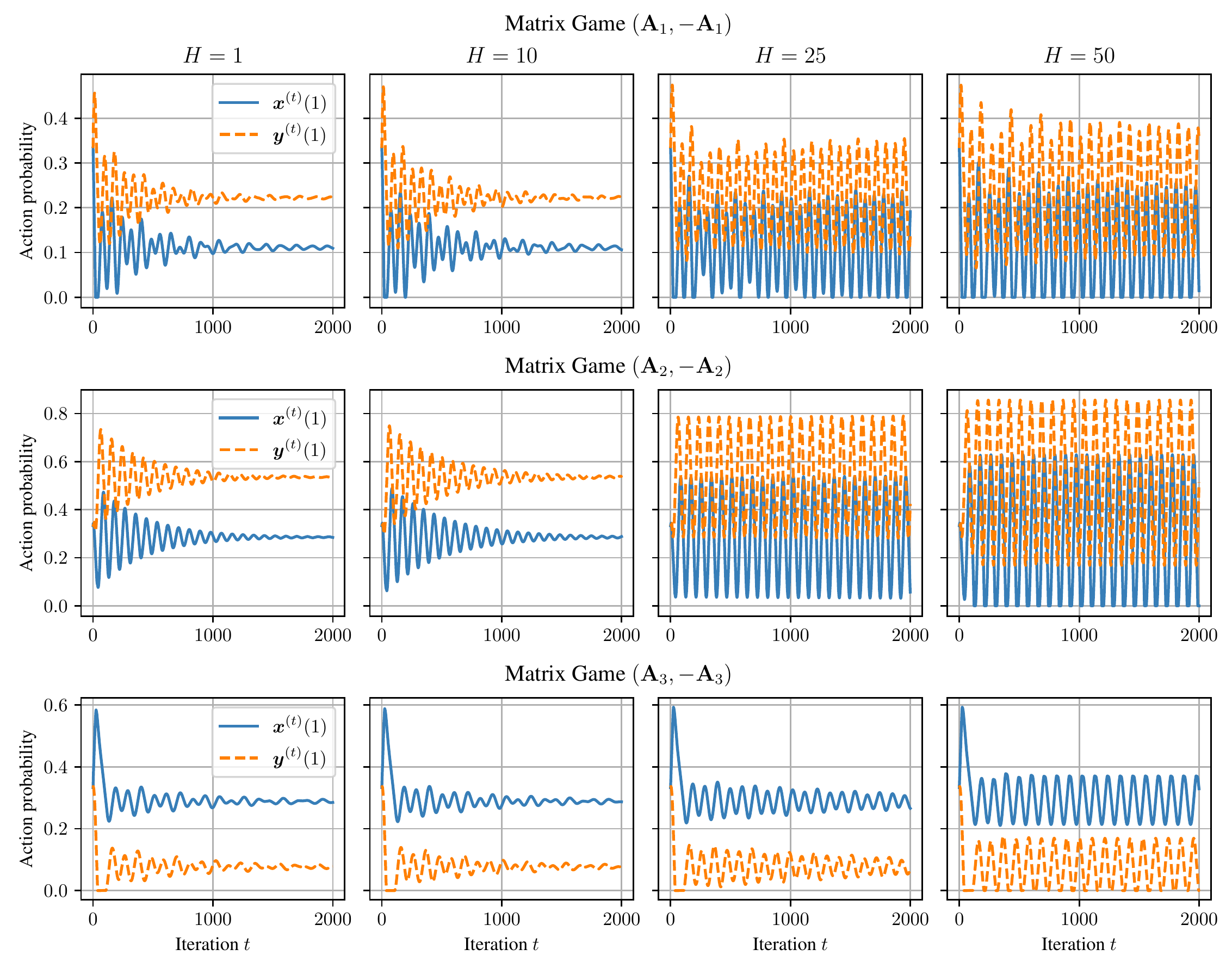}
    \caption{The \eqref{eq:OMD} dynamics in the zero-sum games described in \eqref{eq:zero_sum} for $2000$ iterations. The $y$-axis corresponds to the probability with which each player selects the first action $1$. Both players use Euclidean regularization with $\eta = 0.05$ and $H$-step recency bias (\Cref{item:H-step}) for different values of $H \in \{1, 10, 25, 50\}$. The dynamics under small values of the prediction window $H$ are qualitatively similar. On the other hand, as suggested by \Cref{proposition:advanced}, larger values of $H$ can lead to instability, confirming that the learning rate has to be adapted to $H$.}
    \label{fig:H_step-both}
\end{figure}

\begin{figure}[!ht]
    \centering
    \includegraphics[scale=0.7]{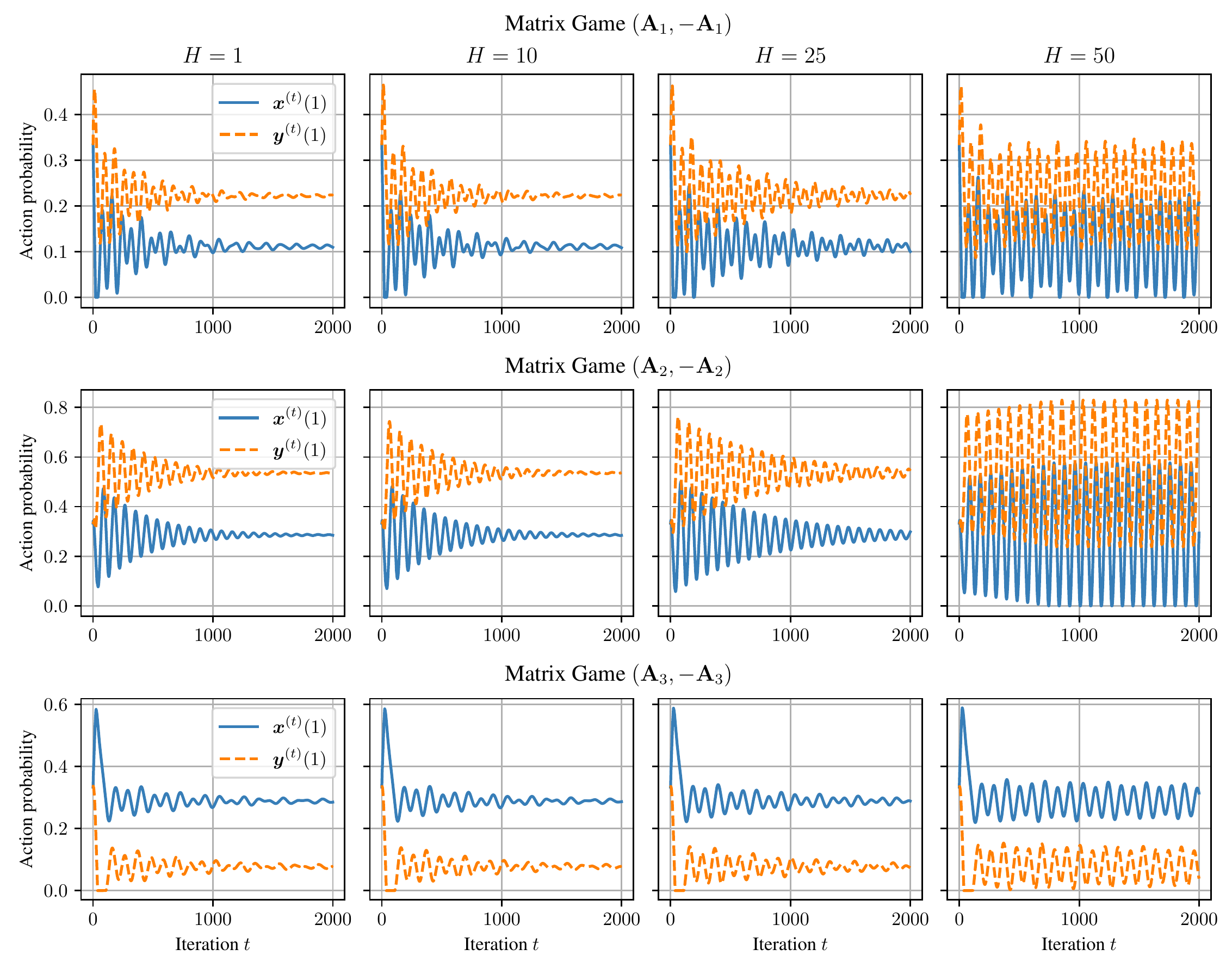}
    \caption{The \eqref{eq:OMD} dynamics in the zero-sum games described in \eqref{eq:zero_sum} for $2000$ iterations. The $y$-axis corresponds to the probability with which each player selects the first action $1$. Both players use Euclidean regularization with $\eta = 0.05$. Unlike \Cref{fig:H_step-both}, player $\cX$ uses one-step recency bias, while player $\cY$ continues to be using $H$-step recency bias for differnt values of $H \in \{1, 10, 25, 50\}$. We observe that larger values of $H$ could lead to unstable behavior if $\eta$ is not selected sufficiently small.}
    \label{fig:H_step-one}
\end{figure}

\begin{figure}[!ht]
    \centering
    \includegraphics[scale=0.7]{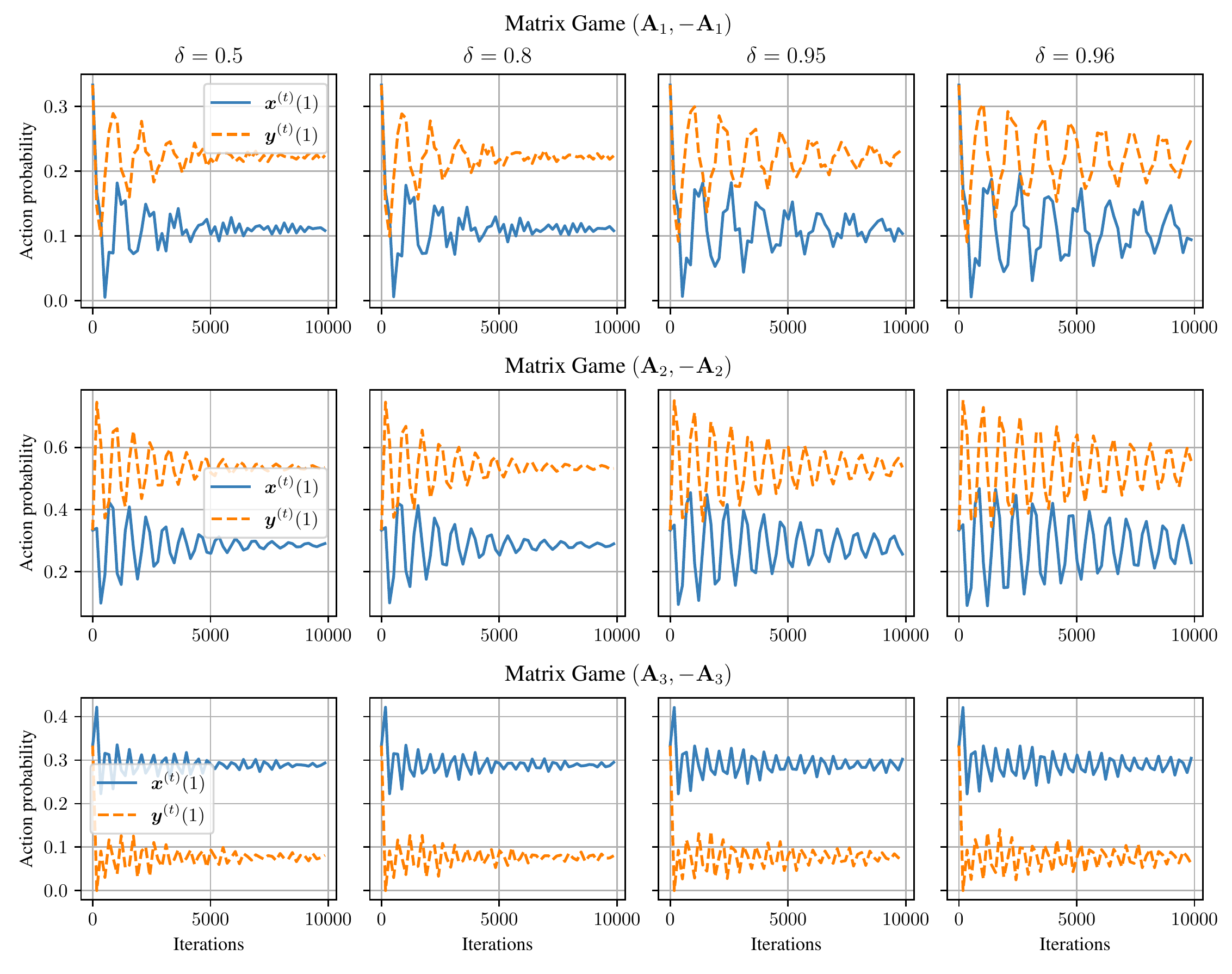}
    \caption{The \eqref{eq:OMD} dynamics in the zero-sum games described in \eqref{eq:zero_sum} for $10000$ iterations. The $y$-axis corresponds to the probability with which each player selects the first action $1$. Both players use Euclidean regularization with $\eta = 0.02$ and geometrically $\delta$-discounted predictions (\Cref{item:discounting}) for different values of $\delta \in \{0.5, 0.8, 0.95, 0.96\}$. We observe that while $\delta$ approaches to $1$ the dynamics become more oscillatory.}
    \label{fig:discounting}
\end{figure}

\begin{figure}[!ht]
    \centering
    \includegraphics[scale=0.7]{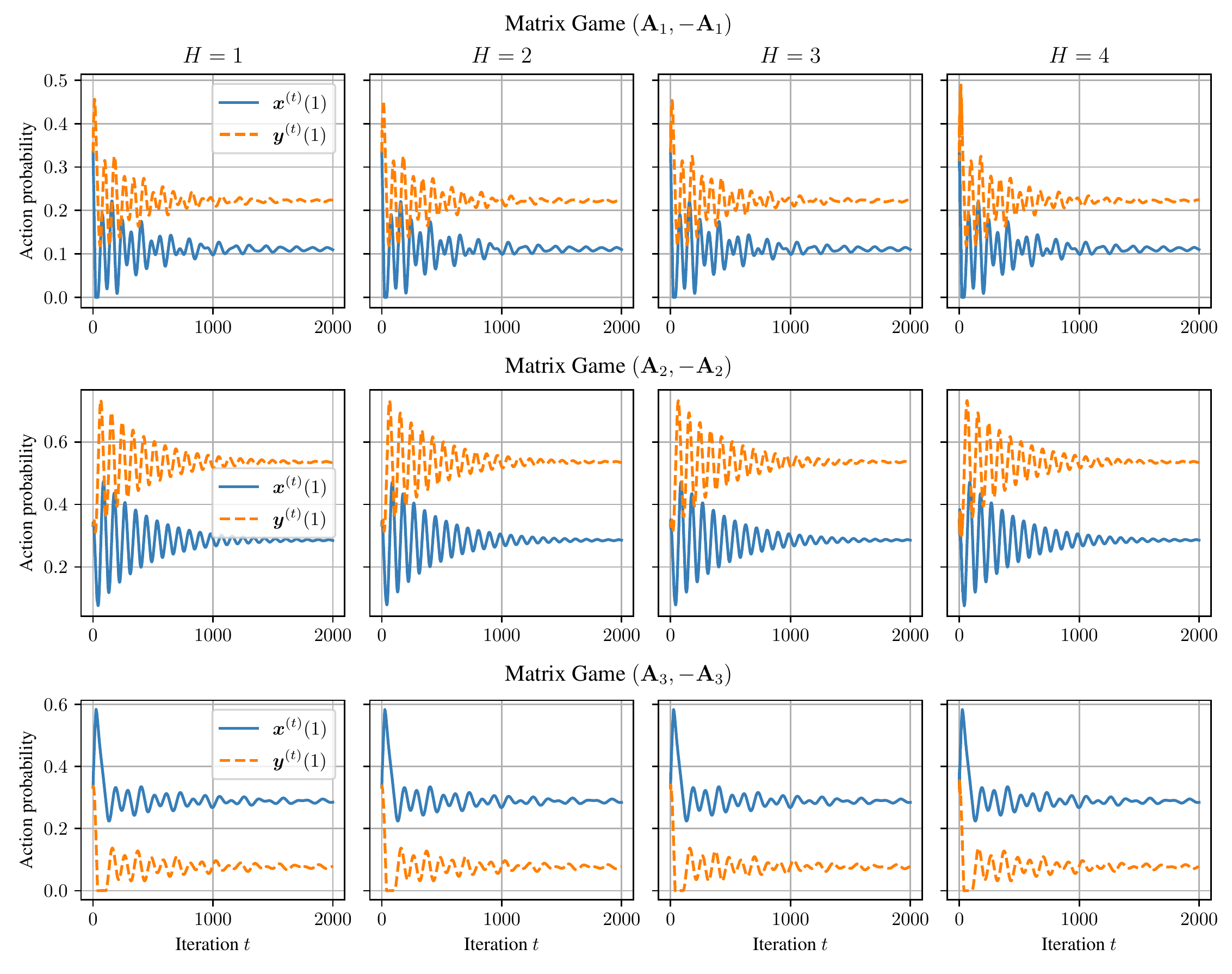}
    \caption{The \eqref{eq:OMD} dynamics in the zero-sum games described in \eqref{eq:zero_sum} for $2000$ iterations. The $y$-axis corresponds to the probability with which each player selects the first action $1$. Both players use Euclidean regularization with $\eta = 0.05$ and $H$-order predictions for different values of $H \in \{1, 2, 3, 4\}$. Interestingly, the dynamics exhibit almost identical behavior to that with one recency bias---equivalently, $1$-order predictions.}
    \label{fig:H_order}
\end{figure}

\paragraph{Strategically Zero-Sum Games} Next, we experiment with strategically zero-sum games (\Cref{def:strat-zero_sum}). The cost matrix for player $\cX$ in each game is still given by \eqref{eq:zero_sum}, but now we assume that the cost matrix of player $\cY$ are given as follows.

\begin{equation}
    \label{eq:strat-zero_sum}
    \mat{B}_1 \defeq \begin{bmatrix}
        -1    & 0.5 & 1  \\
        0     & .5  & .5 \\
        -0.25 & 0   & 1
    \end{bmatrix};
    \mat{B}_2 \defeq \begin{bmatrix}
        0.3   & 0    & 0.3  \\
        -0.2  & .25  & 0.3  \\
        -0.35 & 0.75 & 0.05
    \end{bmatrix};
    \mat{B}_3 \defeq \begin{bmatrix}
        0.7 & 0.5 & 0.56 \\
        0.4 & 0.4 & 0.7  \\
        0.5 & 0.6 & 0.5
    \end{bmatrix}.
\end{equation}

To verify that the games $(\mat{A}_1, \mat{B}_1)$, $(\mat{A}_2, \mat{B}_2)$, and $(\mat{A}_3, \mat{B}_3)$ are strategically zero-sum, observe that
\begin{align*}
    \mat{B}_1 & = - 0.5 \mat{A}_1 + \Vec{1}_3 \begin{bmatrix}
        -0.5 & 0 & 0.5
    \end{bmatrix};
    \\
    \mat{B}_2 & = - 0.5 \mat{A}_2 + \Vec{1}_3 \begin{bmatrix}
        -0.2 & 0.5 & -0.2
    \end{bmatrix}; \\
    \mat{B}_3 & = - 0.2 \mat{A}_3 + \Vec{1}_3 \begin{bmatrix}
        0.5 & 0.5 & 0.5
    \end{bmatrix}.
\end{align*}
where recall that $\Vec{1}_3$ is the all-ones vector in $\R^3$. Thus, the fact that these games are strategically zero-sum follows from the converse of \Cref{theorem:strat-zero_sum}. In particular, the game $(\mat{A}_3, \mat{B}_3)$ is strictly competitive (see \Cref{remark:strictly_competitive}). The \eqref{eq:OMD} dynamics in these games are illustrated in \Cref{fig:strat}.

\begin{figure}[!ht]
    \centering
    \includegraphics[scale=0.7]{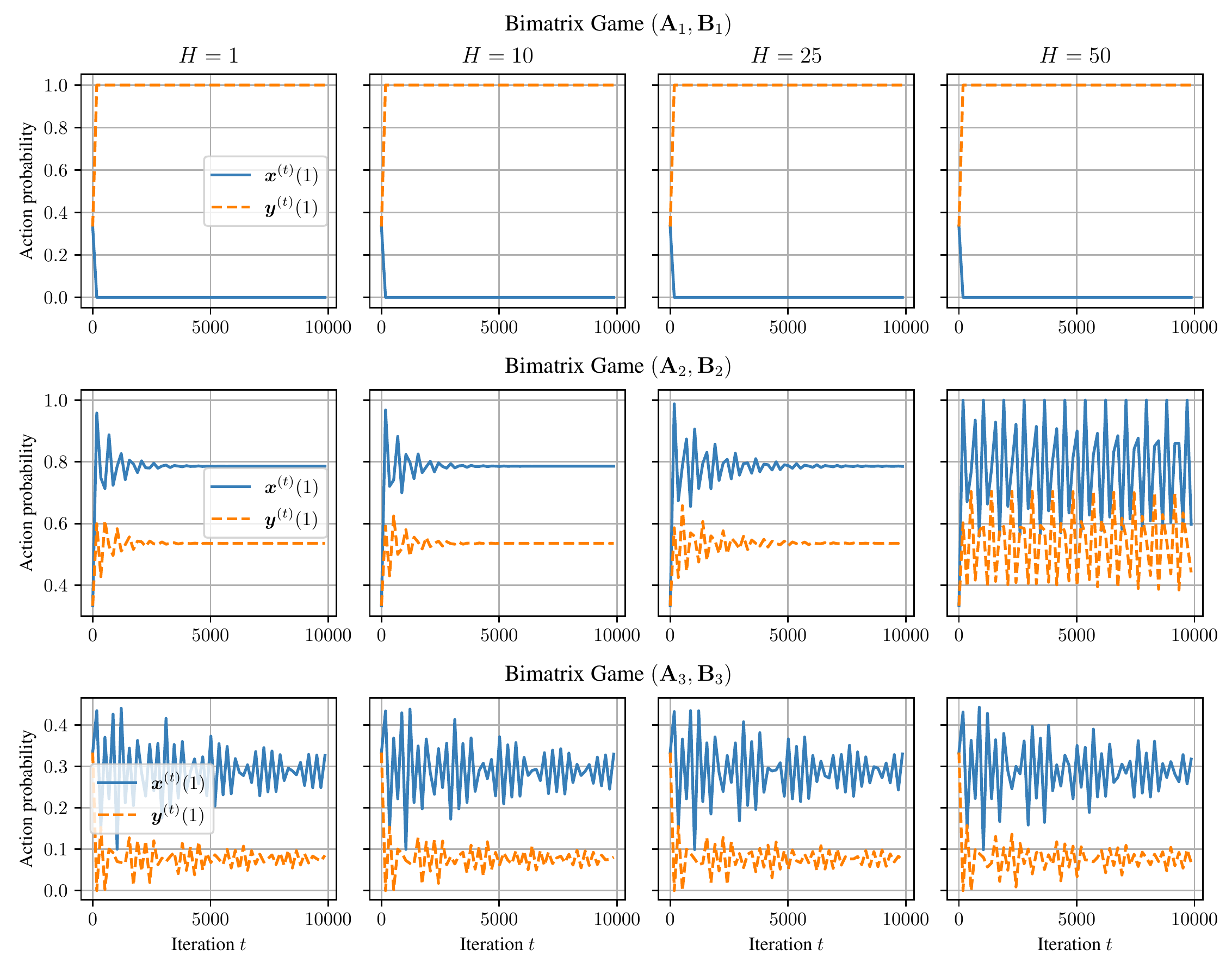}
    \caption{The \eqref{eq:OMD} dynamics in the strategically zero-sum cost-minimization games described in \eqref{eq:zero_sum} and \eqref{eq:strat-zero_sum} for $10000$ iterations. The $y$-axis corresponds to the probability with which each player selects the first action $1$. Both players use Euclidean regularization with $\eta = 0.05$ and $H$-step recency bias for different values of $H \in \{1, 10, 25, 50\}$. The dynamics under small values of the prediction window $H$ are qualitatively similar. In contrast, observe that larger values of $H$ can lead to instability.}
    \label{fig:strat}
\end{figure}

We also experiment with different players using different regularization. In particular, we assume that player $\cX$ uses the Euclidean DGF, while player $\cY$ uses the negative entropy DGF. While our last-iterate guarantees do not appear to apply for this case due to the lack of smoothness of the entropic regularizer (close to the boundary), \Cref{fig:diff} illustrates that the dynamics still converge.

\begin{figure}[!ht]
    \centering
    \includegraphics[scale=0.7]{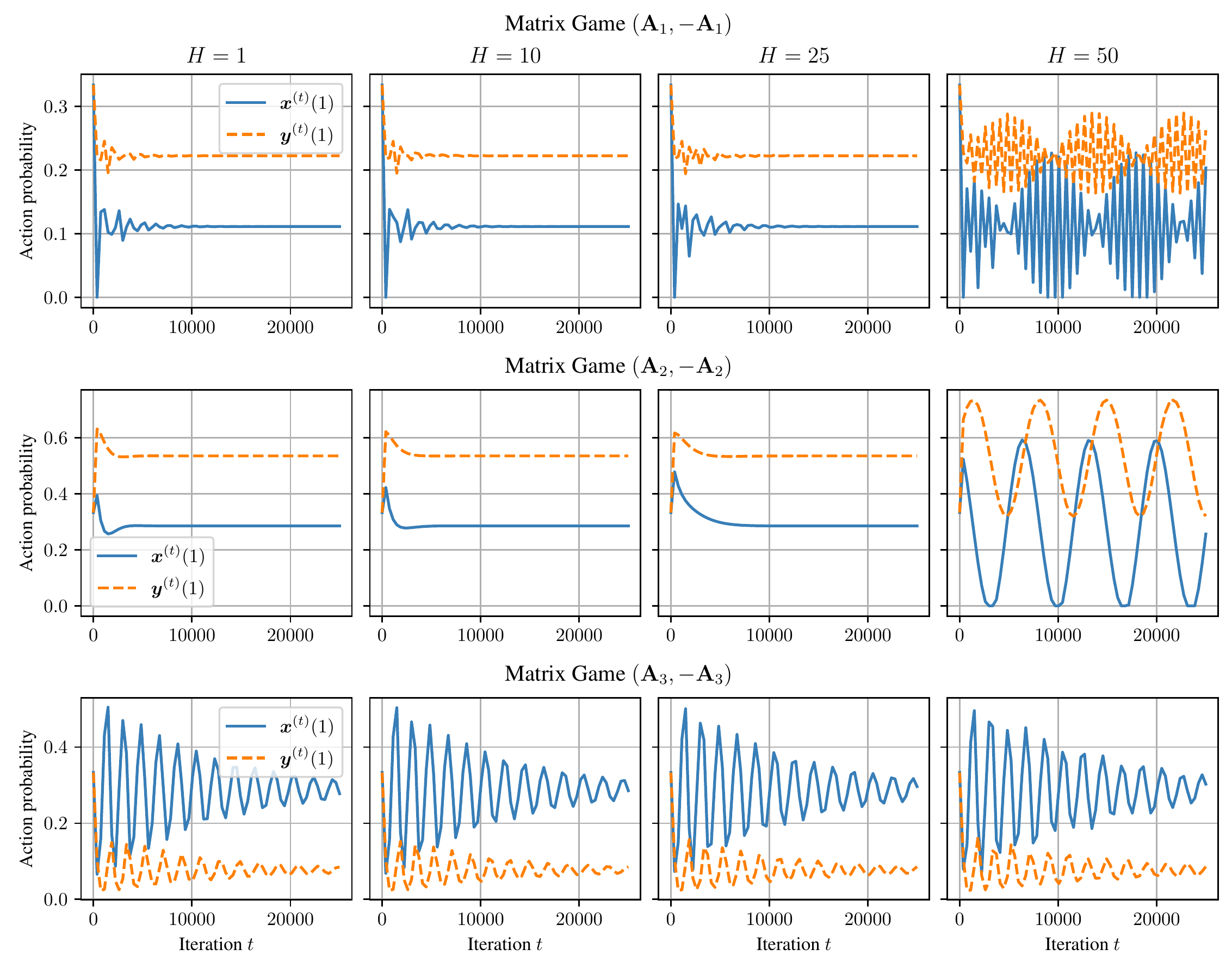}
    \caption{The \eqref{eq:OMD} dynamics in the zero-sum games described in \eqref{eq:zero_sum}. The $y$-axis corresponds to the probability with which each player selects the first action $1$. Player $\cX$ uses the Euclidean DGF, while player $\cY$ uses the negative entropy DGF, both with $\eta = 0.05$ and $H$-order predictions. We observe that the dynamics exhibit convergent behavior.}
    \label{fig:diff}
\end{figure}

\subsection{Continuous Games}

In this subsection with provide experiments on continuous unconstrained games. In particular, we illustrate the behavior of the dynamics for the games constructed for \Cref{proposition:inefficiency,proposition:robustness} in \Cref{fig:efficiency,fig:robustness} respectively.

\begin{figure}[!ht]
    \centering
    \includegraphics[scale=.7]{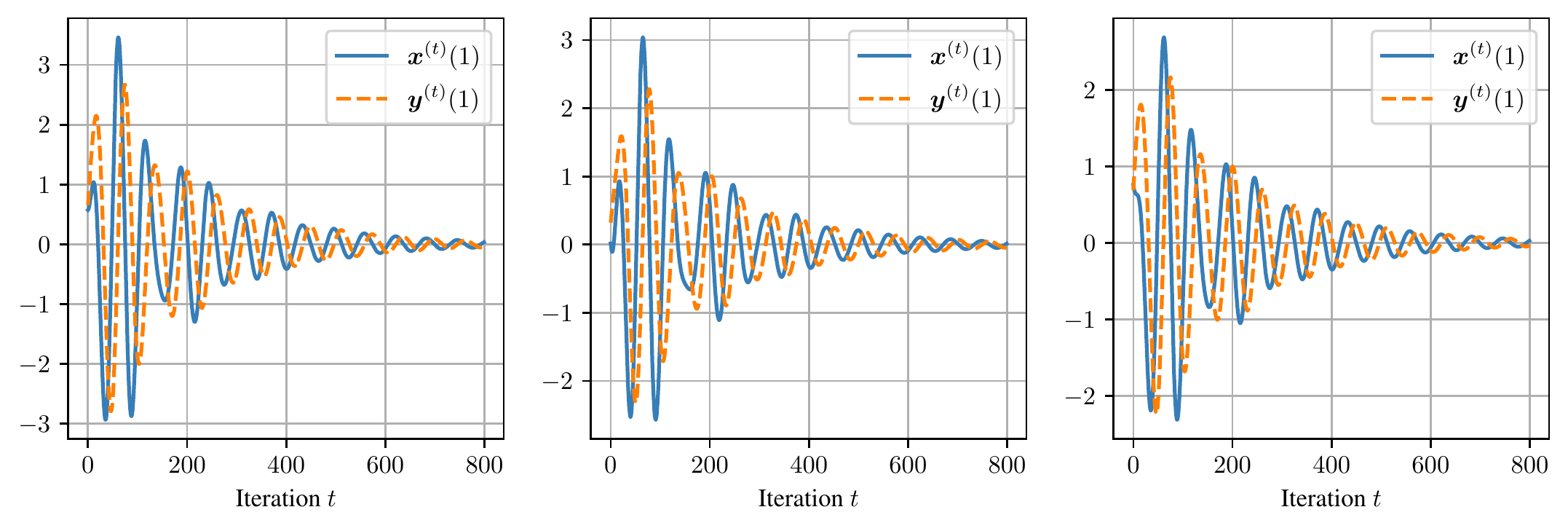}
    \caption{The \eqref{eq:OGD} dynamics for the continuous game described in \eqref{eq:game-inefficiency} for $800$ iterations. The $y$-axis illustrates the first coordinate of $\Vec{x} \in \R^2$ and $\Vec{y} \in \R^2$ respectively. Different columns correspond to different random initializations. As predicted by \Cref{theorem:two_player-convergence}, and subsequently \Cref{proposition:inefficiency}, the dynamics converge to the point $(\Vec{0}, \Vec{0})$, which leads to each player receiving $0$ utility. }
    \label{fig:efficiency}
\end{figure}

\begin{figure}[!ht]
    \centering
    \includegraphics[scale=0.7]{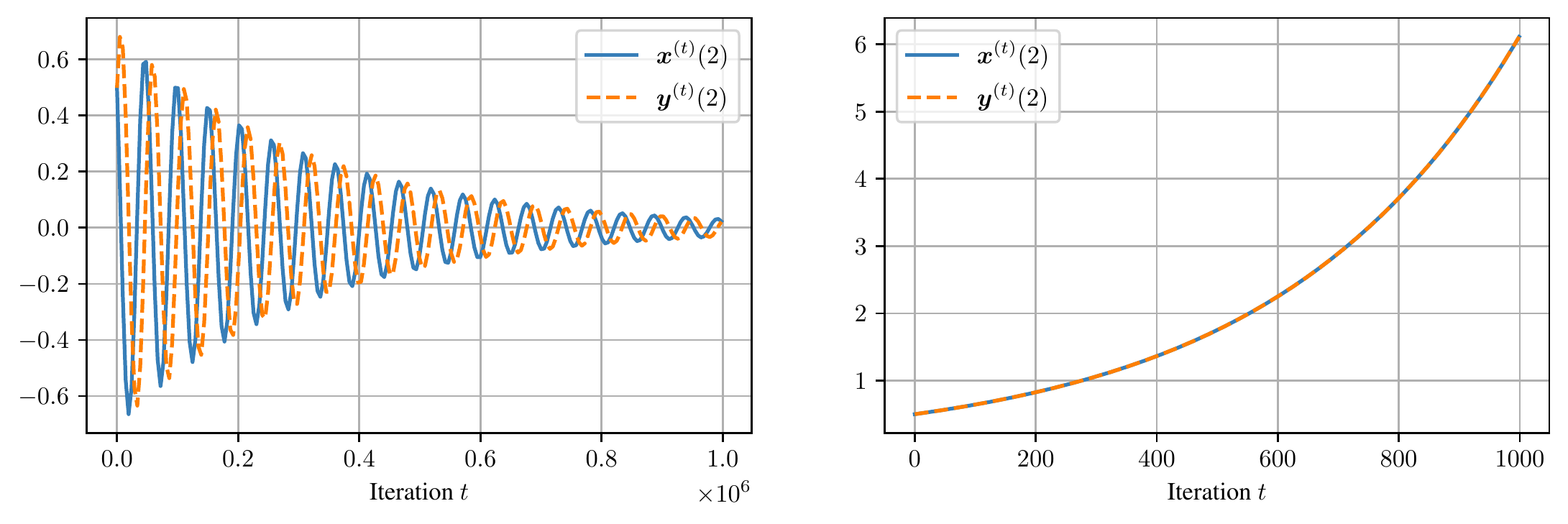}
    \caption{The \eqref{eq:OGD} dynamics for the continuous game described in \eqref{eq:robustness} for $\epsilon = 0.05$. In particular, the left image corresponds to the game $(\mat{A}, -\mat{A})$ for $10^6$ iterations, while the right one to the game $(\mat{A}, \mat{B})$ for $1000$ iterations. Although $\| \mat{A} + \mat{B} \|_F = \epsilon$, the two systems exhibit completely different behavior. This appears to be related to the fact that although the dynamics in the game $(\mat{A}, -\mat{A})$ converge linearly, the rate of convergence is very close to $1$. The same phenomenon occurs for any $\epsilon > 0$, as predicted by \Cref{proposition:robustness}.}
    \label{fig:robustness}
\end{figure}

\end{document}